\documentclass[11pt]{article}
\usepackage[letterpaper,hmargin=1in,vmargin=1.25in]{geometry}

\usepackage{url}

\usepackage{amsmath}
\usepackage{amssymb}

\usepackage{amsthm}

\theoremstyle{plain}
\newtheorem{theorem}{Theorem}[section]

\newtheorem{corollary}[theorem]{Corollary}
\newtheorem{proposition}[theorem]{Proposition}

\theoremstyle{definition}
\newtheorem{definition}[theorem]{Definition}

\newtheorem{postulate}{Postulate}
\newtheorem{remark}[theorem]{Remark}

\newcommand{\abs}[1]{\left\lvert#1\right\rvert}

\newcommand{\rest}[2]{#1\!\!\restriction_{#2}}

\newcommand{\osg}[1]{\left[#1\right]^{\prec}}

\newcommand{\N}{\mathbb{N}}%
\newcommand{\X}{\{0,1\}^*}%

\newcommand{\Bm}[2]{\lambda_{#1}\left(#2\right)}

\newcommand{\ket}[1]{| #1 \rangle}
\newcommand{\bra}[1]{\langle #1 |}
\newcommand{\braket}[2]{\langle #1 | #2 \rangle}
\newcommand{\product}[2]{\lvert #1\rangle\langle #2\rvert}

\newcommand{\PS}{\mathbb{P}}%

\newcommand{\ssoa}{\overline{\mathcal{H}}}

\newcommand{\noi}{\noindent}

\usepackage{hyperref}

\title{\textbf{An analysis of Wigner's friend
in the framework
of quantum mechanics
based on the principle of typicality}}

\author{Kohtaro Tadaki\\
\\
Department of Computer Science, College of Engineering, Chubu University\\
1200 Matsumoto-cho, Kasugai-shi, Aichi 487-8501, Japan\\
E-mail: \textsf{tadaki@fsc.chubu.ac.jp}\\
\url{https://tadaki.org/}}

\date{
\begin{quotation}
\noi\textbf{Abstract.}
The notion of probability plays a crucial role in quantum mechanics.
It appears in quantum mechanics as the \emph{Born rule}.
In modern mathematics which describes quantum mechanics, however, probability theory means nothing other than measure theory,
and therefore any operational characterization of the notion of probability is still missing in quantum mechanics.
In our former works~[K.~Tadaki, arXiv:1804.10174],
based on the toolkit of \emph{algorithmic randomness},
we presented
a refinement of the Born rule,
called the \emph{principle of typicality},
for specifying the property of results of measurements \emph{in an operational way}. \\
The \emph{Wigner's friend paradox} is a Gedankenexperiment regarding when and where the reduction of the state vector occurs in a chain of the measurements by several observers where the state of the consciousness of each observer is measured by the subsequent observer.
In this paper, we extend the framework of the principle of typicality
so that it can be applicable to
situations
where apparatuses perform measurements over other apparatuses.
We then make an analysis of the Wigner's friend paradox
within this extended framework of quantum mechanics based on the principle of typicality.
We
draw common sense conclusions about
it.
\emph{Deutsch's thought experiment} is a variant of the Wigner's friend paradox,
which can, in principle, verify the effect of the consciousness of observer on the reduction of the state vector.
We make an analysis of
it
comprehensively within the extended framework.
We then make a prediction which is testable in principle. \\
In our extended framework, we can analyze still more complicated situations.
As such an example, we introduce a combination of
the above two,
named the \emph{Wigner-Deutsch collaboration},
and perform a thorough analysis of it.
\end{quotation}
\begin{quotation}
\noi\textit{Key words\/}:
Wigner's friend,
Deutsch's thought experiment,
Schr\"{o}dinger's cat,
many-worlds interpretation,
quantum measurement,
Born rule,
reduction of the state vector,
operational characterization of probability,
the principle of typicality,
algorithmic randomness,
Martin-L\"of randomness
\end{quotation}
}

\begin{document}

\maketitle

\tableofcontents

\section{Introduction}

The notion of probability plays a crucial role in quantum mechanics.
It appears in quantum mechanics as the so-called \emph{Born rule}, i.e.,
\emph{the probability interpretation of the wave function} \cite{D58,vN55,NC00}.
In modern mathematics which describes quantum mechanics, however,
probability theory means nothing other than \emph{measure theory}~\cite{K50},
and therefore any \emph{operational characterization of the notion of probability} is still missing
in quantum mechanics.
In this sense, the current form of quantum mechanics is considered to be \emph{imperfect}
as a physical theory which must stand on operational means.

\subsection{Historical background of this work}

The current form of
quantum mechanics was developed from the beginning of the last century
and completed by 1930~\cite{D58,vN55}.
Independently of the development of quantum mechanics, in the last century
there was a comprehensive attempt to provide
an operational characterization of the notion of probability.
Namely, von Mises
developed
a \emph{mathematical theory of repetitive events} which was aimed at reformulating
the theory of probability and statistics
based on an operational characterization of the notion of probability \cite{vM57,vM64}.
In a series of his comprehensive works which began in 1919,
von Mises developed this theory and, in particular,
introduced the notion of \emph{collective} as a mathematical idealization of a long sequence of
outcomes of experiments or observations repeated under a set of invariable conditions,
such as the repeated tossings of a coin or of a pair of dice.
The collective plays a role as an operational characterization of the notion of probability,
and is an infinite sequence of sample points in the sample space of a probability space.
The notion of  collective is defined
by means of
the notion of \emph{place selection},
which is a function from the set of finite strings
of the sample points
to the set~$\{\mathrm{YES}, \mathrm{NO}\}$.
Given an infinite sequence of the sample points,
the place selection is used for selecting positions of the infinite sequence in a certain manner
to form a subsequence of it.
Then, for aiming to introduce a randomness property to the collective,
von Mises defined a collective
to be an infinite sequence of the sample points such that
the law of large numbers with a fixed limit value holds for all its subsequences selected by place selections.
In this definition of the collective, if all possible place selections are allowed,
the definition becomes empty.
Thus, the class of place selections used in the definition is assumed to
a
countable set.

In 1939, however, Ville~\cite{Vi39} revealed
one of the inescapable defects
of the notion of collective.
Namely, he showed that,
no matter how
a
countable collection of place selections is chosen,
there exists an infinite sequence of the sample points
which must be classified
into
a collective according to the definition above,
and which
cannot be regarded as random according to our intuition.
Apart from Ville's result,
the collective has,
in the first place,
an \emph{intrinsic defect} that
\emph{it cannot exclude the possibility that an event with probability zero may occur}.
For the
historical
development of the theory of collectives from the point of view of the definition of randomness,
see Downey and Hirschfeldt~\cite{DH10}.

In 1966, Martin-L\"of \cite{M66} introduced the definition of random sequences,
which is called \emph{Martin-L\"of randomness} nowadays, and
plays a central role in the recent development of algorithmic randomness.
At the same time,
he introduced the notion of \emph{Martin-L\"of randomness with respect to a Bernoulli measure} \cite{M66}.
He then pointed out that this notion overcomes the defect of the collective in the sense of von Mises, and
this can be regarded precisely as the collective which von Mises wanted to define.
However, he did not develop probability theory
based on Martin-L\"of random
sequences
with respect to a Bernoulli measure.

\emph{Algorithmic randomness}, also known as \emph{algorithmic information theory},
is a field of mathematics which studies the definitions of random sequences and
their property~\cite{Solom64,Kol65,C66,M66,Sch73,C75,C87b,N09,DH10}.
However, the recent research on algorithmic randomness would seem only interested in
the notions of randomness themselves and their
interrelation,
and not seem to have
made an attempt
to develop probability theory based on Martin-L\"of randomness with respect to a Bernoulli measure in an operational manner so far.

\subsection{Operational characterization of the notion of probability by algorithmic randomness}

In our former work~\cite{T14,T15,T15CCR,T16arXiv,T17ALC} we made such an attempt.
Namely, in the work
we developed a theory of
\emph{operational characterization of the notion of probability by algorithmic randomness}
for general discrete probability spaces,
using the notion of \emph{Martin-L\"of randomness with respect to a Bernoulli measure} \cite{M66} (and its extension over the Baire space).
Here, that infinite sequence of sample points in the sample space of a probability space
which is Martin-L\"of random with respect to the Bernoulli measure
induced by the probability space
plays a role as an operational characterization of the notion of probability.
In the work
we gave natural and equivalent operational characterizations of the basic notions of probability theory,
such as the notions of conditional probability and the independence of events/random variables,
in terms of
the notion of Martin-L\"of randomness with respect to a Bernoulli measure.
We then made
applications of this
framework
to information theory and cryptography, as examples of the applications,
in order to demonstrate the wide applicability of
our
framework
to the general areas of
science and
technology.

The crucial point is that
in our theory
the underlying discrete probability space is \emph{quite arbitrary}, and
thus
\emph{we do not
need to
impose any computability restrictions on the underlying
discrete
probability space
at all},
in particular.
See Tadaki~\cite{T16arXiv}
for the details of our theory of operational characterization of the notion of probability
in the case of
general finite probability spaces
and its applications,
and Tadaki~\cite{T17ALC}
for ones
in the case of general discrete probability spaces
whose sample spaces are countably infinite.

\subsection{Refinement of quantum mechanics by algorithmic randomness}

As a major application of our framework~\cite{T14,T15,T15CCR,T16arXiv} above to basic science,
in the work~\cite{T14CCR,T15QCOMPINFO,T15Kokyuroku,T16CCR,T16QIP,T16QIT35,T17SCIS,T18arXiv}
we presented an \emph{operational refinement} of the Born rule for pure states, as an alternative rule to it,
based on our theory of operational characterization of the notion of probability
by algorithmic randomness,
for the purpose of making quantum mechanics \emph{operationally perfect}.
Namely, we used the notion of
\emph{Martin-L\"of randomness with respect to a Bernoulli measure}
to present
the operational refinement of the Born rule for pure states,
for specifying the property of the results of
quantum measurements \emph{in an operational way}.
We then presented
an operational refinement of the Born rule for mixed states, as an alternative rule to it,
based on Martin-L\"of randomness with respect to a Bernoulli measure.
In particular, we gave a
\emph{precise definition}
for the notion of \emph{mixed state}.
We then showed that all of
the refined rules of the Born rule for both pure states and mixed states
can be derived
from a \emph{single} postulate, called the \emph{principle of typicality}, in a unified manner.
We did this from the point of view of the \emph{many-worlds interpretation of quantum mechanics}~\cite{E57}.
We then
made an application of our framework to the BB84 quantum key distribution protocol
in order to demonstrate how properly our framework works in
\emph{practical}
problems
in quantum mechanics, based on the principle of typicality.
See the work~\cite{T18arXiv} for the details of our framework of quantum mechanics based on the principle of typicality.

In the subsequent works~\cite{T20Kokyuroku,T24CCR}, we made an application of our framework based on the principle of typicality
to the \emph{argument of local realism versus quantum mechanics}
for refining it,
in order to \emph{further} demonstrate how properly our framework works
in \emph{practical} problems in quantum mechanics.

First, in the work~\cite{T20Kokyuroku} we \emph{refined} and \emph{reformulated}
the argument of \emph{Bell's inequality versus quantum mechanics} \cite{Bell64},
by means of algorithmic randomness.
On the one hand, we refined and reformulated the \emph{assumptions of local realism}~\cite{EPR35}
to lead to Bell's inequality~\cite{Bell64} (in particular, to its variant known as the \emph{CHSH inequality}~\cite{CHSH69}),
in terms of the theory of operational characterization of the notion of probability
by algorithmic randomness,
developed by the work~\cite{T14,T15,T15CCR,T16arXiv} as mentioned above.
On the other hand, we refined and reformulated
the corresponding argument of quantum mechanics to violate Bell's inequality,
based on the principle of typicality.

Second, in the work~\cite{T24CCR} we \emph{refined} and \emph{reformulated}
an entirely different type of the variant of
the argument of Bell's inequality versus quantum mechanics, i.e.,
the argument over the \emph{GHZ experiment} \cite{GHZ89,GHSZ90,Mer90},
in the same manner as above.
In the case of the GHZ experiment,
we made the `two-sided' refinements regarding both
the \emph{predictability} (by quantum mechanics) and
\emph{unexplainability} (by local realism)
of the \emph{perfect correlations of measurement results over three parties}
for a system of three spin-$1/2$ particles, in terms of algorithmic randomness.

On the other hand,
in the work~\cite{T25CCR},
we made an application of our framework
based on the principle of typicality
to the \emph{theory of quantum error-correction}
for refining it,
in order to \emph{still further} demonstrate how properly our framework works
in \emph{practical} problems in quantum mechanics.

The general theory of quantum error-correcting codes is
introduced and developed by
Ekert and Macchiavello~\cite{EM96}, Bennett, et al.~\cite{BDSW96},
and Knill and Laflamme~\cite{KL97}.
In the general theory,
the \emph{quantum error-correction conditions} give necessary and sufficient
conditions for quantum error-correction to be possible.
To show the sufficiency of the quantum error-correction conditions involves
the construction of an error-correction operation.
In the work~\cite{T25CCR},
we \emph{refined} and \emph{reformulated}
the whole process consisting of the noise process and the error-correction operation
in terms of our framework of quantum mechanics based on the principle of typicality,
in order to \emph{operationally} show that quantum error-correction is possible
if the quantum error-correction conditions hold.

The \emph{Calderbank-Shor-Steane codes} (\emph{CSS codes}, for short)
is a large class of quantum error-correcting codes, introduced and developed by
Calderbank and Shor~\cite{CaSh96} and Steane~\cite{Ste96}.
In the work~\cite{T25CCR},
we \emph{refined} and \emph{reformulated}
the whole process consisting of the noise process and the error-correction operation
for CSS codes,
based on the principle of typicality,
in order to show in an \emph{operational} manner that
quantum error-correction by CSS codes is possible.

Algorithmic randomness
is a field of mathematics which enables us to consider
the randomness of an \emph{individual} (and classical) infinite sequence.
It originated in the groundbreaking works of Solomonoff~\cite{Solom64}, Kolmogorov~\cite{Kol65}, and Chaitin~\cite{C66} in the mid-1960s.
They independently introduced the notion of \emph{program-size complexity}, also known as \emph{Kolmogorov complexity},
in order to quantify the randomness of an individual (and classical) \emph{finite} object.
Around the same time, Martin-L\"of \cite{M66} introduced a measure theoretic approach to characterize
the randomness of an individual (and classical) \emph{infinite} binary sequence.
His
approach, called Martin-L\"of randomness nowadays,
is one of the major notions in algorithmic randomness,
as well as
the notion of
program-size complexity.
Later on, in the 1970s
Schnorr~\cite{Sch73} and Chaitin~\cite{C75} showed that Martin-L\"of randomness is equivalent to the randomness defined by program-size complexity
in characterizing
random infinite binary sequences.
See
Nies~\cite{N09} and Downey and Hirschfeldt~\cite{DH10}
for the recent developments as well as the historical details of algorithmic randomness.
In this paper, we use the notion of \emph{Martin-L\"of randomness with respect to a Bernoulli measure}~\cite{M66},
which is a generalization of Martin-L\"of randomness,
in order to state the principle of typicality.

Around 2000, there were comprehensive attempts
to extend the algorithmic randomness notions
to the quantum regime.
Namely, Berthiaume, van Dam, and Laplante~\cite{BvDL00},
Vit\'{a}nyi~\cite{Vit00},
and G\'{a}cs~\cite{Gac01} independently introduced
the so-called \emph{quantum Kolmogorov complexity}
from each standpoint,
for aiming to define the algorithmic information content of a given individual quantum state of a \emph{finite} dimensional quantum system.
Later on, Tadaki~\cite{T06MLQ} introduced in 2004
the notion of $\hat{\Omega}$
by means of modifying the work of G\'{a}cs~\cite{Gac01}.
This $\hat{\Omega}$
is an extension of Chaitin's halting probability $\Omega$~\cite{C75} to
a measurement operator in an \emph{infinite} dimensional quantum system.
Chaitin's $\Omega$ is the halting probability of an optimal self-delimiting Turing machine,
which is used to define the algorithmic information content (i.e., program-size complexity) of a given (classical) finite string,
and is a concrete example of a Martin-L\"of random real~\cite{C75,M66,Sch73}.
It plays a central role in algorithmic randomness~\cite{C87b,N09,DH10}.
Actually, our $\hat{\Omega}$ is a bounded Hermitian operator on a Hilbert space of \emph{infinite} dimension,
and for every computable state $\ket{\Psi}$ it holds that
the inner-product $\bra{\Psi}\hat{\Omega}\ket{\Psi}$,
which has a meaning of ``probability'' in quantum mechanics,
can be represented as a Chaitin's halting probability $\Omega$.
These previous works~\cite{BvDL00,Vit00,Gac01,T06MLQ} are about the extension of algorithmic randomness notions
to the quantum realm.

In contrast, in this paper as well as our former works~\cite{T14CCR,T15QCOMPINFO,T15Kokyuroku,T16CCR,T16QIP,T16QIT35,T17SCIS,T18arXiv,T20Kokyuroku,T24CCR,T25CCR},
we do not extend the algorithmic randomness notions themselves,
while keeping them in their original forms,
but we identify in the mathematical framework of quantum mechanics
the basic idea of Martin-L\"of randomness
that \emph{the random infinite sequences are precisely sequences which are not contained in any effective null set}.
See Martin-L\"of~\cite{M66}, Nies~\cite{N09}, Downey and Hirschfeldt~\cite{DH10},
and Brattka, Miller, and Nies~\cite{BMiN12} for this basic idea of Martin-L\"of randomness.
In our works, we do this identification, in particular,
from the aspect of the many-worlds interpretation of quantum mechanics~\cite{E57},
leading to the principle of typicality.

\subsection{Contributions of the paper}

The \emph{Wigner's friend paradox}~\cite{W61} is a Gedankenexperiment regarding
when and where the reduction of the state vector occurs
in a chain of the measurements by several observers
where the state of the consciousness of each observer is measured by the subsequent observer, except for the last observer in the chain.
The Wigner's friend paradox
is one of the central open questions in the measurement problem of quantum mechanics.

In this paper, we extend our rigorous framework of quantum mechanics based on the principle of typicality,
which was introduced by Tadaki~\cite{T16CCR,T16QIP,T16QIT35,T17SCIS,T18arXiv},
so that it can be applicable to
the situation, such as
the Wigner's friend paradox,
where apparatuses perform measurements
over other apparatuses as well as over
a normal quantum system only being measured.
For that purpose, we introduce two postulates,
Postulates~\ref{CF} and \ref{Recursive Use} below,
for determining the states of the apparatuses in the \emph{intermediate} stages
among all apparatuses which constitute the whole measurement process.
We make an analysis of the Wigner's friend paradox
within this extended framework of quantum mechanics based on the principle of typicality.
We then draw common sense conclusions about the Wigner's friend paradox.

\emph{Deutsch's thought experiment}~\cite{Deu85} is a variant of the Wigner's friend paradox,
which can, in principle, verify the effect of the consciousness of observer
on the reduction of the state vector.
In this paper, we make an analysis of Deutsch's thought experiment
comprehensively
within the extended framework
of
the principle of typicality.
Based on this, we make a prediction which is testable
in principle.

Based on the extended framework of the principle of typicality,
we can analyze still more complicated situations than
the Wigner's friend paradox or Deutsch's thought experiment.
As such an example, we introduce
a combination of the Wigner's friend paradox and Deutsch's thought experiment,
named the \emph{Wigner-Deutsch collaboration}.
We then make an analysis of it comprehensively in the extended framework of the principle of typicality.
This analysis shows how effectively our extended framework works
in even more complicated situations.

\subsection{Organization of the paper}

The paper is organized as follows.
We begin in Section~\ref{preliminaries} with
some mathematical preliminaries, in particular, about measure theory and
Martin-L\"of randomness with respect to an arbitrary probability measure.
We then review the notion of Martin-L\"of randomness with respect to
an arbitrary
Bernoulli measure,
called the \emph{Martin-L\"of $P$-randomness} in this paper,
in Section~\ref{Sec-ML_P-randomness}.
In Section~\ref{FMP} we summarize
the theorems on Martin-L\"of $P$-randomness from Tadaki~\cite{T14,T15,T16arXiv},
which are need to establish the contributions of this paper presented
in the later sections.
In Section~\ref{QM} we review the central postulates of the \emph{conventional} quantum mechanics
according to Nielsen and Chuang~\cite{NC00}.
In Section~\ref{MWI}
we review our rigorous framework of quantum mechanics based on the \emph{principle of typicality},
which was introduced by Tadaki~\cite{T16CCR,T16QIP,T16QIT35,T17SCIS,T18arXiv}.

In Section~\ref{sec:Schroedinger's cat}
we review the Schr\"{o}dinger's cat~\cite{Schr35}
in the terminology of
the conventional quantum mechanics.
Based on this, in Section~\ref{sec:Analysis of Schroedinger's cat}
we make an analysis of the Schr\"{o}dinger's cat
within our rigorous framework of quantum mechanics based on the principle of typicality,
which is reviewed in Section~\ref{MWI}.
This analysis also demonstrates
\emph{how properly our rigorous framework based on the principle of typicality works in practical problems in quantum mechanics}.

In Section~\ref{sec:confirming point}
we extend our rigorous framework
of quantum mechanics
based on the principle of typicality
so that it is applicable to
situations where apparatuses perform measurements over other apparatuses,
and we introduce the two postulates, Postulates~\ref{CF} and \ref{Recursive Use},
for determining the states of the apparatuses in the intermediate stages
of the whole measurement process.
In Section~\ref{sec:Wigner's friend}
we review the Wigner's friend paradox~\cite{W61}
in the terminology of
the conventional quantum mechanics.
Based on this, in Section~\ref{sec:Analysis of Wigner's friend},
we make an analysis of the Wigner's friend paradox
within the extended framework of the principle of typicality,
developed in Section~\ref{sec:confirming point}.
On the other hand,
in Section~\ref{sec:Deutsch's thought experiment},
we review Deutsch's thought experiment~\cite{Deu85}
in the terminology of
the conventional quantum mechanics.
Based on this, in Section~\ref{sec:Analysis of Deutsch's thought experiment}
we make an analysis of Deutsch's thought experiment
within the extended framework of the principle of typicality.
In Section~\ref{sec:Analysis of DTE with a mere quantum system F}
we make an analysis of Deutsch's thought experiment where the `Friend'
is treated as a mere quantum system and not as a measurement apparatus,
within the
original
framework of the principle of typicality, which is reviewed in Section~\ref{MWI}.

In Section~\ref{sec:Wigner collaborates with Deutsch}
we introduce
the \emph{Wigner-Deutsch collaboration},
which is a combination of the Wigner's friend paradox and Deutsch's thought experiment,
and we make a heuristic analysis of it
in terms of
the conventional quantum mechanics.
In Section~\ref{sec:Analysis of the Wigner-Deutsch collaboration}
we make an analysis of the Wigner-Deutsch collaboration
within the extended framework of the principle of typicality.
In Section~\ref{sec:Analysis of WDc with a mere quantum system F}
we make an analysis of the Wigner-Deutsch collaboration where the `Friend'
is treated as a mere quantum system and not as a measurement apparatus,
within the extended framework of the principle of typicality.
We conclude this paper with summary in Section~\ref{sec-Concluding_remarks}.

\section{Mathematical preliminaries}
\label{preliminaries}

\subsection{Basic notation and definitions}
\label{basic notation}

We start with some notation about numbers and strings which will be used in this paper.
$\N=\left\{0,1,2,3,\dotsc\right\}$ is the set of \emph{natural numbers},
and $\N^+$ is the set of \emph{positive integers}.
For any complex number $z$, its \emph{complex conjugate} is denoted by
$\overline{z}$.

An \emph{alphabet} is a non-empty finite set.
Let $\Omega$ be an arbitrary alphabet throughout the rest of this subsection.
A \emph{finite string over $\Omega$} is a finite sequence of elements from the alphabet $\Omega$.
We use $\Omega^*$ to denote the set of all finite strings over $\Omega$,
which contains the \emph{empty string} denoted by $\lambda$.
A subset $S$ of $\Omega^*$ is called
\emph{prefix-free}
if no string in $S$ is a prefix of another string in $S$.

An \emph{infinite sequence over $\Omega$} is an infinite sequence of elements from the alphabet $\Omega$,
where the sequence is infinite to the right but finite to the left.
We use $\Omega^\infty$ to denote the set of all infinite sequences over $\Omega$.

Let $\alpha\in\Omega^\infty$.
For any $n\in\N$
we denote by $\rest{\alpha}{n}\in\Omega^*$ the first $n$ elements
in the infinite sequence $\alpha$, i.e., the prefix of $\alpha$ of length $n$.
For any $n\in\N^+$ we denote
by $\alpha(n)$ the $n$th element in $\alpha$.
Thus, for example, $\rest{\alpha}{4}=\alpha(1)\alpha(2)\alpha(3)\alpha(4)$, and $\rest{\alpha}{0}=\lambda$.

For any $S\subset\Omega^*$, the set
$\{\alpha\in\Omega^\infty\mid\exists\,n\in\N\;\rest{\alpha}{n}\in S\}$
is denoted by $\osg{S}$.
Note that (i)~$\osg{S}\subset\osg{T}$ for every $S\subset T\subset\Omega^*$, and
(ii)~for every set $S\subset\Omega^*$ there exists a prefix-free set $P\subset\Omega^*$ such that
$\osg{S}=\osg{P}$.
For any $\sigma\in\Omega^*$, we denote by $\osg{\sigma}$ the set $\osg{\{\sigma\}}$, i.e.,
the set of all infinite sequences over $\Omega$ extending $\sigma$.
Therefore $\osg{\lambda}=\Omega^\infty$.

Let $\mathcal{H}$ be an arbitrary complex Hilbert space.
A \emph{projector} $P$ on $\mathcal{H}$ is
a Hermitian bounded operator on $\mathcal{H}$ such that $P^2=P$.
A collection $\{P_a\}_{a\in\Theta}$ of projectors on $\mathcal{H}$ is called
a \emph{projection-valued measure} (\emph{PVM}, for short) \emph{in $\mathcal{H}$} if
$\Theta$ is an alphabet and the following~(i) and (ii) hold:
\begin{enumerate}
\item $P_a P_b=\delta_{a,b}P_a$ for every $a,b\in\Theta$.
\item $\sum_{a\in\Theta} P_a=I$, where $I$ denotes the identity operator on $\mathcal{H}$.
\end{enumerate}

\subsection{Martin-L\"of randomness with respect to an arbitrary probability measure}
\label{MLRam}

We introduce
the notion of Martin-L\"of randomness with respect to an arbitrary probability measure.
For that purpose,
we briefly review measure theory according to Nies~\cite[Section 1.9]{N09}.
See also Billingsley~\cite{B95} for measure theory in general.
We here summarize the description of Tadaki~\cite[Section 2.2]{T16arXiv},
which is based on Nies~\cite[Section 1.9]{N09}.

Let $\Omega$ be an arbitrary alphabet.
A \emph{probability measure representation over $\Omega$} is
a function $r\colon\Omega^*\to[0,1]$ such that
(i)~$r(\lambda)=1$ and
(ii) for every $\sigma\in\Omega^*$ it holds that
\begin{equation*}%
       r(\sigma)=\sum_{a\in\Omega}r(\sigma a).
\end{equation*}
A probability measure representation
over $\Omega$ \emph{induces} a probability measure
on $\Omega^\infty$ in the following manner.

A class $\mathcal{F}$ of subsets of $\Omega^\infty$ is called
a \emph{$\sigma$-field on $\Omega^\infty$}
if  $\mathcal{F}$ includes $\Omega^\infty$, is closed under complements,
and is closed under the formation of countable unions.
A subset $\mathcal{R}$ of $\Omega^\infty$ is called \emph{open} if
$\mathcal{R}=\osg{S}$ for some $S\subset\Omega^*$.
The \emph{Borel class} $\mathcal{B}_{\Omega}$ is the $\sigma$-field \emph{generated by}
all open sets on $\Omega^\infty$.
Namely, the Borel class $\mathcal{B}_{\Omega}$ is defined
as the intersection of all the $\sigma$-fields on $\Omega^\infty$ containing
all open sets on $\Omega^\infty$.
A real-valued function $\mu$ defined on the Borel class $\mathcal{B}_{\Omega}$ is called
a \emph{probability measure on $\Omega^\infty$} if the following conditions hold:
\begin{enumerate}
  \item $\mu\left(\emptyset\right)=0$ and $\mu\left(\Omega^\infty\right)=1$;
  \item $\mu\left(\bigcup_{i}\mathcal{D}_i\right)=\sum_{i}\mu\left(\mathcal{D}_i\right)$
    for every sequence $\{\mathcal{D}_i\}_{i\in\N}$ of sets in $\mathcal{B}_{\Omega}$ such that
    $\mathcal{D}_i\cap\mathcal{D}_i=\emptyset$ for all $i\neq j$.
\end{enumerate}
Then we can show that
for every probability measure representation $r$ over $\Omega$
there exists a unique probability measure $\mu$ on $\Omega^\infty$ such that
\begin{equation}\label{mr}
  \mu\left(\osg{\sigma}\right)=r(\sigma)
\end{equation}
for
every $\sigma\in \Omega^*$.
This unique probability measure $\mu$ is called
the \emph{probability measure induced by the probability measure representation $r$}, denoted $\mu_r$.

Now,
we introduce the notion of \emph{Martin-L\"of randomness} \cite{M66}
in a general setting,
as follows.

\begin{definition}[Martin-L\"{o}f randomness with respect to a probability measure]
\label{ML-randomness-wrtm}
Let $\Omega$ be an alphabet, and let $\mu$ be a probability measure on $\Omega^\infty$.
A subset $\mathcal{C}$ of $\N^+\times\Omega^*$ is called a
\emph{Martin-L\"{o}f test with respect to $\mu$} if
$\mathcal{C}$ is a recursively enumerable set and for every $n\in\N^+$ it holds that
$\mathcal{C}_n$ is a prefix-free subset of $\Omega^*$ and
\begin{equation}\label{muocn<2n}
  \mu\left(\osg{\mathcal{C}_n}\right)<2^{-n},
\end{equation}
where $\mathcal{C}_n$ denotes the set
$\left\{\,
    \sigma\mid (n,\sigma)\in\mathcal{C}
\,\right\}$.

For any $\alpha\in\Omega^\infty$, we say that $\alpha$ is
\emph{Martin-L\"{o}f random with respect to $\mu$} if
$$\alpha\notin\bigcap_{n=1}^{\infty}\osg{\mathcal{C}_n}$$
for every Martin-L\"{o}f test $\mathcal{C}$ with respect to $\mu$.\qed
\end{definition}

\section{\boldmath Martin-L\"of $P$-randomness}
\label{Sec-ML_P-randomness}

The principle of typicality, Postulate~\ref{POT} below, is stated by means of the notion of 
Martin-L\"of randomness with respect to an arbitrary probability measure
introduced in the preceding section.
However,
in many situations of the applications of the principle of typicality,
such as in a contribution of this paper described
in Sections~\ref{sec:Analysis of Schroedinger's cat}, \ref{sec:Analysis of Wigner's friend}, \ref{sec:Analysis of Deutsch's thought experiment}, \ref{sec:Analysis of DTE with a mere quantum system F}, \ref{sec:Analysis of the Wigner-Deutsch collaboration}, and \ref{sec:Analysis of WDc with a mere quantum system F},
a more restricted notion is used where
the probability measure is chosen to be
a
Bernoulli measure.
Namely,
the notion of \emph{Martin-L\"of randomness with respect to
a
Bernoulli measure}
is used
in many situations of the applications
of the principle of typicality.
In order to introduce this notion, we first
review the notions of \emph{finite probability space} and \emph{Bernoulli measure}
as follows.

\begin{definition}[Finite probability space]\label{def-FPS}
Let $\Omega$ be an alphabet. A \emph{finite probability space on $\Omega$} is a function $P\colon\Omega\to [0,1]$
such that
(i) $P(a)\ge 0$ for every $a\in \Omega$, and
(ii) $\sum_{a\in \Omega}P(a)=1$.
The set of all finite probability spaces on $\Omega$ is denoted by $\PS(\Omega)$.

Let $P\in\PS(\Omega)$.
The set $\Omega$ is called the \emph{sample space} of $P$,
and elements
of
$\Omega$ are called \emph{sample points} or \emph{elementary events}
of $P$.
For each $A\subset\Omega$, we define $P(A)$ by
$$P(A):=\sum_{a\in A}P(a).$$
A subset of $\Omega$ is called an \emph{event} on $P$, and
$P(A)$ is called the \emph{probability} of $A$
for every event $A$
on $P$.
\qed
\end{definition}

Let $\Omega$ be an alphabet, and let $P\in\PS(\Omega)$.
For each $\sigma\in\Omega^*$, we use $P(\sigma)$ to denote
$$P(\sigma_1)P(\sigma_2)\dots P(\sigma_n)$$
where $\sigma=\sigma_1\sigma_2\dots\sigma_n$ with $\sigma_i\in\Omega$.
Therefore $P(\lambda)=1$, in particular.
For each subset $S$ of $\Omega^*$, we use $P(S)$ to denote
$$\sum_{\sigma\in S}P(\sigma).$$
Therefore $P(\emptyset)=0$, in particular.

Consider a function $r\colon\Omega^*\to[0,1]$ such that
$r(\sigma)=P(\sigma)$ for every $\sigma\in\Omega^*$.
It is then easy to see that the function~$r$ is a probability measure representation over $\Omega$.
The probability measure~$\mu_r$ induced by
the probability measure representation
$r$ is
called
the \emph{Bernoulli measure on $\Omega^\infty$} (\emph{induced by the finite probability space $P$}), denoted $\lambda_{P}$.
The Bernoulli measure~$\lambda_{P}$ on $\Omega^\infty$
satisfies that
\begin{equation}\label{BmPosgsigma=Psigma}
  \Bm{P}{\osg{\sigma}}=P(\sigma)
\end{equation}
for every $\sigma\in \Omega^*$,
due to \eqref{mr}.

The notion of \emph{Martin-L\"of randomness with respect to a Bernoulli measure} is defined as follows.
We call it the \emph{Martin-L\"of $P$-randomness},
since it depends on a finite probability space~$P$.
This notion was, in essence, introduced by Martin-L\"{o}f~\cite{M66},
as well as the notion of Martin-L\"of randomness
with respect to Lebesgue measure.

\begin{definition}[%
Martin-L\"of $P$-randomness,
Martin-L\"{o}f \cite{M66}]\label{ML_P-randomness}
Let $\Omega$ be an alphabet, and let $P\in\PS(\Omega)$.
For any $\alpha\in\Omega^\infty$, we say that $\alpha$ is \emph{Martin-L\"{o}f $P$-random} if
$\alpha$ is Martin-L\"{o}f random with respect to $\lambda_{P}$.\qed
\end{definition}

\emph{Martin-L\"{o}f random} infinite binary sequences \cite{M66} are precisely 
Martin-L\"{o}f $U$-random infinite binary sequences
where $U$ is a finite probability space on $\{0,1\}$ such that $U(0)=U(1)=1/2$.

\section{\boldmath The properties of Martin-L\"of $P$-randomness}
\label{FMP}

In order to
obtain the results in this paper,
we need
several
theorems
on Martin-L\"of $P$-randomness
from
Tadaki~\cite{T14,T15,T16arXiv}.
Originally,
these theorems
played
a key part
in developing
the theory of
\emph{operational characterization of the notion of probability}
based on Martin-L\"of $P$-randomness in
Tadaki~\cite{T14,T15,T16arXiv}.
We enumerate
these theorems
in this section.

\subsection{The law of large numbers}

First, the following theorem shows that the law of large numbers holds
for an arbitrary Martin-L\"of $P$-random infinite sequence.
For the proof of Theorem~\ref{FI},
see Tadaki~\cite[Theorem~14]{T16arXiv}.

\begin{theorem}[The law of large numbers, Tadaki~\cite{T14}]\label{FI}
Let $\Omega$ be an alphabet, and let $P\in\PS(\Omega)$.
For every $\alpha\in\Omega^\infty$, if $\alpha$ is Martin-L\"of $P$-random
then for every $a\in\Omega$ it holds that
$$\lim_{n\to\infty} \frac{N_a(\rest{\alpha}{n})}{n}=P(a),$$
where $N_a(\sigma)$ denotes the number of the occurrences of $a$ in $\sigma$ for every $a\in\Omega$ and every $\sigma\in\Omega^*$.
\qed
\end{theorem}

Note that in Theorem~\ref{FI}
the underlying finite probability space $P$ is \emph{quite arbitrary}, and
thus we do not impose any computability restrictions on
$P$ at all, in particular.

\subsection{Marginal distribution}

The following theorem states that
Martin-L\"of $P$-randomness is
\emph{closed under the formation of marginal distribution}.
This theorem
is an immediate consequence of Theorem~16 of Tadaki~\cite{T16arXiv}.
See Tadaki~\cite[Theorem~9.4]{T18arXiv} for the details of the proof of Theorem~\ref{contraction2}.

\begin{theorem}[Tadaki~\cite{T18arXiv}]\label{contraction2}
Let $\Omega$ and $\Theta$ be alphabets.
Let $P\in\PS(\Omega\times\Theta)$, and let $\alpha\in(\Omega\times\Theta)^\infty$.
Suppose that $\beta$ is an infinite sequence over $\Omega$ obtained from $\alpha$
by replacing each element $(l,m)$ in $\alpha$ by $l$.
If $\alpha$ is Martin-L\"of $P$-random then $\beta$ is Martin-L\"of $Q$-random,
where $Q\in\PS(\Omega)$ such that
$$Q(l):=\sum_{m\in\Theta}P(l,m)$$
for every $l\in\Omega$.
\qed
\end{theorem}

Note that in Theorems~\ref{contraction2}
the underlying finite probability space $P$ on $\Omega\times\Theta$
is \emph{quite arbitrary}, and
thus we do not impose any computability restrictions on $P$ at all, in particular.

\subsection{Non-occurrence of an elementary event with probability zero}

From various aspects, Tadaki~\cite[Section~5.1]{T16arXiv} demonstrated the fact that
\emph{an elementary event with probability zero never occurs
in the conventional quantum mechanics}.
This fact corresponds to Theorem~\ref{thm:zero_probability} below
in the context of the theory of
\emph{operational characterization of the notion of probability},
which was developed by Tadaki~\cite{T14,T15,T16arXiv}
based on Martin-L\"of $P$-randomness.

\begin{theorem}\label{thm:zero_probability}
Let $\Omega$ be an alphabet, and let $P\in\PS(\Omega)$.
Let $\alpha\in\Omega^\infty$ and $a\in\Omega$.
Suppose that
$\alpha$ is Martin-L\"of $P$-random
and $P(a)=0$.
Then $\alpha$ does not contain $a$.
\qed
\end{theorem}

This result for Martin-L\"of $P$-randomness was, in essence,
pointed out by Martin-L\"of~\cite{M66}.
Note that
we do not impose any computability restrictions on
the underlying finite probability space~$P$ at all
in Theorem~\ref{thm:zero_probability}.
For the proof of Theorem~\ref{thm:zero_probability},
see Tadaki~\cite[Theorem~13]{T16arXiv}.
The following corollary is immediate from Theorem~\ref{thm:zero_probability}.

\begin{corollary}\label{cor:always-positive-probability}
Let $\Omega$ be an alphabet, and let $P\in\PS(\Omega)$.
For every $\alpha\in\Omega^\infty$, if $\alpha$ is Martin-L\"of $P$-random
then $P(\alpha(n))>0$ for every $n\in\N^+$.
\qed
\end{corollary}

\section{Postulates of the conventional quantum mechanics}
\label{QM}

In this paper, we call the standard quantum mechanics which is presented in the widely prevalent textbooks on quantum mechanics the \emph{conventional quantum mechanics}.
In this section, we review the central postulates of the conventional quantum mechanics.
Specifically, we here refer to the postulates of
the conventional
quantum mechanics
in the form presented in Nielsen and Chuang \cite[Chapter 2]{NC00},
with slight modification.
Note that their original postulates were given as postulates of
the conventional
quantum mechanics for a \emph{finite}-dimensional quantum system, i.e.,
a quantum system whose state space is a finite-dimensional complex Hilbert space,
since their book is a textbook of the field of quantum computation and quantum information
where finite-dimensional quantum systems are typical.
In contrast,
\emph{we do not impose any restrictions on the dimensionality of the underlying state spaces}
in the central postulates of the conventional quantum mechanics of the form presented in what follows.

The first postulate of the conventional quantum mechanics is about \emph{state space} and \emph{state vector}.

\begin{postulate}[State space and state vector]\label{state_space}
Associated to any isolated physical system is a separable complex Hilbert space
known as the \emph{state space} of the system.
The system is completely described by its \emph{state vector},
which is a unit vector in the system's state space.
\qed
\end{postulate}

The second postulate of the conventional quantum mechanics is about the \emph{composition} of systems.

\begin{postulate}[Composition of systems]\label{composition}
The state space of a composite physical system is the tensor product of the state spaces of the component physical systems.
Moreover, if we have systems numbered $1$ through $n$, and system number~$i$ is
prepared
in the state~$\ket{\Psi_i}$,
then the joint state of the total system is
$\ket{\Psi_1}\otimes\ket{\Psi_2}\otimes\dots\otimes\ket{\Psi_n}$.
\qed
\end{postulate}

The third postulate of the conventional quantum mechanics is about the \emph{time-evolution} of \emph{closed} quantum systems.

\begin{postulate}[Unitary time-evolution]\label{evolution}
The evolution of a \emph{closed} quantum system is described by a \emph{unitary transformation}.
That is,
the state~$\ket{\Psi_1}$ of the system at time~$t_1$ is related to the state~$\ket{\Psi_2}$ of the system at time~$t_2$
by a unitary operator~$U$
on the state space,
which depends only on the times~$t_1$ and $t_2$,
in such a way that
$\ket{\Psi_2}=U\ket{\Psi_1}$.
\qed
\end{postulate}

The forth postulate of the conventional quantum mechanics is about \emph{measurements} on quantum systems.

\begin{postulate}[Quantum measurement]\label{Born-rule}
A quantum measurement is described by a collection $\{M_m\}_{m\in\Omega}$ of \emph{measurement operators}
which is a finite collection of bounded operators on the state space of the system being measured,
satisfying the \emph{completeness equation}
\begin{equation*}
  \sum_{m\in\Omega} M_m^\dag M_m=I,
\end{equation*}
where $M_m^\dag$ denotes the adjoint of the bounded operator $M_m$, and $I$ denotes the identity operator on the state space.
\begin{enumerate}
\item
The set of possible outcomes of the measurement
equals
the
non-empty
finite set $\Omega$.
If the state of the quantum system is $\ket{\Psi}$ immediately before the measurement then
the probability that result~$m$ occurs is given by
\begin{equation*}
  \bra{\Psi}M_m^\dag M_m\ket{\Psi}.
\end{equation*}
\item
Given that outcome $m$ occurred,
the state of the quantum system immediately after the measurement is
\begin{equation*}
  \frac{M_m\ket{\Psi}}{\sqrt{\bra{\Psi}M_m^\dag M_m\ket{\Psi}}}.
\end{equation*}
\qed
\end{enumerate}
\end{postulate}

Postulate~\ref{Born-rule}~(i) is the so-called \emph{Born rule}, i.e,
\emph{the probability interpretation of the wave function},
while Postulate~\ref{Born-rule}~(ii) is called the \emph{projection hypothesis}.
Thus,
the Born rule, Postulate~\ref{Born-rule}~(i),
uses the \emph{notion of probability}.
However, the operational characterization of the notion of probability is not given in the Born rule,
and therefore the relation of its statement to a specific infinite sequence of outcomes of quantum measurements which are being generated by an infinitely repeated measurements is \emph{unclear}.
Tadaki~\cite{T14CCR,T15QCOMPINFO,T15Kokyuroku,T16CCR,T16QIP,T16QIT35,T17SCIS,T18arXiv} fixed this point.

In this paper
as well as our former works~\cite{T14CCR,T15QCOMPINFO,T15Kokyuroku,T16CCR,T16QIP,T16QIT35,T17SCIS,T18arXiv,T20Kokyuroku,T24CCR,T25CCR},
\emph{we keep Postulates~\ref{state_space}, \ref{composition}, and \ref{evolution}
in their original forms without any modifications}.
The principle of typicality, Postulate~\ref{POT} below,
is proposed as a \emph{refinement} of Postulate~\ref{Born-rule} to replace it,
based on the notion of \emph{Martin-L\"of randomness with respect to a probability measure}.

\section{The principle of typicality: Our world is typical}
\label{MWI}

Everett~\cite{E57} introduced
the \emph{many-worlds interpretation of quantum mechanics} (\emph{MWI}, for short)
in 1957.
Actually, his MWI was more than just an interpretation of quantum mechanics.
It aimed to derive
Postulate~\ref{Born-rule}~(i), the Born rule,
from the remaining postulates, i.e.,
from Postulates~\ref{state_space}, \ref{composition}, and \ref{evolution} above.
In this sense, Everett~\cite{E57} proposed his MWI
as a ``metatheory'' of quantum mechanics.
The point is that in MWI the measurement process is fully treated
as the interaction between a system being measured and an apparatus measuring it, based only on
Postulates~\ref{state_space}, \ref{composition}, and \ref{evolution}.
Then his MWI tried to derive Postulate~\ref{Born-rule}~(i) in such a setting.
However, his attempt does not seem successful.
One of the reasons is thought to be its mathematical vagueness.
(See DeWitt and Graham~\cite{DN73}, including the reprints of Everett's paper~\cite{E57} and his original Ph.D.~thesis, and Hartle~\cite{H68} for the many-worlds interpretation of quantum mechanics in general.)

Tadaki~\cite{T16CCR}
reformulated the original framework of MWI by Everett~\cite{E57}
\emph{in a form of mathematical rigor} from a modern point of view.
The point of this rigorous treatment of the framework of MWI
is the use of the notion of \emph{probability measure representation} and
its \emph{induction of probability measure}, as presented in Section~\ref{MLRam}.
Tadaki~\cite{T16CCR} then introduced the \emph{principle of typicality}
in this rigorous framework.
The principle of typicality is an operational refinement of Postulate~\ref{Born-rule}
based on
the notion of Martin-L\"{o}f randomness with respect to a probability measure
given in Definition~\ref{ML-randomness-wrtm}.
In what follows,
we review this rigorous framework of MWI based the principle of typicality.
See Tadaki~\cite{T18arXiv}
for the details of this
rigorous
framework of MWI
and its validity.
Also, for the details of
the difference between
this
framework of MWI
and that of the original MWI~\cite{E57}, see Tadaki~\cite{T18arXiv}.

Now,
according to Tadaki~\cite{T18arXiv},
let us introduce the
setting
of MWI \emph{in terms of our terminology
in a form of mathematical rigor}.
Let $\mathcal{S}$ be an arbitrary quantum system with state space $\mathcal{H}$
of an arbitrary dimension.
Consider a measurement over $\mathcal{S}$ described
by arbitrary measurement operators $\{M_m\}_{m\in\Omega}$ satisfying
the completeness equation,
\begin{equation}\label{completeness-equation}
  \sum_{m\in\Omega}M_m^{\dagger} M_m=I.
\end{equation}
Here, $\Omega$
is the set of all possible outcomes of the measurement.%
\footnote{The set $\Omega$ is non-empty and finite, and therefore it is an alphabet.}
Let $\mathcal{A}$ be an apparatus performing the measurement described by $\{M_m\}_{m\in\Omega}$,
which is a quantum system with state space $\ssoa$.
According to Postulates~\ref{state_space}, \ref{composition}, and \ref{evolution},
the measurement process of
the measurement described by the measurement operators $\{M_m\}_{m\in\Omega}$
is described by a unitary operator $U$ on $\mathcal{H}\otimes\ssoa$ such that
\begin{equation}\label{single_measurement}
  U(\ket{\Psi}\otimes\ket{\Phi_{\mathrm{init}}})=\sum_{m\in\Omega}(M_m\ket{\Psi})\otimes\ket{\Phi[m]}
\end{equation}
for every $\ket{\Psi}\in\mathcal{H}$.
Actually,
$U$ describes the interaction between the system $\mathcal{S}$ and the apparatus $\mathcal{A}$.
The vector $\ket{\Phi_{\mathrm{init}}}\in\ssoa$ is
the initial state of the apparatus $\mathcal{A}$, and $\ket{\Phi[m]}\in\ssoa$ is
a final state of the apparatus $\mathcal{A}$ for each $m\in\Omega$,
with $\braket{\Phi[m]}{\Phi[m']}=\delta_{m,m'}$.
For every $m\in\Omega$, the state $\ket{\Phi[m]}$ indicates that
\emph{the apparatus $\mathcal{A}$ records the value $m$
as the measurement outcome}.
By the unitary interaction~\eqref{single_measurement} as a measurement process,
a correlation (i.e., entanglement) is generated between
the system~$\mathcal{S}$ and the apparatus~$\mathcal{A}$.
The unitary interaction~\eqref{single_measurement} of the measurement process is
due to von Neumann~\cite[Section~VI.3]{vN55}
in the case where the measurement operators $\{M_m\}_{m\in\Omega}$ form
a PVM
in $\mathcal{H}$
such that $M_m$ is of rank $1$, i.e., $\dim M_m(\mathcal{H})=1$, for every $m\in\Omega$.
For the derivation of \eqref{single_measurement} in the general case,
see Tadaki~\cite[Section~7.1]{T18arXiv}.

In the framework of MWI,
we consider countably infinite copies of the system $\mathcal{S}$,
and consider a countably infinite repetition of the measurements
described by the identical measurement operators $\{M_m\}_{m\in\Omega}$
performed over each of such copies in a sequential order,
where each of the measurements is described
by the unitary time-evolution~\eqref{single_measurement}.
As repetitions of the measurement progressed,
correlations between the systems and the apparatuses are being generated in sequence
in the superposition of the total system consisting of the systems and the apparatuses.
The detail
is described as follows.

For each $n\in\N^+$, let $\mathcal{S}_n$ be the $n$th copy of the system $\mathcal{S}$ and
$\mathcal{A}_n$ the $n$th copy of the apparatus $\mathcal{A}$.
Each $\mathcal{S}_n$ is prepared in a state $\ket{\Psi_n}$
while all $\mathcal{A}_n$ are prepared in the identical state $\ket{\Phi_{\mathrm{init}}}$.
The measurement
described by the measurement operators $\{M_m\}_{m\in\Omega}$
is performed over each $\mathcal{S}_n$ one by one
in the increasing order of $n$,
by interacting each $\mathcal{S}_n$ with $\mathcal{A}_n$
according to the unitary time-evolution~\eqref{single_measurement}.
For each $n\in\N^+$,
let $\mathcal{H}_n$ be the state space of the \emph{total} system consisting of
the first $n$ copies
$\mathcal{S}_1, \mathcal{A}_1, \mathcal{S}_2, \mathcal{A}_2,\dots,\mathcal{S}_n, \mathcal{A}_n$
of the system $\mathcal{S}$ and the apparatus $\mathcal{A}$.
These successive interactions between the copies of
the system $\mathcal{S}$ and the apparatus $\mathcal{A}$ as measurement processes
proceed in the following manner:

The
starting
state of the total system,
which consists of $\mathcal{S}_1$ and $\mathcal{A}_1$,
is $\ket{\Psi_1}\otimes\ket{\Phi_{\mathrm{init}}}\in\mathcal{H}_1$.
Immediately after the measurement
described by $\{M_m\}_{m\in\Omega}$
over $\mathcal{S}_1$,
the total system results in the state
\begin{align*}
  \sum_{m_1\in\Omega} (M_{m_1}\ket{\Psi_1})\otimes\ket{\Phi[m_1]}\in\mathcal{H}_1
\end{align*}
by the interaction~\eqref{single_measurement} as a measurement process.
In general,
immediately before
performing
the measurement
described by $\{M_m\}_{m\in\Omega}$
over $\mathcal{S}_n$,
the state of the total system,
which consists of
$\mathcal{S}_1, \mathcal{A}_1, \mathcal{S}_2, \mathcal{A}_2,\dots,\mathcal{S}_n, \mathcal{A}_n$,
is
\begin{align*}
  \sum_{m_1,\dots,m_{n-1}\in\Omega}
  (M_{m_1}\ket{\Psi_1})\otimes\dots\otimes(M_{m_{n-1}}\ket{\Psi_{n-1}})\otimes\ket{\Psi_n}
  \otimes\ket{\Phi[m_1]}\otimes\dots\otimes\ket{\Phi[m_{n-1}]}\otimes\ket{\Phi_{\mathrm{init}}}
\end{align*}
in $\mathcal{H}_n$,
where
$\ket{\Psi_n}$ is the initial state of $\mathcal{S}_n$ and
$\ket{\Phi_{\mathrm{init}}}$ is the initial state of $\mathcal{A}_n$.
Immediately after the measurement
described by $\{M_m\}_{m\in\Omega}$
over $\mathcal{S}_n$,
the total system results in the state
\begin{align}
   &\sum_{m_1,\dots,m_{n}\in\Omega}
  (M_{m_1}\ket{\Psi_1})\otimes\dots\otimes(M_{m_{n}}\ket{\Psi_n})
  \otimes\ket{\Phi[m_1]}\otimes\dots\otimes\ket{\Phi[m_{n}]} \nonumber \\ %
  =&\sum_{m_1,\dots,m_{n}\in\Omega}
  (M_{m_1}\ket{\Psi_1})\otimes\dots\otimes(M_{m_{n}}\ket{\Psi_n})
  \otimes\ket{\Phi[m_1\dots m_{n}]} \label{total_system}
\end{align}
in $\mathcal{H}_n$,
by the interaction~\eqref{single_measurement} as a measurement process between
the system $\mathcal{S}_n$ prepared in the state $\ket{\Psi_n}$
and the apparatus $\mathcal{A}_n$ prepared in the state $\ket{\Phi_{\mathrm{init}}}$.
The vector $\ket{\Phi[m_1\dots m_n]}$ denotes
the vector $\ket{\Phi[m_1]}\otimes\dots\otimes\ket{\Phi[m_n]}$ which represents
the state of $\mathcal{A}_1,\dots,\mathcal{A}_n$.
This state indicates that \emph{the apparatuses $\mathcal{A}_1,\dots,\mathcal{A}_n$ record
the values $m_1\dots m_n$
as the measurement outcomes
over $\mathcal{S}_1,\dots,\mathcal{S}_n$, respectively}.

In the superposition~\eqref{total_system},
on letting $n\to\infty$,
the length of
the records $m_1\dots m_n$ of the values
as the measurement outcomes
in the apparatuses $\mathcal{A}_1,\dots,\mathcal{A}_n$ diverges to infinity.
The
consideration of
this
limiting case results in the
definition of a \emph{world}.
Namely, a \emph{world} is defined as
an
infinite sequence of records of the values
as the measurement outcomes
in the apparatuses.
Thus, in the case described
so far,
a world is an infinite sequence over $\Omega$,
and the finite records $m_1\dots m_n$ in each state $\ket{\Phi[m_1\dots m_n]}$
in the superposition~\eqref{total_system} of the total system is a \emph{prefix} of a world.

For aiming at deriving Postulate~\ref{Born-rule}, the original MWI~\cite{E57} assigned
``weight'' to each of worlds.
In our rigorous framework of MWI,
we introduce a \emph{probability measure} on the set of all worlds in the following manner:
First, we introduce a probability measure representation on the set of prefixes of worlds, i.e.,
the set $\Omega^*$ in this case.
This probability measure representation is given by a function $r\colon\Omega^*\to[0,1]$ with
\begin{equation}\label{rpmwi}
  r(m_1\dotsc m_n)=\prod_{k=1}^n\bra{\Psi_k}M_{m_k}^\dag M_{m_k}\ket{\Psi_k},
\end{equation}
which is the square of the norm of each vector
$(M_{m_1}\ket{\Psi_1})\otimes\dots\otimes(M_{m_{n}}\ket{\Psi_n})\otimes\ket{\Phi[m_1\dots m_{n}]}$
in the superposition~\eqref{total_system}.
Using the completeness equation~\eqref{completeness-equation},
it is easy to check that $r$ is certainly a probability measure representation over $\Omega$.
We call the probability measure representation $r$
\emph{the
probability
measure representation
for the prefixes of
worlds}.
We then adopt
the probability measure \emph{induced} by the probability measure representation $r$
for the prefixes of worlds
as \emph{the probability measure on the set of all worlds}.

We summarize the above consideration and clarify
the definitions of the notion of \emph{world} and
the notion of \emph{the probability measure representation for the prefixes of worlds}
as in the following.

\begin{definition}%
[The probability measure representation for the prefixes of worlds, Tadaki~\cite{T18arXiv}]\label{pmrpwst}
Consider an arbitrary
quantum system $\mathcal{S}$ and
a
measurement over $\mathcal{S}$
described by arbitrary \emph{measurement operators $\{M_m\}_{m\in\Omega}$},
where the measurement process is described by \eqref{single_measurement} as an interaction
of the system $\mathcal{S}$ with an apparatus $\mathcal{A}$.
We suppose the following situation:
\begin{enumerate}
\item There are countably infinite copies $\mathcal{S}_1, \mathcal{S}_2, \mathcal{S}_3 \dotsc$
  of the system $\mathcal{S}$ and countably infinite copies
  $\mathcal{A}_1, \mathcal{A}_2, \mathcal{A}_3, \dotsc$ of the apparatus $\mathcal{A}$.
\item For each $n\in\N^+$, the system $\mathcal{S}_n$ is prepared in a state $\ket{\Psi_n}$,%
  \footnote{In Definition~\ref{pmrpwst}, all $\ket{\Psi_n}$ are not required to be an identical state.}
  while the apparatus $\mathcal{A}_n$ is prepared in the state $\ket{\Phi_{\mathrm{init}}}$,
  and then the measurement described by $\{M_m\}_{m\in\Omega}$ is performed over $\mathcal{S}_n$
  by interacting it with the apparatus $\mathcal{A}_n$ according to
  the unitary time-evolution~\eqref{single_measurement}.
\item Starting the measurement described by $\{M_m\}_{m\in\Omega}$ over $\mathcal{S}_1$,
  the measurement described by $\{M_m\}_{m\in\Omega}$ over each $\mathcal{S}_n$ is
  performed in the increasing order of $n$.
\end{enumerate}
We then note that,
for each $n\in\N^+$,
immediately after the measurement described by $\{M_m\}_{m\in\Omega}$ over $\mathcal{S}_n$,
the state of the total system consisting of
$\mathcal{S}_1, \mathcal{A}_1, \mathcal{S}_2, \mathcal{A}_2,\dots,\mathcal{S}_n, \mathcal{A}_n$ is
$$\ket{\Theta_n}:=\sum_{m_1,\dots,m_n\in\Omega}\ket{\Theta(m_1,\dots,m_n)},$$
where
\begin{equation}\label{eq:kTm1n=MmkP1nkPm1n}
\ket{\Theta(m_1,\dots,m_n)}
:=(M_{m_1}\ket{\Psi_1})\otimes\dots\otimes(M_{m_{n}}\ket{\Psi_n})
\otimes\ket{\Phi[m_1]}\otimes\dots\otimes\ket{\Phi[m_n]}.
\end{equation}
The vectors $M_{m_1}\ket{\Psi_1},\dots,M_{m_{n}}\ket{\Psi_n}$ are states of $\mathcal{S}_1,\dots,\mathcal{S}_n$ up to the normalization factors, respectively, and
the vectors
$$\ket{\Phi[m_1]},\dots,\ket{\Phi[m_n]}$$
are
states of $\mathcal{A}_1,\dots,\mathcal{A}_n$, respectively.
The state
vector
$\ket{\Theta_n}$ of the total system is normalized while
each of the vectors $\{\ket{\Theta(m_1,\dots,m_n)}\}_{m_1,\dots,m_n\in\Omega}$
is not necessarily normalized.
Then,
\emph{the
probability
measure representation for the prefixes of
worlds}
is defined
as
a
function $p\colon\Omega^*\to[0,1]$ such that
$p(\lambda):=1$ and
\begin{equation}\label{p=bTmkT}
  p(m_1\dotsc m_n):=\braket{\Theta(m_1,\dots,m_n)}{\Theta(m_1,\dots,m_n)}
\end{equation}
for every $n\in\N^+$ and every $m_1,\dots,m_n\in\Omega$.
Moreover,
an infinite sequence over $\Omega$,
i.e., an infinite sequence of
possible
outcomes of the measurement described by $\{M_m\}_{m\in\Omega}$,
is called a \emph{world}.
\qed
\end{definition}

In Definition~\ref{pmrpwst}, it is easy to check that the function $p$ defined by \eqref{p=bTmkT} is
certainly a probability measure representation over $\Omega$.
Actually, based on
\eqref{p=bTmkT}, \eqref{eq:kTm1n=MmkP1nkPm1n}, and
the completeness equation~\eqref{completeness-equation},
we have that
\begin{equation*}
\begin{split}
  p(m_1\dotsc m_n)&=\prod_{k=1}^{n}\bra{\Psi_k}M_{m_k}^\dag M_{m_k}\ket{\Psi_k}
  =\left[\,\prod_{k=1}^{n}\bra{\Psi_k}M_{m_k}^\dag M_{m_k}\ket{\Psi_k}\right]\braket{\Psi_{n+1}}{\Psi_{n+1}} \\
  &=\left[\,\prod_{k=1}^{n}\bra{\Psi_k}M_{m_k}^\dag M_{m_k}\ket{\Psi_k}\right]
  \sum_{m_{n+1}\in\Omega}\bra{\Psi_{n+1}}M_{m_{n+1}}^\dag M_{m_{n+1}}\ket{\Psi_{n+1}} \\
  &=\sum_{m_{n+1}\in\Omega}\prod_{k=1}^{n+1}\bra{\Psi_k}M_{m_k}^\dag M_{m_k}\ket{\Psi_k} \\
  &=\sum_{m_{n+1}\in\Omega} p(m_1\dotsc m_n m_{n+1})
\end{split}
\end{equation*}
for each $n\in\N$ and each $m_1,\dots,m_n\in\Omega$, as desired.

As mentioned above, the original MWI by Everett~\cite{E57} aimed to derive
Postulate~\ref{Born-rule}~(i), the Born rule,
from the remaining postulates.
However, it would seem impossible to do this for several reasons
(see Tadaki~\cite{T18arXiv} for
these reasons).
Instead, it is appropriate to introduce an additional postulate
in our rigorous framework of MWI developed above,
in order to overcome the defect of the original MWI
and to make quantum mechanics \emph{operationally perfect}.
Thus, we put forward the principle of typicality as follows.

Intuitively, the principle of typicality states that \emph{our world is typical}.
Namely, our world is Martin-L\"of random with respect to the probability measure on the set of all worlds,
induced by the probability measure representation for the prefixes of worlds,
in the superposition of the total system which consists of
systems being measured and apparatuses measuring them.
Formally, the principle of typicality is stated as
the following postulate.

\begin{postulate}[The principle of typicality, Tadaki~\cite{T16CCR}]\label{POT}
We assume
the notation and notions introduced and defined in Definition~\ref{pmrpwst},
and consider the same setting of measurements as considered in Definition~\ref{pmrpwst}.
Let $\omega\in\Omega^\infty$ be our world.
Then $\omega$ is Martin-L\"{o}f random with respect to
the probability measure $\mu_p$ induced by
the probability measure representation $p\colon\Omega^*\to[0,1]$
for the prefixes of worlds.
Moreover, in our world $\omega$,
for every $n\in\N^+$
the state of
the composite system $\mathcal{S}_n+\mathcal{A}_n$,
which consists of
the $n$th copy $\mathcal{S}_n$ of the system $\mathcal{S}$ and
the $n$th copy $\mathcal{A}_n$ of the apparatus $\mathcal{A}$,
immediately after the measurement performed by $\mathcal{A}_n$ over $\mathcal{S}_n$
is given by
\begin{equation*}
(M_{\omega(n)}\ket{\Psi_n})\otimes\ket{\Phi[\omega(n)]},
\end{equation*}
up to the normalization factor.
\qed
\end{postulate}

For the comprehensive arguments of the validity of
Postulate~\ref{POT}, the principle of typicality,
see Tadaki~\cite{T18arXiv}.
For example, based on the results of
Tadaki~\cite{T14,T15,T16arXiv},
we can see that Postulate~\ref{POT} is certainly a refinement of
Postulate~\ref{Born-rule}
from the point of view of our intuitive understanding of the notion of probability.

Now, let us assume
the notation and notions introduced and defined in Definition~\ref{pmrpwst},
and consider the same setting of measurements as considered in Definition~\ref{pmrpwst}.
Assume further that all of the states $\{\ket{\Psi_n}\}_{n\in\N^+}$ are \emph{identical}, that is,
assume that there exists a state $\ket{\Psi}$ of the system $\mathcal{S}$
such that $\ket{\Psi_n}=\ket{\Psi}$ for all $n\in\N^+$.
Then, according to Definition~\ref{pmrpwst},
we see that
in this setting of measurements
the \emph{probability measure representation for the prefixes of worlds}
is the mapping
\begin{equation}\label{eq:O*nislmtPs}
  \Omega^*\ni \sigma\longmapsto P(\sigma),
\end{equation}
where $P$ is a finite probability space on $\Omega$ such that
$P(m)$ is the square of the norm of the vector
\begin{equation*}
  (M_{m}\ket{\Psi})\otimes\ket{\Phi[m]}
\end{equation*}
for every $m\in\Omega$,
i.e.,
$P(m)=\bra{\Psi}M_{m}^\dag M_{m}\ket{\Psi}$ for every $m\in\Omega$.
Therefore,
in this setting of measurements
the probability measure \emph{induced by} the probability measure representation for the prefixes of worlds
is the Bernoulli measure $\lambda_P$ on $\Omega^\infty$,
which satisfies \eqref{BmPosgsigma=Psigma}.
Let $\omega\in\Omega^\infty$ be our world in this setting of measurements.
Then it follows from Postulate~\ref{POT}
that $\omega$ is Martin-L\"{o}f random with respect to the Bernoulli measure $\lambda_{P}$, that is, \emph{$\omega$ is Martin-L\"of $P$-random}.
In the applications of Postulate~\ref{POT}
studied in Sections~\ref{sec:Analysis of Schroedinger's cat},
\ref{sec:Analysis of Wigner's friend},
\ref{sec:Analysis of Deutsch's thought experiment},
\ref{sec:Analysis of DTE with a mere quantum system F},
\ref{sec:Analysis of the Wigner-Deutsch collaboration}, and
\ref{sec:Analysis of WDc with a mere quantum system F} below,
\emph{every} probability measure representation for the prefixes of worlds is of the form~\eqref{eq:O*nislmtPs}, and therefore,
\emph{every} probability measure induced by the probability measure representation for the prefixes of worlds is a Bernoulli measure.

In Section~\ref{sec:Analysis of Schroedinger's cat} below,
we will make an analysis of the Schr\"{o}dinger's cat~\cite{Schr35}
in our rigorous framework of quantum mechanics based on
Postulates~\ref{POT}, the principle of typicality,
together with Postulates~\ref{state_space}, \ref{composition}, and \ref{evolution}.
This analysis serves as one of the demonstrations of how properly our framework
based on the principle of typicality
works in practical problems in quantum mechanics.
Before making such an analysis,
we review the Schr\"{o}dinger's cat~\cite{Schr35}
\emph{in the terminology of the conventional quantum mechanics}
in the next section.

\section{Schr\"{o}dinger's cat}
\label{sec:Schroedinger's cat}

In this section,
we describe the Schr\"{o}dinger's cat~\cite{Schr35}
in the terminology of the conventional quantum mechanics,
which is reviewed in Section~\ref{QM}.

Consider a single qubit system $\mathcal{S}$ with state space
$\mathcal{H}_\mathcal{S}$,
and consider a quantum system $\mathcal{C}$ with state space $\mathcal{H}_\mathcal{C}$, where the latter serves as a ``cat''.
Let $\ket{0}$ and $\ket{1}$ form an orthonormal basis of
the state space $\mathcal{H}_\mathcal{S}$ of the
system $\mathcal{S}$.
We assume that the system $\mathcal{S}$ and the system $\mathcal{C}$
interact with each other according to the unitary time-evolution
described by a unitary operator $U_\mathcal{C}$
acting on
$\mathcal{H}_\mathcal{S}\otimes\mathcal{H}_\mathcal{C}$
such that
\begin{equation}\label{eq:U-SC-Q0}
  U_\mathcal{C}(\ket{0}\otimes\ket{\Phi^{\mathcal{C}}_{\mathrm{init}}})
  =\ket{0}\otimes\ket{\Phi^{\mathcal{C}}[0]},
\end{equation}
and
\begin{equation}\label{eq:U-SC-Q1}
  U_\mathcal{C}(\ket{1}\otimes\ket{\Phi^{\mathcal{C}}_{\mathrm{init}}})
  =\ket{1}\otimes\ket{\Phi^{\mathcal{C}}[1]}.
\end{equation}
Here, the vector $\ket{\Phi^{\mathcal{C}}_{\mathrm{init}}}\in\mathcal{H}_\mathcal{C}$ is
the initial state of the system $\mathcal{C}$ in the interaction,
and $\ket{\Phi^{\mathcal{C}}[k]}\in\mathcal{H}_\mathcal{C}$ is
a final state of the system $\mathcal{C}$ in the interaction for each $k=0,1$,
with $\braket{\Phi^{\mathcal{C}}[k]}{\Phi^{\mathcal{C}}[k']}=\delta_{k,k'}$.
In the setting of the Schr\"{o}dinger's cat~\cite{Schr35},
the system~$\mathcal{C}$ is chosen to be a `macroscopic' system,
and the states $\ket{\Phi^{\mathcal{C}}[0]}$ and $\ket{\Phi^{\mathcal{C}}[1]}$
are `macroscopically distinguishable'  states of the system~$\mathcal{C}$
from each other.

Then we consider an observer $\mathcal{W}$ who is on the outside of
the composite system $\mathcal{S}+\mathcal{C}$
consisting of the system $\mathcal{S}$ and the system $\mathcal{C}$. 
Immediately after the interaction above between the system $\mathcal{S}$ and the system $\mathcal{C}$
via the unitary time-evolution described by the unitary operator $U_\mathcal{C}$,
the observer $\mathcal{W}$
performs a measurement $\mathcal{M}^\mathcal{W}$
described by the measurement operators
$\{M^{\mathcal{W}}_0,M^{\mathcal{W}}_1,M^{\mathcal{W}}_2\}$
over the system $\mathcal{C}$, according to Postulate~\ref{Born-rule}, where
\begin{equation}\label{eq:SC-MW-def}
\begin{split}
  M^{\mathcal{W}}_0&:=\product{\Phi^{\mathcal{C}}[0]}{\Phi^{\mathcal{C}}[0]},\\
  M^{\mathcal{W}}_1&:=\product{\Phi^{\mathcal{C}}[1]}{\Phi^{\mathcal{C}}[1]},\\
  M^{\mathcal{W}}_2&:=\sqrt{I_\mathcal{C}-{M^{\mathcal{W}}_0}^\dag M^{\mathcal{W}}_0-{M^{\mathcal{W}}_1}^\dag M^{\mathcal{W}}_1}
  =I_\mathcal{C}-M^{\mathcal{W}}_0-M^{\mathcal{W}}_1,
\end{split}
\end{equation}
and $I_\mathcal{C}$
denotes
the identity operator acting on $\mathcal{H}_\mathcal{C}$.

Now, let $c_0$ and $c_1$ be \emph{arbitrary} two non-zero complex numbers such that $\abs{c_0}^2+\abs{c_1}^2=1$, and we assume that
the system $\mathcal{S}$ is initially in a state $\ket{+}$
in the above setting,
where $\ket{+}$ is defined by
\begin{equation*}%
  \ket{+}:=c_0\ket{0}+c_1\ket{1}.
\end{equation*}
Then,
using \eqref{eq:U-SC-Q0} and \eqref{eq:U-SC-Q1} we see that
the composite system $\mathcal{S}+\mathcal{C}$ is in a state
$\ket{\Psi^{\mathcal{S}+\mathcal{C}}[+]}\in\mathcal{H}_\mathcal{S}\otimes\mathcal{H}_\mathcal{C}$
defined by
\begin{equation}\label{eq:SC-Def-kPS+F+}
  \ket{\Psi^{\mathcal{S}+\mathcal{C}}[+]}:=c_0\ket{0}\otimes\ket{\Phi^{\mathcal{C}}[0]}+c_1\ket{1}\otimes\ket{\Phi^{\mathcal{C}}[1]}
\end{equation}
immediately after the interaction above between the system $\mathcal{S}$ and the system $\mathcal{C}$
via the unitary time-evolution described by the unitary operator $U_\mathcal{C}$.
Thus, according to Postulate~\ref{Born-rule},
the measurement $\mathcal{M}^\mathcal{W}$ performed
by the observer $\mathcal{W}$ over the system $\mathcal{C}$ results in
one of the two possibilities (i) and (ii) below:
\begin{enumerate}
\item The measurement $\mathcal{M}^\mathcal{W}$ gives the outcome $0$,
  and the state of the composite system $\mathcal{S}+\mathcal{C}$
  immediately after the measurement $\mathcal{M}^\mathcal{W}$
  is $\ket{0}\otimes\ket{\Phi^{\mathcal{C}}[0]}\in\mathcal{H}_\mathcal{S}\otimes\mathcal{H}_\mathcal{C}$.
\item The measurement $\mathcal{M}^\mathcal{W}$ gives the outcome $1$,
  and the state of the composite system $\mathcal{S}+\mathcal{C}$
  immediately after the measurement $\mathcal{M}^\mathcal{W}$
  is $\ket{1}\otimes\ket{\Phi^{\mathcal{C}}[1]}\in\mathcal{H}_\mathcal{S}\otimes\mathcal{H}_\mathcal{C}$.
\end{enumerate}
Here, the possibilities (i) and (ii) occur with probabilities $\abs{c_0}^2$ and $\abs{c_1}^2$, respectively.

The state~\eqref{eq:SC-Def-kPS+F+} is a superposition of the `macroscopically distinguishable' states
$\ket{0}\otimes\ket{\Phi^{\mathcal{C}}[0]}$
and $\ket{1}\otimes\ket{\Phi^{\mathcal{C}}[1]}$
of the composite system $\mathcal{S}+\mathcal{C}$,
since the states $\ket{\Phi^{\mathcal{C}}[0]}$ and $\ket{\Phi^{\mathcal{C}}[1]}$
are `macroscopically distinguishable'  states of the system~$\mathcal{C}$
and the coefficients $c_0$ and $c_1$ are both non-zero.
The state~\eqref{eq:SC-Def-kPS+F+} is called a \emph{Schr\"{o}dinger's cat state}.

In the above argument of this section, according to Schr\"{o}dinger~\cite{Schr35}
the quantum system $\mathcal{C}$, which serves as a ``cat'',
is treated as a \emph{mere} quantum system and not as a measurement apparatus.
In contrast,
a situation where the quantum system $\mathcal{C}$ is treated as a measurement apparatus instead of as a mere quantum system
in the setting of the Schr\"{o}dinger's cat~\cite{Schr35} above
is just the situation of the Wigner's friend~\cite{W61},
which will be investigated in Sections~\ref{sec:Wigner's friend} and \ref{sec:Analysis of Wigner's friend} below.

\section{Analysis of Schr\"{o}dinger's cat based on the principle of typicality}
\label{sec:Analysis of Schroedinger's cat}

In this section, we make an analysis of
the Schr\"{o}dinger's cat,
which is described in the preceding section in the terminology of the conventional quantum mechanics,
in terms of our rigorous framework of quantum mechanics based on
Postulates~\ref{POT},
the principle of typicality,
together with Postulates~\ref{state_space}, \ref{composition}, and \ref{evolution}.

To proceed this program,
first of all
\emph{we have to implement everything
in the setting of the Schr\"{o}dinger's cat described in the preceding section
by unitary time-evolution}
so that the premise assumed in Definition~\ref{pmrpwst} is fulfilled
for applying Postulates~\ref{POT}.
First,
the observer $\mathcal{W}$
is regarded
as a quantum system with state space $\mathcal{H}_\mathcal{W}$.
Then, according to \eqref{single_measurement},
the measurement process of
the measurement
$\mathcal{M}^\mathcal{W}$
is described by a unitary operator $U_\mathcal{W}$
acting on
$\mathcal{H}_\mathcal{C}\otimes\mathcal{H}_\mathcal{W}$
such that
\begin{equation}\label{eq:SC-U-MW-F}
  U_\mathcal{W}(\ket{\Psi}\otimes\ket{\Phi^{\mathcal{W}}_{\mathrm{init}}})
  =\sum_{l=0,1,2}(M^{\mathcal{W}}_l\ket{\Psi})\otimes\ket{\Phi^{\mathcal{W}}[l]}
\end{equation}
for every $\ket{\Psi}\in\mathcal{H}_\mathcal{C}$.
The unitary operator
$U_\mathcal{W}$ describes the interaction
between the system $\mathcal{C}$ and the observer $\mathcal{W}$
which is regarded as a quantum system.
The vector $\ket{\Phi^{\mathcal{W}}_{\mathrm{init}}}\in\mathcal{H}_\mathcal{W}$ is
the initial state of the
observer
$\mathcal{W}$, and $\ket{\Phi^{\mathcal{W}}[l]}\in\mathcal{H}_\mathcal{W}$ is
a final state of the
observer
$\mathcal{W}$ for each $l=0,1,2$,
with $\braket{\Phi^{\mathcal{W}}[l]}{\Phi^{\mathcal{W}}[l']}=\delta_{l,l'}$.
For each $l=0,1,2$,
the state $\ket{\Phi^{\mathcal{W}}[l]}$ indicates that
\emph{the observer $\mathcal{W}$ records the value $l$}.

Then, for complying with the setting in Definition~\ref{pmrpwst} exactly,
the measurement process $U_\mathcal{W}$ by the observer $\mathcal{W}$ over the system $\mathcal{C}$
must be extended over the composite system $\mathcal{S}+\mathcal{C}$ as follows:
We denote the set $\{0,1,2\}$ by $\Omega$,
and we define a finite
collection
$\{M_{l}\}_{l\in\Omega}$ of operators acting on
$\mathcal{H}_\mathcal{S}\otimes\mathcal{H}_\mathcal{C}$ by
\begin{equation}\label{eq:POT-SC-Mkl=dlkMFk}
  M_{l}:=I_\mathcal{S}\otimes M^{\mathcal{W}}_l,
\end{equation}
where $I_\mathcal{S}$ denotes the identity operator acting on $\mathcal{H}_\mathrm{S}$.
Since $\{M^{\mathcal{W}}_l\}_{l\in\Omega}$ forms measurement operators acting on $\mathcal{H}_\mathcal{C}$,
it follows that
\begin{equation*}
  \sum_{l\in\Omega} M_{l}^\dag M_{l}
  =I_\mathcal{S}\otimes\sum_{l=0,1,2}{M^{\mathcal{W}}_l}^\dag M^{\mathcal{W}}_l=I_\mathcal{S}\otimes I_\mathcal{C}.
\end{equation*}
Thus, the finite collection $\{M_{l}\}_{l\in\Omega}$ satisfies the \emph{completeness equation}, and therefore forms \emph{measurement operators}.
We define a unitary operator $U_{\mathrm{whole}}$ acting on $\mathcal{H}_\mathcal{S}\otimes\mathcal{H}_\mathcal{C}\otimes\mathcal{H}_\mathcal{W}$
by
\[
  U_{\mathrm{whole}}:=I_\mathcal{S}\otimes U_{\mathcal{W}}.
\]
Then
using \eqref{eq:SC-U-MW-F} and \eqref{eq:POT-SC-Mkl=dlkMFk} we see that
\begin{equation*}%
\begin{split}
  U_{\mathrm{whole}}(\ket{\Psi_\mathcal{S}}\otimes\ket{\Psi_\mathcal{C}}\otimes\ket{\Phi^{\mathcal{W}}_{\mathrm{init}}})
  &=(I_\mathcal{S}\otimes U_{\mathcal{W}})
      (\ket{\Psi_\mathcal{S}}\otimes\ket{\Psi_\mathcal{C}}\otimes\ket{\Phi^{\mathcal{W}}_{\mathrm{init}}}) \\
  &=\ket{\Psi_\mathcal{S}}\otimes U_{\mathcal{W}}
      (\ket{\Psi_\mathcal{C}}\otimes\ket{\Phi^{\mathcal{W}}_{\mathrm{init}}}) \\
  &=\ket{\Psi_\mathcal{S}}\otimes
      \sum_{l=0,1,2}(M^{\mathcal{W}}_l\ket{\Psi_\mathcal{C}})
      \otimes\ket{\Phi^{\mathcal{W}}[l]} \\
  &=\sum_{l\in\Omega}(M_l(\ket{\Psi_\mathcal{S}}\otimes\ket{\Psi_\mathcal{C}}))
      \otimes\ket{\Phi^{\mathcal{W}}[l]}
\end{split}
\end{equation*}
for each $\ket{\Psi_\mathcal{S}}\in\mathcal{H}_\mathcal{S}$
and $\ket{\Psi_\mathcal{C}}\in\mathcal{H}_\mathcal{C}$.
Therefore, we have that
\begin{equation}\label{eq:SC-all}
  U_{\mathrm{whole}}(\ket{\Psi}\otimes\ket{\Phi^{\mathcal{W}}_{\mathrm{init}}})
  =\sum_{l\in\Omega}(M_l\ket{\Psi})\otimes\ket{\Phi^{\mathcal{W}}[l]}
\end{equation}
for every state $\ket{\Psi}\in\mathcal{H}_\mathcal{S}\otimes\mathcal{H}_\mathcal{C}$ of the composite system $\mathcal{S}+\mathcal{C}$.
The unitary time-evolution of the measurement process of the measurement $\mathcal{M}^\mathcal{W}$ by the observer $\mathcal{W}$ over the system $\mathcal{C}$ in the form of \eqref{eq:SC-all} corresponds to
``the measurement process described by \eqref{single_measurement} as an interaction
of the system $\mathcal{S}$ with an apparatus $\mathcal{A}$''
in Definition~\ref{pmrpwst}.

On the other hand, recall that
the state of the composite system $\mathcal{S}+\mathcal{C}$
immediately after the interaction between the system $\mathcal{S}$ and the system $\mathcal{C}$
according to
the unitary time-evolution
described by the unitary operator $U_\mathcal{C}$
is given by $\ket{\Psi^{\mathcal{S}+\mathcal{C}}[+]}$,
which is defined by \eqref{eq:SC-Def-kPS+F+}.
This state $\ket{\Psi^{\mathcal{S}+\mathcal{C}}[+]}$ of the composite system $\mathcal{S}+\mathcal{C}$ corresponds to
``the initial state $\ket{\Psi_n}$ of the system $\mathcal{S}_n$'' in Definition~\ref{pmrpwst}
for all $n\in\N^+$.

\subsection{Application of the principle of typicality}
\label{subsec:SC-Appl-POT}

Thus, in the above setting of the Schr\"{o}dinger's cat,
the unitary operator $U_{\mathrm{whole}}$ applying to
the state
$\ket{\Psi^{\mathcal{S}+\mathcal{C}}[+]}\otimes\ket{\Phi^{\mathcal{W}}_{\mathrm{init}}}$ describes
the \emph{repeated once} of the infinite repetition of the measurement
where
the measurement described by the measurement operators $\{M_{l}\}_{l\in\Omega}$
and performed over the composite system $\mathcal{S}+\mathcal{C}$ in the ``initial state'' $\ket{\Psi^{\mathcal{S}+\mathcal{C}}[+]}$
is infinitely repeated.
Actually,
as is confirmed from the form of \eqref{eq:SC-all},
the application of $U_{\mathrm{whole}}$ itself, of course,
forms a \emph{single measurement}, which is described by
the measurement operators $\{M_{l}\}_{l\in\Omega}$
and whose all possible outcomes form the set $\Omega$.

Hence, we can apply Definition~\ref{pmrpwst}
to this scenario of the setting of measurements.
Thus,
according to Definition~\ref{pmrpwst},
we can see that a \emph{world} is an infinite sequence over $\Omega$
and the probability measure induced by
the \emph{probability measure representation for the prefixes of worlds}
is a Bernoulli measure $\lambda_P$ on
$\Omega^\infty$,
where $P$ is a finite probability space on $\Omega$ such that
$P(l)$ is the square of the norm of the vector
\begin{equation*}
  (M_{l}\ket{\Psi^{\mathcal{S}+\mathcal{C}}[+]})\otimes\ket{\Phi^{\mathcal{W}}[l]}
\end{equation*}
for every $l\in\Omega$.
Here $\Omega$ is the set of all possible records of
the observer $\mathcal{W}$
in the \emph{repeated once} of the experiments.
Let us calculate the explicit form of $P(l)$.
First, using \eqref{eq:POT-SC-Mkl=dlkMFk}, \eqref{eq:SC-MW-def}, and \eqref{eq:SC-Def-kPS+F+} we have that
\begin{equation}\label{eq:SC_Mklm+Pklm=fklmoss2-0}
  (M_{l}\ket{\Psi^{\mathcal{S}+\mathcal{C}}[+]})\otimes\ket{\Phi^{\mathcal{W}}[l]}
  =c_l\ket{l}\otimes\ket{\Phi^{\mathcal{C}}[l]}\otimes\ket{\Phi^{\mathcal{W}}[l]}
\end{equation}
for every $l\in\{0,1\}$,
and therefore
\begin{equation}\label{eq:SC-finitepsP-0-1}
  P(l)=\abs{c_l}^2
\end{equation}
for every $l\in\{0,1\}$.
On the other hand, since $P$ is a finite probability space on $\Omega$,
we have that
\[
  \sum_{l\in\Omega} P(l)=1.
\]
Therefore, it follows from \eqref{eq:SC-finitepsP-0-1} and $\abs{c_0}^2+\abs{c_1}^2=1$ that
\begin{equation}\label{eq:SC-finitepsP-2}
  P(2)=0.
\end{equation}

Now, let us apply
Postulate~\ref{POT}, the \emph{principle of typicality},
to the setting of measurements
developed above.
Let $\omega$ be \emph{our world} in the infinite repetition of
the measurement described by the measurement operators $\{M_{l}\}_{l\in\Omega}$
in the above setting.
This $\omega$
is an infinite sequence over $\Omega$
consisting of records
in the observer $\mathcal{W}$
which is being generated by the infinite repetition of the measurement
described by the measurement operators $\{M_{l}\}_{l\in\Omega}$
in the above setting.
Since the Bernoulli measure $\lambda_P$ on $\Omega^\infty$ is
the probability measure induced by the
probability
measure representation
for the prefixes of
worlds
in the above setting,
it follows from
Postulate~\ref{POT}
that $\omega$ is Martin-L\"{o}f random with respect to the Bernoulli measure $\lambda_{P}$, that is, \emph{$\omega$ is Martin-L\"of $P$-random}.

Then, since $\omega$ is Martin-L\"of $P$-random, it follows from \eqref{eq:SC-finitepsP-2} and
Corollary~\ref{cor:always-positive-probability}
that
\begin{equation}\label{eq:SC-omega-neq-2}
  \omega(n)\neq 2
\end{equation}
for every $n\in\N^+$, i.e., $\omega$ is an infinite binary sequence.%
\footnote{Based on \eqref{eq:SC-finitepsP-0-1},
we can further show that $\omega$ is a Martin-L\"of $B$-random infinite binary sequence,
where $B$ is a finite probability space on $\{0,1\}$ such that
$B(0)=\abs{c_0}^2$ and $B(1)=\abs{c_1}^2$.}
Thus, it follows from Postulate~\ref{POT} and \eqref{eq:SC-all}
that, in our world $\omega$, for each $n\in\N^+$ the state of
the $n$th composite system $\mathcal{S}+\mathcal{C}+\mathcal{W}$,
i.e., the $n$th copy of a composite system consisting of
the two systems $\mathcal{S}$ and $\mathcal{C}$ and the observer $\mathcal{W}$,
immediately after the measurement $\mathcal{M}^\mathcal{W}$ is given by
\begin{equation}\label{eq:SC-state-immafter-MW}
  (M_{\omega(n)}\ket{\Psi^{\mathcal{S}+\mathcal{C}}[+]})\otimes\ket{\Phi^{\mathcal{W}}[\omega(n)]}
  =c_{\omega(n)}\ket{\omega(n)}\otimes\ket{\Phi^{\mathcal{C}}[\omega(n)]}\otimes\ket{\Phi^{\mathcal{W}}[\omega(n)]}
\end{equation}
up to the normalization factor,
where the equality follows from \eqref{eq:SC-omega-neq-2} and \eqref{eq:SC_Mklm+Pklm=fklmoss2-0}.
Moreover, since $\omega$ is Martin-L\"of $P$-random,
it follows from Theorem~\ref{FI} and \eqref{eq:SC-finitepsP-0-1} that
for every $l\in\{0,1\}$ it holds that
\begin{equation}\label{eq:SC-imafterW-by-PoT-LLN}
  \lim_{n\to\infty} \frac{N_l(\rest{\omega}{n})}{n}
  =\abs{c_l}^2,
\end{equation}
where $N_l(\rest{\omega}{n})$ denotes the number of the occurrences of $l$
in the prefix of $\omega$ of length $n$.
Recall that, for each $l\in\{0,1\}$, the state $\ket{\Phi^{\mathcal{W}}[l]}$ indicates that
\emph{the observer $\mathcal{W}$ records the value $l$}.
Thus, based on \eqref{eq:SC-state-immafter-MW} and \eqref{eq:SC-imafterW-by-PoT-LLN},
 we have that, in our world $\omega$, \emph{the following holds for each $l\in\{0,1\}$:
In a proportion of $\abs{c_l}^2$ out of the infinite repetitions of the experiment,
the state of the composite system $\mathcal{S}+\mathcal{C}$ is $\ket{l}\otimes\ket{\Phi^{\mathcal{C}}[l]}$ and
the observer $\mathcal{W}$ records the value $l$,
immediately after the measurement $\mathcal{M}^\mathcal{W}$},
as expected from the point of view of the conventional quantum mechanics.

Hence,
in this manner, we have operationally refined and reformulated
the argument of the Schr\"{o}dinger's cat
given in Section~\ref{sec:Schroedinger's cat},
in our rigorous framework of quantum mechanics based on
Postulates~\ref{POT},
the principle of typicality,
together with Postulates~\ref{state_space}, \ref{composition}, and \ref{evolution}.
In comparison,
the statements and results given in Section~\ref{sec:Schroedinger's cat}
regarding
the Schr\"{o}dinger's cat
are
at least
operationally vague
since they are described
in terms of the conventional quantum mechanics
where the operational characterization of the notion of probability is not given.

In the above analysis of this section,
we treat the quantum system $\mathcal{C}$
as a \emph{mere quantum system} and not as a measuring apparatus,
within
our
rigorous framework of quantum mechanics based on Postulates~\ref{POT},
together with Postulates~\ref{state_space}, \ref{composition}, and \ref{evolution}.
In order to make the corresponding analysis for
the situation of the Wigner's friend~\cite{W61}
where the quantum system $\mathcal{C}$ is treated as a \emph{measurement apparatus} instead of as a mere quantum system
in the setting of the Schr\"{o}dinger's cat above,
we must
extend our framework based on the principle of typicality
over such situations.
We will do this in the next section comprehensively from a general point of view.

\section{Apparatus measured by apparatuses and its confirming point}
\label{sec:confirming point}

In this section we further develop the
rigorous
framework based on the principle of typicality
reviewed
in Section~\ref{MWI},
in order to make it applicable to
the situation where apparatuses perform measurements
over other apparatuses as well as
over
a normal quantum system being measured.

\subsection{The general setting of the problem}
\label{subsec:The general setting of the problem}

Consider measuring apparatuses $\mathcal{A}_1,\dots,\mathcal{A}_n$
performing measurements $\mathcal{M}^1,\dots,\mathcal{M}^n$
described by measurement operators
$\{M^1_{m_1}\}_{m_1\in\Omega_1},\dots,\{M^n_{m_n}\}_{m_n\in\Omega_n}$,
respectively,
where $\Omega_1,\dots,\Omega_n$ are
the sets of all possible outcomes of the measurements $\mathcal{M}^1,\dots,\mathcal{M}^n$.
The measuring apparatuses $\mathcal{A}_1,\dots,\mathcal{A}_n$ themselves
are quantum systems with state spaces
$\overline{\mathcal{H}}_1\dots,\overline{\mathcal{H}}_n$,
respectively.
We also consider
a quantum system $\mathcal{S}$ with state space $\mathcal{H}_{\mathcal{S}}$.
We assume the situation in which the following two conditions~(i) and (ii) hold:
\begin{enumerate}
\item The measurements $\mathcal{M}^1,\dots,\mathcal{M}^n$ are performed in this order
in a manner that
the measurement $\mathcal{M}^{k+1}$ is performed immediately after
the measurement $\mathcal{M}^{k}$ is completed for every $k=1,\dots,n-1$.
\item
The apparatus $\mathcal{A}_k$ performs the measurement $\mathcal{M}^k$
over the composite system consisting of the system $\mathcal{S}$
and the apparatuses $\mathcal{A}_{1},\dots,\mathcal{A}_{k-1}$
for every $k=1,\dots,n$.
\end{enumerate}
Therefore,
according to the condition~(ii) above,
the measurement operator $M^k_{m_k}$ is
then
an endomorphism of
$\mathcal{H}_{\mathcal{S}}\otimes\overline{\mathcal{H}}_1\otimes\dots\otimes\overline{\mathcal{H}}_{k-1}$
for every $k=1,\dots,n$ and every $m_k\in\Omega_k$.
Note in particular that
the apparatus $\mathcal{A}_1$ performs the measurement $\mathcal{M}^1$ over the system $\mathcal{S}$, and
the measurement operator $M^1_{m_1}$ is an endomorphism of $\mathcal{H}_{\mathcal{S}}$ for every $m_1\in\Omega_1$.

Thus, according to \eqref{single_measurement}, for each $k=1,\dots,n$,
the measurement process of the measurement $\mathcal{M}^k$
is described by a unitary operator $U_k$ acting on
$\mathcal{H}_\mathcal{S}\otimes\overline{\mathcal{H}}_1\otimes\dots\otimes\overline{\mathcal{H}}_k$
such that
\begin{equation}\label{eq:Ui-TSCP}
  U_k(\ket{\Psi}\otimes\ket{\Phi^k_{\mathrm{init}}})
  =\sum_{m_k\in\Omega_k}(M^k_{m_k}\ket{\Psi})\otimes\ket{\Phi^k[m_k]}
\end{equation}
for every $\ket{\Psi}\in\mathcal{H}_\mathcal{S}\otimes\overline{\mathcal{H}}_1\otimes\dots\otimes\overline{\mathcal{H}}_{k-1}$.
For each $k=1,\dots,n$, the vector $\ket{\Phi^k_{\mathrm{init}}}\in\overline{\mathcal{H}}_k$ is
the initial state of the apparatus $\mathcal{A}_k$,
and the vector $\ket{\Phi^k[m_k]}\in\overline{\mathcal{H}}_k$ is
a final state of the apparatus $\mathcal{A}_k$ for each $m_k\in\Omega_k$,
with $\braket{\Phi^k[m_k]}{\Phi^k[m_k']}=\delta_{m_k,m_k'}$.
For every $k=1,\dots,n$ and every $m_k\in\Omega_k$, the state $\ket{\Phi^k[m_k]}$ indicates that
\emph{the apparatus $\mathcal{A}_k$ records the value $m_k$
as the measurement outcome}.
For each $k=1,\dots,n$, let $\mathcal{H}_k$ be
a subspace of $\overline{\mathcal{H}}_k$ spanned by
the vectors
\[
  \{\ket{\Phi^k[m_k]}\}_{m_k\in\Omega_k}.
\]
As a natural requirement, for each $k=2,\dots,n$, we assume that
\begin{equation}\label{eq:Ui-TSCP-MO}
  M^k_{m_k}\ket{\Psi}\in\mathcal{H}_\mathcal{S}\otimes\mathcal{H}_1\otimes\dots\otimes\mathcal{H}_{k-1}
\end{equation}
for every $m_k\in\Omega_k$ and every $\ket{\Psi}\in\mathcal{H}_\mathcal{S}\otimes\mathcal{H}_k\otimes\dots\otimes\mathcal{H}_{k-1}$.

Let us consider the total system $\mathcal{S}_{\mathrm{Total}}$
consisting of the system $\mathcal{S}$
and the apparatuses $\mathcal{A}_1,\dots,\mathcal{A}_n$.
For any $i,j\in\{0,\dots,n\}$ with $i\le j$, we use $$I_{i,i+1,\dots,j-1,j}$$ to denote
the identity operator acting on $\overline{\mathcal{H}}_i\otimes\overline{\mathcal{H}}_{i+1}\otimes\dots\otimes\overline{\mathcal{H}}_{j-1}\otimes\overline{\mathcal{H}}_j$,
where $\overline{\mathcal{H}}_0$ denotes $\mathcal{H}_{\mathcal{S}}$.
Therefore,
for example,
$I_0$ denotes the identity operator acting on $\mathcal{H}_{\mathcal{S}}$,
$I_{0,1,2}$ denotes the identity operator acting on
$\mathcal{H}_{\mathcal{S}}\otimes\overline{\mathcal{H}}_1\otimes\overline{\mathcal{H}}_2$, and
$I_{2,3,4,5}$ denotes the identity operator acting on
$\overline{\mathcal{H}}_2\otimes\overline{\mathcal{H}}_3\otimes\overline{\mathcal{H}}_4\otimes\overline{\mathcal{H}}_5$.
For each $k=1,\dots,n$, we define a unitary operator $U_{\mathrm{Total}}^{1,\dots,k}$
on $\mathcal{H}_\mathcal{S}\otimes\overline{\mathcal{H}}_1\otimes\dots\otimes\overline{\mathcal{H}}_n$ by
\begin{equation}\label{eq:def-UT1tok:=UkotI-U1otI}
  U_{\mathrm{Total}}^{1,\dots,k}:=
  (U_{k}\otimes I_{k+1,\dots,n})\circ(U_{k-1}\otimes I_{k,\dots,n})\circ
  \dotsm\dotsm\circ
  (U_2\otimes I_{3,\dots,n})\circ(U_1\otimes I_{2,\dots,n}).
\end{equation}
Since $U_{1},\dots,U_{n}$ are unitary operators, based on \eqref{eq:Ui-TSCP}
and \eqref{eq:Ui-TSCP-MO}, it is easy to see that
for every $k=1,\dots,n$ there exist
unique
measurement operators
$\{M^{1,\dots,k}_{m_1,\dots,m_k}\}_{(m_1,\dots,m_k)\in\Omega_1\times\dots\times\Omega_k}$ on $\mathcal{H}_{\mathcal{S}}$,
satisfying the completeness equation
\[
  \sum_{(m_1,\dots,m_k)\in\Omega_1\times\dots\times\Omega_k}
  {M^{1,\dots,k}_{m_1,\dots,m_k}}^\dag M^{1,\dots,k}_{m_1,\dots,m_k}
  =I_{0},
\]
such that for every $\ket{\Psi}\in\mathcal{H}_\mathcal{S}$ it holds that
\begin{equation}\label{eq:Ui-TSCP-1tok}
\begin{split}
  &U_{\mathrm{Total}}^{1,\dots,k}(\ket{\Psi}\otimes\ket{\Phi^1_{\mathrm{init}}}\otimes\dots\otimes\ket{\Phi^n_{\mathrm{init}}}) \\
  &=\sum_{(m_1,\dots,m_k)\in\Omega_1\times\dots\times\Omega_k}
  (M^{1,\dots,k}_{m_1,\dots,m_k}\ket{\Psi})
  \otimes\ket{\Phi^1[m_1]}\otimes\dots\otimes\ket{\Phi^k[m_k]}
  \otimes\ket{\Phi^{k+1}_{\mathrm{init}}}\otimes\dots\otimes\ket{\Phi^n_{\mathrm{init}}}.
\end{split}
\end{equation}
For every $\ket{\Psi}\in\mathcal{H}_\mathcal{S}$,
if $\ket{\Psi}$ is a state of the system $\mathcal{S}$
immediately before the measurement $\mathcal{M}^1$,
then $\ket{\Psi}\otimes\ket{\Phi^1_{\mathrm{init}}}\otimes\dots\otimes\ket{\Phi^n_{\mathrm{init}}}$
is the state of the total system $\mathcal{S}_{\mathrm{Total}}$
immediately before the measurement $\mathcal{M}^1$
and $U_{\mathrm{Total}}^{1,\dots,k}(\ket{\Psi}\otimes\ket{\Phi^1_{\mathrm{init}}}\otimes\dots\otimes\ket{\Phi^n_{\mathrm{init}}})$
is the state of the total system $\mathcal{S}_{\mathrm{Total}}$
immediately after the measurement $\mathcal{M}^k$
for every $k=1,\dots,n$.

Let $\ket{\Psi^0_{\mathrm{init}}}$ be an arbitrary state of the system $\mathcal{S}$.
In what follows,
we assume that
the system $\mathcal{S}$ is in the state $\ket{\Psi^0_{\mathrm{init}}}$
immediately before the measurement $\mathcal{M}^1$.
Then, for each $k=1,\dots,n$,
we define $\ket{\Psi^k_{\mathrm{Total}}}$ by
\begin{equation}\label{eq:def-kPkT:=UT1tokkP0i-kP1i-kPni}
  \ket{\Psi^k_{\mathrm{Total}}}:=
  U_{\mathrm{Total}}^{1,\dots,k}(\ket{\Psi^0_{\mathrm{init}}}\otimes\ket{\Phi^1_{\mathrm{init}}}\otimes\dots\otimes\ket{\Phi^n_{\mathrm{init}}}),
\end{equation}
which is the state of the total system $\mathcal{S}_{\mathrm{Total}}$
immediately after the measurement $\mathcal{M}^k$.
In particular, it follows from \eqref{eq:Ui-TSCP-1tok} that
the state $\ket{\Psi^n_{\mathrm{Total}}}$
of the total system $\mathcal{S}_{\mathrm{Total}}$
immediately after the measurement $\mathcal{M}^n$ is given by
\begin{equation}\label{eq:Ui-TSCP-1ton}
\begin{split}
  \ket{\Psi^n_{\mathrm{Total}}}
  =\sum_{(m_1,\dots,m_n)\in\Omega_1\times\dots\times\Omega_n}
  (M^{1,\dots,n}_{m_1,\dots,m_n}\ket{\Psi^0_{\mathrm{init}}})
  \otimes\ket{\Phi^1[m_1]}\otimes\dots\otimes\ket{\Phi^n[m_n]}.
\end{split}
\end{equation}

In our framework based on the \emph{principle of typicality},
the unitary operator $U_{\mathrm{Total}}^{1,\dots,n}$ applying to
the initial state $\ket{\Psi^0_{\mathrm{init}}}\otimes\ket{\Phi^1_{\mathrm{init}}}\otimes\dots\otimes\ket{\Phi^n_{\mathrm{init}}}$ describes
the \emph{repeated once} of the infinite repetition of the measurements
where the succession of
the measurements $\mathcal{M}^1,\dots,\mathcal{M}^n$
in this order
is infinitely repeated.
It follows from \eqref{eq:Ui-TSCP-1ton} that
the application of $U_{\mathrm{Total}}^{1,\dots,n}$
can be regarded as a
\emph{single measurement} which is described by
the measurement operators
$\{M^{1,\dots,n}_{m_1,\dots,m_n}\}_{(m_1,\dots,m_n)\in\Omega_1\times\dots\times\Omega_n}$
and whose all possible outcomes form the set $\Omega_1\times\dots\times\Omega_n$.
Hence, we can apply Definition~\ref{pmrpwst}
to this scenario of the setting of measurements.
Thus,
according to Definition~\ref{pmrpwst},
we can see that a \emph{world} is an infinite sequence over $\Omega_1\times\dots\times\Omega_n$
and the probability measure induced by
the \emph{probability measure representation for the prefixes of worlds}
is a Bernoulli measure $\lambda_P$ on
$(\Omega_1\times\dots\times\Omega_n)^\infty$,
where $P$ is a finite probability space on $\Omega_1\times\dots\times\Omega_n$
such that $P(m_1,\dots,m_n)$ is the square of the norm of
the vector
\begin{equation*}%
  (M^{1,\dots,n}_{m_1,\dots,m_n}\ket{\Psi^0_{\mathrm{init}}})
  \otimes\ket{\Phi^1[m_1]}\otimes\dots\otimes\ket{\Phi^n[m_n]}
\end{equation*}
for every $(m_1,\dots,m_n)\in\Omega_1\times\dots\times\Omega_n$.
Here $\Omega_1\times\dots\times\Omega_n$ is the set of all possible records of
the apparatuses $\mathcal{A}_1,\dots,\mathcal{A}_n$
in the \emph{repeated once} of the experiments.

Thus, we can apply
Postulate~\ref{POT}, the \emph{principle of typicality},
to the setting of measurements developed above.
Let $\omega$ be \emph{our world} in the infinite repetition of
the measurement described by the measurement operators
$\{M^{1,\dots,n}_{m_1,\dots,m_n}\}_{(m_1,\dots,m_n)\in\Omega_1\times\dots\times\Omega_n}$
in the above setting.
This $\omega$ is an infinite sequence over $\Omega_1\times\dots\times\Omega_n$
consisting of records
in the apparatuses $\mathcal{A}_1,\dots,\mathcal{A}_n$
which is being generated by the infinite repetition of the measurement
described by the measurement operators
$\{M^{1,\dots,n}_{m_1,\dots,m_n}\}_{(m_1,\dots,m_n)\in\Omega_1\times\dots\times\Omega_n}$
in the above setting.
Since the Bernoulli measure $\lambda_P$ on 
$(\Omega_1\times\dots\times\Omega_n)^\infty$ is
the probability measure induced by the probability measure representation
for the prefixes of worlds in the above setting,
it follows from Postulate~\ref{POT} that $\omega$ is Martin-L\"{o}f random with respect to the Bernoulli measure $\lambda_{P}$, that is, \emph{$\omega$ is Martin-L\"of $P$-random}.

We use $\omega^1,\dots,\omega^n$
to denote the infinite sequences over $\Omega_1,\dots,\Omega_n$, respectively,
such that $(\omega^1(l),\dots,\omega^n(l))=\omega(l)$ for every $l\in\N^+$.
Then, it follows from Postulate~\ref{POT} and \eqref{eq:Ui-TSCP-1ton}
that, in our world $\omega$, for each $l\in\N^+$ the state of
the $l$th total system $\mathcal{S}_{\mathrm{Total}}$,
i.e., the $l$th copy of a composite system consisting of
the system $\mathcal{S}$ and the apparatuses $\mathcal{A}_1,\dots,\mathcal{A}_n$,
immediately after the measurement $\mathcal{M}^n$ is given by
\begin{equation}\label{eq:Ui-TSCP-nth-individual}
  (M^{1,\dots,n}_{\omega^1(l),\dots,\omega^n(l)}\ket{\Psi^0_{\mathrm{init}}})
  \otimes\ket{\Phi^1[\omega^1(l)]}\otimes\dots\otimes\ket{\Phi^n[\omega^n(l)]},
\end{equation}
up to the normalization factor.

Now, originally,
the state $\ket{\Psi^n_{\mathrm{Total}}}$
of the total system $\mathcal{S}_{\mathrm{Total}}$ is
a superposition of the vectors of the form
\begin{equation}\label{eq:Ui-TSCP-individual}
  (M^{1,\dots,n}_{m_1,\dots,m_n}\ket{\Psi^0_{\mathrm{init}}})
  \otimes\ket{\Phi^1[m_1]}\otimes\dots\otimes\ket{\Phi^n[m_n]}
\end{equation}
immediately after the measurement $\mathcal{M}^n$
according to \eqref{eq:Ui-TSCP-1ton}.
However, as we saw
in the above argument leading to the state \eqref{eq:Ui-TSCP-nth-individual},
the application of the principle of typicality
implies that 
the state of the total system $\mathcal{S}_{\mathrm{Total}}$ is actually
given by a specific vector of the form of \eqref{eq:Ui-TSCP-individual}
(up to the normalization factor)
immediately after the measurement $\mathcal{M}^n$
in each repetition of the
measurement
described by the measurement operators
$\{M^{1,\dots,n}_{m_1,\dots,m_n}\}_{(m_1,\dots,m_n)\in\Omega_1\times\dots\times\Omega_n}$.

In what follows, we study the problem to determine
the whole state or the respective states of the system $\mathcal{S}$
and the apparatuses $\mathcal{A}_1,\dots,\mathcal{A}_k$
in the total system $\mathcal{S}_{\mathrm{Total}}$
immediately after each measurement $\mathcal{M}^k$ with $1\le k<n$,
provided that
the state of the total system $\mathcal{S}_{\mathrm{Total}}$
immediately after the measurement $\mathcal{M}^n$ is specified as
a vector of the form
\[
  \ket{\Psi}\otimes\ket{\Phi^1[m_1]}\otimes\dots\otimes\ket{\Phi^n[m_n]}
\]
with $\ket{\Psi}\in\mathcal{H}_{\mathcal{S}}$ and
$(m_1,\dots,m_n)\in\Omega_1\times\dots\times\Omega_n$,
just like
the vector \eqref{eq:Ui-TSCP-nth-individual}.
In order to solve this problem, we must introduce some postulates.
For that purpose, we first introduce some terminology.

\subsection{Confirming points for the final states of apparatuses}

\begin{definition}\label{def:unchanged-state}
For any $j=1,\dots,n$, we say that 
\emph{the states of the system $\mathcal{S}$ are unchanged
after the measurement by the apparatus $\mathcal{A}_j$} if
for every $m_j\in\Omega_j$, every $\ket{\Psi}\in\mathcal{H}_{\mathcal{S}}$, and
every $\ket{\Psi_R}\in\mathcal{H}_{1}\otimes\dots\otimes\mathcal{H}_{j-1}$
there exists a vector
$\ket{\Theta_R}\in\mathcal{H}_{1}\otimes\dots\otimes\mathcal{H}_{j-1}$ such that
\begin{equation*}
  M^j_{m_j}(\ket{\Psi}\otimes\ket{\Psi_R})=\ket{\Psi}\otimes\ket{\Theta_R}.
\end{equation*}

For any $j=1,\dots,n$, we say that
\emph{the states of the system $\mathcal{S}$ are
confirmed before the measurement by the apparatus $\mathcal{A}_j$}
if the states of the system $\mathcal{S}$ are
unchanged after the measurement by the apparatus $\mathcal{A}_k$ for all $k\ge j$.
\qed
\end{definition}

\begin{definition}[Confirming point for the states of system]
\label{def:confirming-point-state}
The notion of the
\emph{confirming point for the states of the system $\mathcal{S}$}
is defined as follows:
Let $V$ be the set of all $j\in\{1,\dots,n\}$
such that the states of the system $\mathcal{S}$ are confirmed
before the measurement by the apparatus $\mathcal{A}_j$.
If $V\neq\emptyset$, then the
\emph{confirming point for the states of the system $\mathcal{S}$}
is defined as the system $\mathcal{S}$ itself if $j=1$ and the apparatus $\mathcal{A}_{j-1}$
otherwise, where $j:=\min V$.
If $V=\emptyset$, then the
\emph{confirming point for the states of the system $\mathcal{S}$}
is defined as the apparatus $\mathcal{A}_{n}$.
\qed
\end{definition}

It is then easy to see that
the notion of the confirming point for the states of the system $\mathcal{S}$
can be equivalently
characterized as follows:
Let $C$ be the set of all $k\in\{1,\dots,n\}$ such that
the states of the system $\mathcal{S}$ are \emph{not}
unchanged after the measurement by the apparatus $\mathcal{A}_k$.
If $C\neq\emptyset$, then the confirming point for the states of the system $\mathcal{S}$
is the apparatus $\mathcal{A}_{l}$ with $l:=\max C$.
If $C=\emptyset$, then the confirming point for the states of the system $\mathcal{S}$
is the system $\mathcal{S}$ itself.

\begin{definition}\label{def:unchanged-apparatus}
For any $i,j=1,\dots,n$ with $i<j$, we say that 
\emph{the final states of the apparatus $\mathcal{A}_i$ are
unchanged after the measurement by
the apparatus $\mathcal{A}_j$} if for every $m_i\in\Omega_i$, every $m_j\in\Omega_j$,
every $\ket{\Psi_L}\in\mathcal{H}_{\mathcal{S}}\otimes\mathcal{H}_1\otimes\dots\otimes\mathcal{H}_{i-1}$,
and every $\ket{\Psi_R}\in\mathcal{H}_{i+1}\otimes\dots\otimes\mathcal{H}_{j-1}$
there exist a set $\{\ket{\Theta_L[l]}\}_l$ of vectors in
$\mathcal{H}_{\mathcal{S}}\otimes\mathcal{H}_1\otimes\dots\otimes\mathcal{H}_{i-1}$
and a set $\{\ket{\Theta_R[l]}\}_l$ of vectors in
$\mathcal{H}_{i+1}\otimes\dots\otimes\mathcal{H}_{j-1}$ such that
\begin{equation*}
  M^j_{m_j}(\ket{\Psi_L}\otimes\ket{\Phi^i[m_i]}\otimes\ket{\Psi_R})
  =\sum_l \ket{\Theta_L[l]}\otimes \ket{\Phi^i[m_i]}\otimes \ket{\Theta_R[l]}.
\end{equation*}

For any $i,j=1,\dots,n$ with $i<j$, we say that
\emph{the final states of the apparatus $\mathcal{A}_i$ are confirmed
before the measurement by the apparatus $\mathcal{A}_j$}
if the final states of the apparatus $\mathcal{A}_i$ are
unchanged after the measurement by the apparatus $\mathcal{A}_k$ for all $k\ge j$.
\qed
\end{definition}

\begin{definition}[Confirming point for the final states of apparatus]
\label{def:confirming-point-apparatus}
Let $i\in\{1,\dots,n\}$.
The notion of the
\emph{confirming point for the final states of the apparatus $\mathcal{A}_i$}
is defined as follows:
Let $V_i$ be the set of all $j\in\{i+1,\dots,n\}$
such that the final states of the apparatus $\mathcal{A}_i$ are confirmed
before the measurement by the apparatus $\mathcal{A}_j$.
If $V_i\neq\emptyset$, then the
\emph{confirming point for the final states of the apparatus $\mathcal{A}_i$}
is defined as the apparatus $\mathcal{A}_{j-1}$ with $j:=\min V_i$.
If $V_i=\emptyset$, then the
\emph{confirming point for the final states of the apparatus $\mathcal{A}_i$}
is defined as the apparatus $\mathcal{A}_{n}$.
Thus, the confirming point for the final states of the apparatus $\mathcal{A}_n$
is the apparatus $\mathcal{A}_{n}$ itself.
\qed
\end{definition}

It is then easy to see that
the notion of the confirming point for the final states of an apparatus
can be equivalently
characterized as follows:
Let $i\in\{1,\dots,n\}$, and let $C_i$ be the set of all $k\in\{i+1,\dots,n\}$ such that
the final states of the apparatus $\mathcal{A}_i$ are \emph{not} unchanged
after the measurement by the apparatus $\mathcal{A}_k$.
If $C_i\neq\emptyset$,
then the confirming point for the final states of the apparatus $\mathcal{A}_i$
is the apparatus $\mathcal{A}_{l}$ with $l:=\max C_i$.
If $C_i=\emptyset$,
then the confirming point for the final states of the apparatus $\mathcal{A}_i$
is the apparatus $\mathcal{A}_{i}$ itself.

Now, let us consider the case where there are only two apparatuses $\mathcal{A}_1$ and
$\mathcal{A}_2$ with $n=2$
and the final states of the apparatus $\mathcal{A}_1$ are \emph{not}
unchanged after the measurement by the apparatus $\mathcal{A}_2$.
In this case,
we cannot think the measurement $\mathcal{M}^1$ performed by the apparatus $\mathcal{A}_1$
to be completed immediately after the measurement $\mathcal{M}^1$ itself,
since the state of the apparatus $\mathcal{A}_1$ is disturbed
by the subsequent measurement $\mathcal{M}^2$ performed by the apparatus $\mathcal{A}_2$.
However,
it follows from \eqref{eq:Ui-TSCP-1tok} that
\begin{equation*}%
\begin{split}
  U_{\mathrm{Total}}^{1,2}(\ket{\Psi}\otimes\ket{\Phi^1_{\mathrm{init}}}\otimes\ket{\Phi^2_{\mathrm{init}}})
  =\sum_{(m_1,m_2)\in\Omega_1\times\Omega_2}
  (M^{1,2}_{m_1,m_2}\ket{\Psi})\otimes\ket{\Phi^1[m_1]}\otimes\ket{\Phi^2[m_2]}
\end{split}
\end{equation*}
for every $\ket{\Psi}\in\mathcal{H}_\mathcal{S}$.
This implies that the combination of the measurement $\mathcal{M}^1$ performed by the apparatus $\mathcal{A}_1$ with its subsequent measurement $\mathcal{M}^2$ performed by the apparatus $\mathcal{A}_2$ can be regarded as a \emph{single measurement},
which is described by the measurement operators
$\{M^{1,2}_{m_1,m_2}\}_{(m_1,m_2)\in\Omega_1\times\Omega_2}$
and whose all possible outcomes form the set $\Omega_1\times\Omega_2$.
Thus,
we can say that \emph{the measurement $\mathcal{M}^1$ performed by the apparatus $\mathcal{A}_1$ is really completed at the instant when the measurement $\mathcal{M}^2$ performed by the apparatus $\mathcal{A}_2$ is completed}.
The generalization of this idea to the general case
where there are $n$ apparatuses $\mathcal{A}_1,\dots,\mathcal{A}_n$ with an arbitrary positive integer $n$ 
is the notion of
the confirming point for the final states of the apparatus $\mathcal{A}_i$ with
an arbitrary integer $i\in\{1,\dots,n\}$,
introduced in Definition~\ref{def:confirming-point-apparatus} above.
We think that
\emph{the measurement by each apparatus $\mathcal{A}_i$ is really completed
at the confirming point for the final states of the apparatus $\mathcal{A}_i$ itself}.

Now,
to motivate Postulate~\ref{CF} below, we first prove the following proposition.

\begin{proposition}\label{prop:TSCP}
Without assuming  Postulate~\ref{POT},
the following~(i) and (ii) hold based on \eqref{eq:Ui-TSCP} and \eqref{eq:Ui-TSCP-MO}:
\begin{enumerate}
\item
Let $i\in\{1,\dots,n\}$.
For every $\ket{\Psi}\in\mathcal{H}_{\mathcal{S}}$,
if the state of the system $\mathcal{S}$ is $\ket{\Psi}$
immediately after the measurement $\mathcal{M}^n$
and the states of the system $\mathcal{S}$ are confirmed
before the measurement by the apparatus $\mathcal{A}_i$,
then the state of the system $\mathcal{S}$ is $\ket{\Psi}$
immediately before the measurement $\mathcal{M}^i$.
\item
Let $i,j\in\{1,\dots,n\}$ with $i<j$.
For every $m_i\in\Omega_i$,
if the state of the apparatus $\mathcal{A}_i$ is $\ket{\Phi^i[m_i]}$
immediately after the measurement $\mathcal{M}^n$
and the final states of the apparatus $\mathcal{A}_i$ are confirmed
before the measurement by the apparatus $\mathcal{A}_j$,
then the state of the apparatus $\mathcal{A}_i$ is $\ket{\Phi^i[m_i]}$
immediately before the measurement $\mathcal{M}^j$.
\end{enumerate}
\end{proposition}

\begin{proof}
(i) Let $i\in\{1,\dots,n\}$, and let $\ket{\Psi}\in\mathcal{H}_{\mathcal{S}}$.
Suppose that the state of the system $\mathcal{S}$ is $\ket{\Psi}$
immediately after the measurement $\mathcal{M}^n$
and the states of the system $\mathcal{S}$ are confirmed
before the measurement by the apparatus $\mathcal{A}_i$.
Then, in order to prove Proposition~\ref{prop:TSCP}~(i),
it is sufficient to show that, for every $k=i,\dots,n$,
if the state of the system $\mathcal{S}$ is $\ket{\Psi}$
immediately after the measurement $\mathcal{M}^{k}$, then
the state of the system $\mathcal{S}$ is $\ket{\Psi}$
immediately before the measurement $\mathcal{M}^{k}$.

Now, for an arbitrary $k=i,\dots,n$, assume that
the state of the system $\mathcal{S}$ is $\ket{\Psi}$
immediately after the measurement $\mathcal{M}^{k}$.
It follows from \eqref{eq:def-kPkT:=UT1tokkP0i-kP1i-kPni} and \eqref{eq:Ui-TSCP-1tok} that
the state $\ket{\Psi^{k-1}_{\mathrm{Total}}}$
of the total system $\mathcal{S}_{\mathrm{Total}}$
immediately before the measurement $\mathcal{M}^k$ has the form
\begin{equation}\label{eq:kPk-1T=kPk-1otkPki-otkPni}
  \ket{\Psi^{k-1}_{\mathrm{Total}}}
  =\ket{\Psi^{k-1}}\otimes\ket{\Phi^{k}_{\mathrm{init}}}\otimes\dots\otimes\ket{\Phi^n_{\mathrm{init}}}
\end{equation}
for some state $\ket{\Psi^{k-1}}\in\mathcal{H}_\mathcal{S}\otimes\mathcal{H}_1\otimes\dots\otimes\mathcal{H}_{k-1}$.
Let $\{\ket{\Psi[l]}\}_l$ be an orthonormal basis of $\mathcal{H}_\mathcal{S}$ such that
$\ket{\Psi[1]}=\ket{\Psi}$.
Then the vector $\ket{\Psi^{k-1}}$ can be written in the following form:
\begin{equation}\label{eq:Ui-TSCP-Prop1-state}
  \ket{\Psi^{k-1}}
  =\sum_l \ket{\Psi[l]}\otimes \ket{\Theta_R[l]},
\end{equation}
where $\{\ket{\Theta_R[l]}\}_{l}$ is a set of vectors in
$\mathcal{H}_{1}\otimes\dots\otimes\mathcal{H}_{k-1}$.
For each $l$, we use $\ket{\Gamma[l]}$
to denote the vector
\[
  \ket{\Psi[l]}\otimes \ket{\Theta_R[l]}.
\]
Since the states of the system $\mathcal{S}$ are
unchanged after the measurement by the apparatus $\mathcal{A}_{k}$,
we have that
there exists a set $\{\ket{\Upsilon_R[l,m_k]}\}_{l,m_k}$ of vectors in
$\mathcal{H}_{1}\otimes\dots\otimes\mathcal{H}_{k-1}$ such that,
for every $l$ and every $m_{k}\in\Omega_{k}$,
\begin{equation}\label{eq:TSCP-prop1-Mkmk}
  M^{k}_{m_{k}}\ket{\Gamma[l]}
  =\ket{\Psi[l]}\otimes \ket{\Upsilon_R[l,m_k]}.
\end{equation}
On the other hand, it follows from \eqref{eq:def-kPkT:=UT1tokkP0i-kP1i-kPni} and \eqref{eq:def-UT1tok:=UkotI-U1otI} that the state $\ket{\Psi^{k}_{\mathrm{Total}}}$
of the total system $\mathcal{S}_{\mathrm{Total}}$
immediately after the measurement $\mathcal{M}^k$
has the following form:
\begin{equation}\label{eq:TSCP-prop1-Pkt}
\begin{split}
  \ket{\Psi^{k}_{\mathrm{Total}}}
  &=(U_{k}\otimes I_{k+1,\dots,n})\ket{\Psi^{k-1}_{\mathrm{Total}}} \\
  &=(U_k(\ket{\Psi^{k-1}}\otimes\ket{\Phi^k_{\mathrm{init}}}))
  \otimes\ket{\Phi^{k+1}_{\mathrm{init}}}\otimes\dots\otimes\ket{\Phi^n_{\mathrm{init}}} \\
  &=\sum_{m_k\in\Omega_k}(M^k_{m_k}\ket{\Psi^{k-1}})\otimes\ket{\Phi^k[m_k]}
  \otimes\ket{\Phi^{k+1}_{\mathrm{init}}}\otimes\dots\otimes\ket{\Phi^n_{\mathrm{init}}} \\
  &=\sum_{l}\sum_{m_k\in\Omega_k}(M^k_{m_k}\ket{\Gamma[l]})\otimes\ket{\Phi^k[m_k]}
  \otimes\ket{\Phi^{k+1}_{\mathrm{init}}}\otimes\dots\otimes\ket{\Phi^n_{\mathrm{init}}},
\end{split}
\end{equation}
where the second equality follows from \eqref{eq:kPk-1T=kPk-1otkPki-otkPni},
the third equality follows from \eqref{eq:Ui-TSCP}, and
the last equality follows from \eqref{eq:Ui-TSCP-Prop1-state}.
Since the state of the system $\mathcal{S}$ is $\ket{\Psi}$
immediately after the measurement $\mathcal{M}^{k}$ by the assumption,
it follows from \eqref{eq:TSCP-prop1-Pkt} and \eqref{eq:TSCP-prop1-Mkmk} that $M^{k}_{m_{k}}\ket{\Gamma[l]}=0$
for every $m_k\in\Omega_k$ and every $l\ge 2$.
Thus, using the completeness equation
\[
  \sum_{m_k\in\Omega_k} {M^{k}_{m_{k}}}^\dag M^{k}_{m_{k}}=I_{0,1,\dots,k-1},
\]
we have that $\ket{\Gamma[l]}=0$ for every $l\ge 2$.
Therefore, it follows from \eqref{eq:kPk-1T=kPk-1otkPki-otkPni} and \eqref{eq:Ui-TSCP-Prop1-state} that
the state $\ket{\Psi^{k-1}_{\mathrm{Total}}}$
of the total system $\mathcal{S}_{\mathrm{Total}}$
immediately before the measurement $\mathcal{M}^k$ has the form
\[
  \ket{\Psi^{k-1}_{\mathrm{Total}}}
  =\ket{\Psi}\otimes \ket{\Theta_R[1]}
  \otimes\ket{\Phi^{k}_{\mathrm{init}}}\otimes\dots\otimes\ket{\Phi^n_{\mathrm{init}}}.
\]
Hence, the state of the system $\mathcal{S}$ is $\ket{\Psi}$
immediately before the measurement $\mathcal{M}^{k}$, as desired.

(ii) Let $i,j\in\{1,\dots,n\}$ with $i<j$, and let $m_i\in\Omega_i$.
Suppose that the state of the apparatus $\mathcal{A}_i$ is $\ket{\Phi^i[m_i]}$
immediately after the measurement $\mathcal{M}^n$
and the final states of the apparatus $\mathcal{A}_i$ are confirmed
before the measurement by the apparatus $\mathcal{A}_j$.
Then, in order to prove Proposition~\ref{prop:TSCP}~(ii),
it is sufficient to show that, for every $k=j,\dots,n$,
if the state of the apparatus $\mathcal{A}_i$ is $\ket{\Phi^i[m_i]}$
immediately after the measurement $\mathcal{M}^{k}$, then
the state of the apparatus $\mathcal{A}_i$ is $\ket{\Phi^i[m_i]}$
immediately before the measurement $\mathcal{M}^{k}$.

Now, for an arbitrary $k=j,\dots,n$, assume that
the state of the apparatus $\mathcal{A}_i$ is $\ket{\Phi^i[m_i]}$
immediately after the measurement $\mathcal{M}^{k}$.
It follows from \eqref{eq:def-kPkT:=UT1tokkP0i-kP1i-kPni} and \eqref{eq:Ui-TSCP-1tok} that
the state $\ket{\Psi^{k-1}_{\mathrm{Total}}}$
of the total system $\mathcal{S}_{\mathrm{Total}}$
immediately before the measurement $\mathcal{M}^k$ has the form
\begin{equation}\label{eq:kPk-1T=kPotkPki-otkPni}
  \ket{\Psi^{k-1}_{\mathrm{Total}}}
  =\ket{\Psi}\otimes\ket{\Phi^{k}_{\mathrm{init}}}\otimes\dots\otimes\ket{\Phi^n_{\mathrm{init}}}
\end{equation}
for some state $\ket{\Psi}\in\mathcal{H}_\mathcal{S}\otimes\mathcal{H}_1\otimes\dots\otimes\mathcal{H}_{k-1}$.
Then
the vector $\ket{\Psi}$ can be written in the following form:
\begin{equation}\label{eq:Ui-TSCP-Prop-state}
  \ket{\Psi}
  =\sum_{m\in\Omega_i}\sum_l
  \ket{\Theta_L[l]}\otimes \ket{\Phi^i[m]}\otimes \ket{\Theta_R[m,l]},
\end{equation}
where $\{\ket{\Theta_L[l]}\}_l$ is a set of vectors in
$\mathcal{H}_{\mathcal{S}}\otimes\mathcal{H}_1\otimes\dots\otimes\mathcal{H}_{i-1}$
and $\{\ket{\Theta_R[m,l]}\}_{m,l}$ is a set of vectors in
$\mathcal{H}_{i+1}\otimes\dots\otimes\mathcal{H}_{k-1}$.
For each $m\in\Omega_i$, we use $\ket{\Gamma[m]}$
to denote the vector
\[
  \sum_l\ket{\Theta_L[l]}\otimes \ket{\Phi^i[m]}\otimes \ket{\Theta_R[m,l]}.
\]
Since the final states of the apparatus $\mathcal{A}_i$ are
unchanged after the measurement by the apparatus $\mathcal{A}_{k}$,
we have that
there exist a set $\{\ket{\Upsilon_L[m,m_{k},h]}\}_{m,m_k,h}$ of vectors in
$\mathcal{H}_{\mathcal{S}}\otimes\mathcal{H}_1\otimes\dots\otimes\mathcal{H}_{i-1}$
and a set $\{\ket{\Upsilon_L[m,m_{k},h]}\}_{m,m_k,h}$ of vectors in
$\mathcal{H}_{i+1}\otimes\dots\otimes\mathcal{H}_{k-1}$ such that,
for every $m\in\Omega_i$ and every $m_{k}\in\Omega_{k}$,
\begin{equation}\label{eq:TSCP-prop2-Mkmk}
  M^{k}_{m_{k}}\ket{\Gamma[m]}
  =\sum_h \ket{\Upsilon_L[m,m_{k},h]}\otimes \ket{\Phi^i[m]}
  \otimes \ket{\Upsilon_R[m,m_{k},h]}.
\end{equation}
On the other hand, it follows from \eqref{eq:def-kPkT:=UT1tokkP0i-kP1i-kPni} and \eqref{eq:def-UT1tok:=UkotI-U1otI} that the state $\ket{\Psi^{k}_{\mathrm{Total}}}$
of the total system $\mathcal{S}_{\mathrm{Total}}$
immediately after the measurement $\mathcal{M}^k$
has the following form:
\begin{equation}\label{eq:TSCP-prop2-Pkt}
\begin{split}
  \ket{\Psi^{k}_{\mathrm{Total}}}
  &=(U_{k}\otimes I_{k+1,\dots,n})\ket{\Psi^{k-1}_{\mathrm{Total}}} \\
  &=(U_k(\ket{\Psi}\otimes\ket{\Phi^k_{\mathrm{init}}}))
  \otimes\ket{\Phi^{k+1}_{\mathrm{init}}}\otimes\dots\otimes\ket{\Phi^n_{\mathrm{init}}} \\
  &=\sum_{m_k\in\Omega_k}(M^k_{m_k}\ket{\Psi})\otimes\ket{\Phi^k[m_k]}
  \otimes\ket{\Phi^{k+1}_{\mathrm{init}}}\otimes\dots\otimes\ket{\Phi^n_{\mathrm{init}}} \\
  &=\sum_{m\in\Omega_i}\sum_{m_k\in\Omega_k}(M^k_{m_k}\ket{\Gamma[m]})\otimes\ket{\Phi^k[m_k]}
  \otimes\ket{\Phi^{k+1}_{\mathrm{init}}}\otimes\dots\otimes\ket{\Phi^n_{\mathrm{init}}},
\end{split}
\end{equation}
where the second equality follows from \eqref{eq:kPk-1T=kPotkPki-otkPni},
the third equality follows from \eqref{eq:Ui-TSCP}, and
the last equality follows from \eqref{eq:Ui-TSCP-Prop-state}.
Since the state of the apparatus $\mathcal{A}_i$ is $\ket{\Phi^i[m_i]}$
immediately after the measurement $\mathcal{M}^{k}$ by the assumption,
it follows from \eqref{eq:TSCP-prop2-Pkt} and \eqref{eq:TSCP-prop2-Mkmk} that $M^{k}_{m_{k}}\ket{\Gamma[m]}=0$
for every $m_k\in\Omega_k$ and every $m\in\Omega_i$ with $m\neq m_i$.
Thus, using the completeness equation
\[
  \sum_{m_k\in\Omega_k} {M^{k}_{m_{k}}}^\dag M^{k}_{m_{k}}=I_{0,1,\dots,k-1},
\]
we have that $\ket{\Gamma[m]}=0$ for every $m\in\Omega_i$ with $m\neq m_i$.
Therefore, it follows from \eqref{eq:kPk-1T=kPotkPki-otkPni} and \eqref{eq:Ui-TSCP-Prop-state} that
the state $\ket{\Psi^{k-1}_{\mathrm{Total}}}$
of the total system $\mathcal{S}_{\mathrm{Total}}$
immediately before the measurement $\mathcal{M}^k$ has the form
\[
  \ket{\Psi^{k-1}_{\mathrm{Total}}}
  =\sum_l \ket{\Theta_L[l]}\otimes \ket{\Phi^i[m_i]}\otimes \ket{\Theta_R[m_i,l]}
  \otimes\ket{\Phi^{k}_{\mathrm{init}}}\otimes\dots\otimes\ket{\Phi^n_{\mathrm{init}}}.
\]
Hence, the state of the apparatus $\mathcal{A}_i$ is $\ket{\Phi^i[m_i]}$
immediately before the measurement $\mathcal{M}^{k}$, as desired.
\end{proof}

\subsection{Postulates for determining the states of apparatuses in the intermediate stages of the whole measurement process}

In the preceding subsection, we have made an argument based on Postulate~\ref{evolution},
the unitary time-evolution,
and not based on Postulate~\ref{POT}, the principle of typicality.
The application of Postulate~\ref{POT}
results in
the \emph{non-unitary} time-evolution
in the measurement process,
as we saw in Section~\ref{subsec:The general setting of the problem}
as well as in Section~\ref{sec:Analysis of Schroedinger's cat} in the context of the Schr\"{o}dinger's cat.
In the rest of this paper,
we regard the states $\ket{\Psi^1_{\mathrm{Total}}},\ket{\Psi^2_{\mathrm{Total}}},\dots,\ket{\Psi^n_{\mathrm{Total}}}$
of the total system $\mathcal{S}_{\mathrm{Total}}$
immediately after the measurements
$\mathcal{M}^{1},\mathcal{M}^{2},\dots,\mathcal{M}^{n}$, respectively,
as \emph{virtual} states
which provide a useful means for the calculation
in our framework based on
the principle of typicality,
whereas
the initial state
$\ket{\Psi^0_{\mathrm{init}}}\otimes\ket{\Phi^1_{\mathrm{init}}}\otimes\dots\otimes\ket{\Phi^n_{\mathrm{init}}}$
of the total system $\mathcal{S}_{\mathrm{Total}}$ immediately before the measurement $\mathcal{M}^{1}$ is treated
as a \emph{real} state.
Actually, the state $\ket{\Psi^n_{\mathrm{Total}}}$
of the total system $\mathcal{S}_{\mathrm{Total}}$
at the instant when the whole measurement consisting of $\mathcal{M}^{1},\mathcal{M}^{2},\dots,\mathcal{M}^{n}$ is completed has been regarded
as a virtual state in our framework based on
the principle of typicality,
which is developed in Section~\ref{MWI} and then applied in
Section~\ref{sec:Analysis of Schroedinger's cat} and
Section~\ref{subsec:The general setting of the problem}.

However,
suggested by Proposition~\ref{prop:TSCP},
we propose the following postulate for \emph{real} states of the system $\mathcal{S}$
and the apparatuses $\mathcal{A}_1,\dots,\mathcal{A}_{n-1}$
in the \emph{intermediate} stages of
the whole measurement process consisting of $\mathcal{M}^{1},\mathcal{M}^{2},\dots,\mathcal{M}^{n}$.

\begin{postulate}\label{CF}
\
\begin{enumerate}
\item
Let $i\in\{1,\dots,n\}$.
For every $\ket{\Psi}\in\mathcal{H}_{\mathcal{S}}$,
if the state of the system $\mathcal{S}$ is $\ket{\Psi}$
immediately after the measurement $\mathcal{M}^n$
and the states of the system $\mathcal{S}$ are confirmed
before the measurement by the apparatus $\mathcal{A}_i$,
then the state of the system $\mathcal{S}$ is $\ket{\Psi}$
immediately before the measurement $\mathcal{M}^i$.
\item
Let $i,j\in\{1,\dots,n\}$ with $i<j$.
For every $m_i\in\Omega_i$,
if the state of the apparatus $\mathcal{A}_i$ is $\ket{\Phi^i[m_i]}$
immediately after the measurement $\mathcal{M}^n$
and the final states of the apparatus $\mathcal{A}_i$ are confirmed
before the measurement by the apparatus $\mathcal{A}_j$,
then the state of the apparatus $\mathcal{A}_i$ is $\ket{\Phi^i[m_i]}$
immediately before the measurement $\mathcal{M}^j$.
\qed
\end{enumerate}
\end{postulate}

We remark the following:
The statement of Postulate~\ref{CF} appears to be the same
as that of Proposition~\ref{prop:TSCP} on the surface.
However, Proposition~\ref{prop:TSCP} assumes the unitary time-evolution
in the total measurement process
consisting of the measurements $\mathcal{M}^{1},\mathcal{M}^{2},\dots,\mathcal{M}^{n}$,
as we saw above.
In contrast, Postulate~\ref{CF} does not assume the unitary time-evolution
in the total measurement process.
Postulate~\ref{CF} plays a major role
particularly
in the analysis of Wigner's friend made in Section~\ref{sec:Analysis of Wigner's friend}.

Next, we consider the problem to determine
the whole state of the total system $\mathcal{S}_{\mathrm{Total}}$,
consisting of the system $\mathcal{S}$
and the apparatuses $\mathcal{A}_1,\dots,\mathcal{A}_n$,
immediately after the measurement $\mathcal{M}^{n-1}$,
given the whole state of the total system $\mathcal{S}_{\mathrm{Total}}$
immediately after the measurement $\mathcal{M}^n$.

Let $k\in\N$,
and let $i_1,\dots,i_k$ be integers such that $1\le i_1 <\dots<i_k\le n$.
Let $\ket{\Phi^{i_1}},\dots,\ket{\Phi^{i_k}}$ be
states of the apparatuses $\mathcal{A}_{i_1},\dots,\mathcal{A}_{i_k}$, respectively.
In the case of $k\ge 1$,
we define $$P(\ket{\Phi^{i_1}},\dots,\ket{\Phi^{i_k}})$$ as a
projector
onto the subspace of the state space of the total system $\mathcal{S}_{\mathrm{Total}}$
where the states of the apparatuses $\mathcal{A}_{i_1},\dots,\mathcal{A}_{i_k}$ are
determined as $\ket{\Phi^{i_1}},\dots,\ket{\Phi^{i_k}}$, respectively.
Formally, in the case of $k\ge 1$, $P(\ket{\Phi^{i_1}},\dots,\ket{\Phi^{i_k}})$ is defined as
the
projector
acting on the state space
$\mathcal{H}_\mathcal{S}\otimes\overline{\mathcal{H}}_1\otimes\dots\otimes\overline{\mathcal{H}}_n$
of the total system $\mathcal{S}_{\mathrm{Total}}$ of the form
\begin{align*}
  I_{0,1,\dots,i_1-1}&\otimes\product{\Phi^{i_1}}{\Phi^{i_1}}\otimes
  I_{i_1+1,\dots,i_2-1}\otimes\product{\Phi^{i_2}}{\Phi^{i_2}}
  \otimes I_{i_2+1,\dots,i_3-1}\otimes\dotsm\dotsm \\
  &\dotsm\dotsm
  \otimes I_{i_{k-2}+1,\dots,i_{k-1}-1}\otimes\product{\Phi^{i_{k-1}}}{\Phi^{i_{k-1}}}
  \otimes I_{i_{k-1}+1,\dots,i_k-1}\otimes\product{\Phi^{i_k}}{\Phi^{i_k}}
  \otimes I_{i_k+1,\dots,n}.
\end{align*}
In the case of $k=0$, we define
$P(\ket{\Phi^{i_1}},\dots,\ket{\Phi^{i_k}})$, i.e., $P()$, as the identity operator
$I_{0,1,\dots,n}$ acting on the whole state space
$\mathcal{H}_\mathcal{S}\otimes\overline{\mathcal{H}}_1\otimes\dots\otimes\overline{\mathcal{H}}_n$
of the total system $\mathcal{S}_{\mathrm{Total}}$.

\begin{postulate}\label{Recursive Use}
Suppose that the measurement operators $\{M^i_{m_i}\}_{m_i\in\Omega_i}$ form a PVM
in $\mathcal{H}_{\mathcal{S}}\otimes\overline{\mathcal{H}}_1\otimes\dots\otimes\overline{\mathcal{H}}_{i-1}$
for every $i=1,\dots,n$.
Suppose that the state of the total system $\mathcal{S}_{\mathrm{Total}}$,
consisting of the system $\mathcal{S}$
and the apparatuses $\mathcal{A}_1,\dots,\mathcal{A}_n$,
immediately after the measurement $\mathcal{M}^n$ is
\[
  \ket{\Psi}\otimes\ket{\Phi^1[m_1]}\otimes\dots\otimes\ket{\Phi^n[m_n]},
\]
where $m_1\in\Omega_1,\dots,m_n\in\Omega_n$ and
$\ket{\Psi}\in\mathcal{H}_{\mathcal{S}}$ is a state of the system $\mathcal{S}$.
Let $D$ be the set of $i\in\{1,\dots,n-1\}$ such that
the final states of the apparatus $\mathcal{A}_i$ are
unchanged after the measurement by the apparatus $\mathcal{A}_n$.
We define a finite sequence $i_1,i_2,\dots,i_k\in\{1,2,\dots,n-1\}$ by the condition that
$D=\{i_1,i_2,\dots,i_k\}$ and $i_1<i_2<\dots <i_k$.
Then the state of
the total system $\mathcal{S}_{\mathrm{Total}}$
immediately after the measurement $\mathcal{M}^{n-1}$ is given by
\begin{equation}\label{eq:Postulate-RS-Proj}
  P(\ket{\Phi^{i_1}[m_{i_1}]},\dots,\ket{\Phi^{i_k}[m_{i_k}]})\ket{\Psi^{n-1}_{\mathrm{Total}}},
\end{equation}
up to the normalization factor.
Note that
this state is just $\ket{\Psi^{n-1}_{\mathrm{Total}}}$
in the case of $D=\emptyset$, i.e.,
in the case of
$k=0$.
\qed
\end{postulate}

Postulate~\ref{Recursive Use} only states the relation between
the state of the total system $\mathcal{S}_{\mathrm{Total}}$
immediately after the measurement $\mathcal{M}^{n}$ and
the state of the total system $\mathcal{S}_{\mathrm{Total}}$
immediately after the measurement $\mathcal{M}^{n-1}$.
However,
as we will see in
Section~\ref{subsec:the Wigner-Deutsch collaboration_Application of the principle of typicality},
we can use Postulate~\ref{Recursive Use} in a \emph{recursive manner}
to state the relation between
the state of the total system $\mathcal{S}_{\mathrm{Total}}$
immediately after the measurement $\mathcal{M}^{n}$ and
the state of the total system $\mathcal{S}_{\mathrm{Total}}$
immediately after the measurement $\mathcal{M}^{i}$ with $1\le i<n$.

In the rest of this section,
let us investigate immediate consequences of Postulate~\ref{CF} and Postulate~\ref{Recursive Use}, and thus investigate the validity of these postulates.
For that purpose, we
consider the normal case where,
for each $i=2,\dots,n$,
the apparatus $\mathcal{A}_i$ performs the measurement $\mathcal{M}^i$
only over the system $\mathcal{S}$
and not over any of the apparatuses $\mathcal{A}_{1},\dots,\mathcal{A}_{i-1}$.
Thus,
we assume that, for every $i=2,\dots,n$ and every $m_i\in\Omega_i$,
the measurement operator $M^i_{m_i}$ has the following form:
\[
  M^i_{m_i}=L^i_{m_i}\otimes I_{1,\dots,i-1},
\]
where $\{L^i_{m_i}\}_{m_i\in\Omega_i}$ are measurement operators acting on
$\mathcal{H}_{\mathcal{S}}$.
It is then easy to see that, for every $i=1,\dots,n-1$,
the final states of the apparatus $\mathcal{A}_i$ are confirmed
before the measurement by the apparatus $\mathcal{A}_{i+1}$.
Thus, for every $i=1,\dots,n-1$,
the confirming point for the final states of the apparatus $\mathcal{A}_i$
is the apparatus $\mathcal{A}_{i}$ itself.

Suppose that the state of the total system $\mathcal{S}_{\mathrm{Total}}$,
consisting of the system $\mathcal{S}$
and the apparatuses $\mathcal{A}_1,\dots,\mathcal{A}_n$,
immediately after the measurement $\mathcal{M}^n$ is
\[
  \ket{\Psi}\otimes\ket{\Phi^1[m_1]}\otimes\dots\otimes\ket{\Phi^n[m_n]},
\]
where $m_1\in\Omega_1,\dots,m_n\in\Omega_n$ and
$\ket{\Psi}\in\mathcal{H}_{\mathcal{S}}$ is a state of the system $\mathcal{S}$.
Then, according to Postulate~\ref{CF} we have that,
for every $i=1,\dots,n-1$,
the state of the apparatus $\mathcal{A}_i$ is $\ket{\Phi^i[m_i]}$
immediately before the measurement $\mathcal{M}^{i+1}$,
that is,
the state of the apparatus $\mathcal{A}_i$ is $\ket{\Phi^i[m_i]}$
immediately after the measurement $\mathcal{M}^{i}$.
This result is comprehensively used in our former works
\emph{implicitly}
in order to make the operational refinements of quantum mechanics and quantum information theory
in our framework based on the principle of typicality \cite{T15QCOMPINFO,T15Kokyuroku,T16CCR,T16QIP,T16QIT35,T17SCIS,T18arXiv,T25CCR}.

Suppose further that
the measurement operators $\{M^i_{m_i}\}_{m_i\in\Omega_i}$ form a PVM
in $\mathcal{H}_{\mathcal{S}}\otimes\overline{\mathcal{H}}_1\otimes\dots\otimes\overline{\mathcal{H}}_{i-1}$
for every $i=1,\dots,n$.
Then, since the final states of the apparatus $\mathcal{A}_i$ are
unchanged after the measurement by the apparatus $\mathcal{A}_n$
for every $i=1,\dots,n-1$,
according to Postulate~\ref{Recursive Use} we have that
the state of the total system $\mathcal{S}_{\mathrm{Total}}$
immediately after the measurement $\mathcal{M}^{n-1}$ is given by
\begin{equation}\label{eq:Postulate-RS-Proj-All}
  P(\ket{\Phi^{1}[m_{1}]},\dots,\ket{\Phi^{n-1}[m_{n-1}]})\ket{\Psi^{n-1}_{\mathrm{Total}}},
\end{equation}
up to the normalization factor.
This result
is also comprehensively used in our former works \emph{implicitly}
in order to make the operational refinements of quantum mechanics and quantum information theory
in our framework based on the principle of typicality \cite{T15QCOMPINFO,T15Kokyuroku,T16CCR,T16QIP,T16QIT35,T17SCIS,T18arXiv,T25CCR}.
In particular,
the result~\eqref{eq:Postulate-RS-Proj-All}
implicitly played a crucial role in
the derivations of Postulate~8 of Tadaki~\cite{T18arXiv},
which is a refined rule of the Born rule for mixed states,
from the principle of typicality in several scenarios of the setting of measurements,
as described in Sections~11 and 12 of Tadaki~\cite{T18arXiv}.
Thus, this effectiveness of the result~\eqref{eq:Postulate-RS-Proj-All}
supports general use of Postulate~\ref{Recursive Use}
in a specific situation
such as Wigner's friend, Deutsch's thought experiment, and the Wigner-Deutsch collaboration, which will be investigated in the subsequent sections.

The reason why the measurement operators
$\{M^1_{m_1}\}_{m_1\in\Omega_1},\dots,\{M^n_{m_n}\}_{m_n\in\Omega_n}$
must be
restricted to PVMs in Postulate~\ref{Recursive Use}
is the following:
If we allow $\{M^1_{m_1}\}_{m_1\in\Omega_1},\dots,\{M^n_{m_n}\}_{m_n\in\Omega_n}$
to be general measurement operators
then Postulate~\ref{Recursive Use} becomes too strong
and results in a contradiction.
We present the details of this argument in
Remark~\ref{rem:Validity2-PRU} in Section~\ref{subsec:the Wigner-Deutsch collaboration_Application of the principle of typicality}.
For another justification of Postulate~\ref{Recursive Use},
see Remark~\ref{rem:Validity1-PRU}.

In the subsequent sections,
we apply Postulates~\ref{CF} and \ref{Recursive Use} to
several examples of
a general case where
the apparatus $\mathcal{A}_i$ performs the measurement $\mathcal{M}^i$
over the composite system consisting of the system $\mathcal{S}$
and the apparatuses $\mathcal{A}_{1},\dots,\mathcal{A}_{i-1}$
for each $i=1,\dots,n$.

\section{Wigner's friend}
\label{sec:Wigner's friend}

In this section,
we review the Wigner's friend paradox~\cite{W61}
\emph{in the terminology of the conventional quantum mechanics}.

Consider a single qubit system $\mathcal{S}$ with state space
$\mathcal{H}_\mathcal{S}$.
An observer $\mathcal{F}$ performs a measurement
$\mathcal{M}^\mathcal{F}$
described
by the measurement operators
$\{M^{\mathcal{F}}_0,M^{\mathcal{F}}_1\}$
over the system $\mathcal{S}$
(according to Postulate~\ref{Born-rule}),
where
\begin{equation}\label{eq:WF-MF0=00_MF1=11}
  M^{\mathcal{F}}_0:=\product{0}{0}\quad\text{and}\quad M^{\mathcal{F}}_1:=\product{1}{1},
\end{equation}
and $\ket{0}$ and $\ket{1}$
form
an orthonormal basis of
the state space $\mathcal{H}_\mathcal{S}$ of the
system $\mathcal{S}$.

Hereafter, we assume that
the observer $\mathcal{F}$ itself is
a quantum system with state space $\mathcal{H}_\mathcal{F}$.
Then we consider an observer $\mathcal{W}$ who is on the outside of
the composite system $\mathcal{S}+\mathcal{F}$
consisting of the system $\mathcal{S}$ and the observer $\mathcal{F}$. 
The observer $\mathcal{W}$ can also be regarded as
an
environment for
the composite system $\mathcal{S}+\mathcal{F}$.
Then, according to \eqref{single_measurement},
the measurement process of
the measurement
$\mathcal{M}^\mathcal{F}$
is described by a unitary operator $U_\mathcal{F}$
acting on
$\mathcal{H}_\mathcal{S}\otimes\mathcal{H}_\mathcal{F}$
such that
\begin{equation}\label{eq:U-MF-S}
  U_\mathcal{F}(\ket{\Psi}\otimes\ket{\Phi^{\mathcal{F}}_{\mathrm{init}}})
  =\sum_{k=0,1}(M^{\mathcal{F}}_k\ket{\Psi})\otimes\ket{\Phi^{\mathcal{F}}[k]}
\end{equation}
for every $\ket{\Psi}\in\mathcal{H}_\mathcal{S}$.
The unitary operator
$U_\mathcal{F}$ describes the interaction between the system $\mathcal{S}$ and
the observer $\mathcal{F}$ as a quantum system.
The vector $\ket{\Phi^{\mathcal{F}}_{\mathrm{init}}}\in\mathcal{H}_\mathcal{F}$ is
the initial state of the
observer
$\mathcal{F}$, and $\ket{\Phi^{\mathcal{F}}[k]}\in\mathcal{H}_\mathcal{F}$ is
a final state of the
observer
$\mathcal{F}$ for each $k=0,1$,
with $\braket{\Phi^{\mathcal{F}}[k]}{\Phi^{\mathcal{F}}[k']}=\delta_{k,k'}$.
For each $k=0,1$,
the state $\ket{\Phi^{\mathcal{F}}[k]}$ indicates that
\emph{the observer $\mathcal{F}$ records the value $k$}.
The unitary time-evolution \eqref{eq:U-MF-S} is the quantum mechanical description of
the measurement process of the measurement $\mathcal{M}^\mathcal{F}$
over
the composite system $\mathcal{S}+\mathcal{F}$,
from the point of view of the observer $\mathcal{W}$.

Immediately after the measurement $\mathcal{M}^\mathcal{F}$
performed
by $\mathcal{F}$ over the system $\mathcal{S}$,
the observer $\mathcal{W}$
performs a measurement $\mathcal{M}^\mathcal{W}$
described by the measurement operators
$\{M^{\mathcal{W}}_0,M^{\mathcal{W}}_1,M^{\mathcal{W}}_2\}$
over the system $\mathcal{F}$
(according to Postulate~\ref{Born-rule}),
where
\begin{equation}\label{eq:WF-MW_0_1_2}
\begin{split}
  M^{\mathcal{W}}_0&:=\product{\Phi^{\mathcal{F}}[0]}{\Phi^{\mathcal{F}}[0]},\\
  M^{\mathcal{W}}_1&:=\product{\Phi^{\mathcal{F}}[1]}{\Phi^{\mathcal{F}}[1]},\\
  M^{\mathcal{W}}_2&:=\sqrt{I_\mathcal{F}-{M^{\mathcal{W}}_0}^\dag M^{\mathcal{W}}_0-{M^{\mathcal{W}}_1}^\dag M^{\mathcal{W}}_1}
  =I_\mathcal{F}-M^{\mathcal{W}}_0-M^{\mathcal{W}}_1,
\end{split}
\end{equation}
and $I_\mathcal{F}$
denotes
the identity operator acting on $\mathcal{H}_\mathcal{F}$.

Now, let $c_0$ and $c_1$ be \emph{arbitrary} two non-zero complex numbers such that $\abs{c_0}^2+\abs{c_1}^2=1$, and we assume that
the system $\mathcal{S}$ is initially in a state $\ket{+}$
in the above setting,
where $\ket{+}$ is defined by
\begin{equation}\label{eq:WF-k+=f1s2k0+k1}
  \ket{+}:=c_0\ket{0}+c_1\ket{1}.
\end{equation}
Then, we have the following two cases,
depending on how we treat the measurement $\mathcal{M}^\mathcal{F}$:

\paragraph{Case 1.}

\emph{Applying Postulate~\ref{Born-rule} to the measurement $\mathcal{M}^\mathcal{F}$},
either the following two possibilities (i) or (ii) occurs
immediately after the measurement $\mathcal{M}^\mathcal{F}$:
\begin{enumerate}
\item The system $\mathcal{S}$ and the observer $\mathcal{F}$ are in the states
  $\ket{0}$ and $\ket{\Phi^{\mathcal{F}}[0]}$, respectively.
  Therefore,
  the state of the composite system $\mathcal{S}+\mathcal{F}$
  is
  $\ket{0}\otimes\ket{\Phi^{\mathcal{F}}[0]}$.
\item The system $\mathcal{S}$ and the observer $\mathcal{F}$ are in the states
  $\ket{1}$ and $\ket{\Phi^{\mathcal{F}}[1]}$, respectively.
  Therefore,
  the state of the composite system
  $\mathcal{S}+\mathcal{F}$
  is
  $\ket{1}\otimes\ket{\Phi^{\mathcal{F}}[1]}$.
\end{enumerate}
Here, the possibilities (i) and (ii) occur with probabilities $\abs{c_0}^2$ and $\abs{c_1}^2$, respectively.

\paragraph{Case 2.}

\emph{According to the unitary time-evolution~\eqref{eq:U-MF-S}},
the state of the composite system $\mathcal{S}+\mathcal{F}$
consisting of $\mathcal{S}$ and $\mathcal{F}$ is
\begin{equation}\label{eq:SFstateW}
  c_0\ket{0}\otimes\ket{\Phi^{\mathcal{F}}[0]}+c_1\ket{1}\otimes\ket{\Phi^{\mathcal{F}}[1]}
\end{equation}
immediately after the measurement $\mathcal{M}^\mathcal{F}$.

\bigskip

Thus,
the state~\eqref{eq:SFstateW} of the composite system $\mathcal{S}+\mathcal{F}$ in Case~2
is different from neither $\ket{0}\otimes\ket{\Phi^{\mathcal{F}}[0]}$ in the possibility~(i)
nor $\ket{1}\otimes\ket{\Phi^{\mathcal{F}}[1]}$ in the possibility~(ii) in Case~1,
immediately after the measurement $\mathcal{M}^\mathcal{F}$.
This inconsistency is the content of the paradox in the Wigner's friend paradox~\cite{W61}.
Note that, even if the composite system
$\mathcal{S}+\mathcal{F}$
is in the state \eqref{eq:SFstateW}
immediately
after the measurement $\mathcal{M}^\mathcal{F}$,
according to Postulate~\ref{Born-rule}
the subsequent measurement $\mathcal{M}^\mathcal{W}$ results in
one of the two possibilities (i) and (ii)
above
with probabilities $\abs{c_0}^2$ and $\abs{c_1}^2$, respectively,
for the
post-measurement
state of the composite system
$\mathcal{S}+\mathcal{F}$,
just as in Case~1.

\section{Analysis of Wigner's friend based on the principle of typicality}
\label{sec:Analysis of Wigner's friend}

In this section, we make an analysis of the Wigner's friend paradox,
which is described in the preceding section in the terminology of the conventional quantum mechanics,
in terms of our framework of quantum mechanics based on
Postulates~\ref{POT}, \ref{CF}, and \ref{Recursive Use},
together with Postulates~\ref{state_space}, \ref{composition}, and \ref{evolution}.
We will then see that the measurement $\mathcal{M}^\mathcal{F}$
settles
the composite system
$\mathcal{S}+\mathcal{F}$
down into one of the two possibilities (i) and (ii) in Case~1 above,
and not into the `undecided' state \eqref{eq:SFstateW} in Case~2.

To proceed this program,
first of all
\emph{we have to implement everything, i.e.,
the
two
measurements $\mathcal{M}^\mathcal{F}$ and $\mathcal{M}^\mathcal{W}$
in the setting of the Wigner's friend paradox,
by unitary time-evolution}.
The measurement process of the measurement
$\mathcal{M}^\mathcal{F}$
is already described by the unitary operator $U_\mathcal{F}$
given in \eqref{eq:U-MF-S}.
For applying the principle of typicality,
the observer $\mathcal{W}$
is regarded
as a quantum system with state space $\mathcal{H}_\mathcal{W}$.
Then, according to \eqref{single_measurement},
the measurement process of
the measurement
$\mathcal{M}^\mathcal{W}$
is described by a unitary operator $U_\mathcal{W}$
acting on
$\mathcal{H}_\mathcal{F}\otimes\mathcal{H}_\mathcal{W}$
such that
\begin{equation}\label{eq:U-MW-F}
  U_\mathcal{W}(\ket{\Psi}\otimes\ket{\Phi^{\mathcal{W}}_{\mathrm{init}}})
  =\sum_{l=0,1,2}(M^{\mathcal{W}}_l\ket{\Psi})\otimes\ket{\Phi^{\mathcal{W}}[l]}
\end{equation}
for every $\ket{\Psi}\in\mathcal{H}_\mathcal{F}$.
The unitary operator
$U_\mathcal{W}$ describes the interaction
between the observer $\mathcal{F}$ and the observer $\mathcal{W}$
where both the observers are regarded as a quantum system.
The vector $\ket{\Phi^{\mathcal{W}}_{\mathrm{init}}}\in\mathcal{H}_\mathcal{W}$ is
the initial state of the
observer
$\mathcal{W}$, and $\ket{\Phi^{\mathcal{W}}[l]}\in\mathcal{H}_\mathcal{W}$ is
a final state of the
observer
$\mathcal{W}$ for each $l=0,1,2$,
with $\braket{\Phi^{\mathcal{W}}[l]}{\Phi^{\mathcal{W}}[l']}=\delta_{l,l'}$.
For each $l=0,1$,
the state $\ket{\Phi^{\mathcal{W}}[l]}$ indicates that
\emph{the observer $\mathcal{W}$ knows that
the observer $\mathcal{F}$ records the value $l$}.

We denote the set $\{0,1\}\times\{0,1,2\}$ by $\Omega$,
and we define a finite
collection
$\{M_{k,l}\}_{(k,l)\in\Omega}$ of operators acting on
the state space $\mathcal{H}_\mathcal{S}$ by
\begin{equation}\label{eq:POT-WF-Mkl=dlkMFk}
  M_{k,l}:=\delta_{k,l}M^{\mathcal{F}}_k.
\end{equation}
Then it follows from \eqref{eq:WF-MF0=00_MF1=11} that
\begin{equation*}
  \sum_{(k,l)\in\Omega} M_{k,l}^\dag M_{k,l}
  =\sum_{k=0,1}{M^{\mathcal{F}}_k}^\dag M^{\mathcal{F}}_k=I_\mathcal{S},
\end{equation*}
where $I_\mathcal{S}$ denotes
the identity operator acting on $\mathcal{H}_\mathrm{S}$.
Thus, the finite collection $\{M_{k,l}\}_{(k,l)\in\Omega}$ satisfies the \emph{completeness equation}, and therefore forms \emph{measurement operators}.
We denote the identity operator acting on $\mathcal{H}_\mathrm{W}$ by $I_\mathcal{W}$,
and we define a unitary operator $U_{\mathrm{whole}}$ acting on $\mathcal{H}_\mathcal{S}\otimes\mathcal{H}_\mathcal{F}\otimes\mathcal{H}_\mathcal{W}$
by
\[
  U_{\mathrm{whole}}:=(I_\mathcal{S}\otimes U_{\mathcal{W}})\circ(U_{\mathcal{F}}\otimes I_\mathcal{W}).
\]
Then, using \eqref{eq:U-MF-S}, \eqref{eq:U-MW-F}, \eqref{eq:WF-MW_0_1_2}, and \eqref{eq:POT-WF-Mkl=dlkMFk} we see that, as a whole, 
the sequential applications of $U_{\mathcal{F}}$ and $U_{\mathcal{W}}$
to a state $\ket{\Psi}\otimes\ket{\Phi_{\mathrm{init}}}$ of the composite system $\mathcal{S}+\mathcal{F}+\mathcal{W}$
result in:
\begin{equation}\label{eq:WignerFriend-all}
\begin{split}
  U_{\mathrm{whole}}(\ket{\Psi}\otimes\ket{\Phi_{\mathrm{init}}})
  &=((I_\mathcal{S}\otimes U_{\mathcal{W}})\circ(U_{\mathcal{F}}\otimes I_\mathcal{W}))
      (\ket{\Psi}\otimes\ket{\Phi_{\mathrm{init}}}) \\
  &=\sum_{k=0,1}(M^{\mathcal{F}}_k\ket{\Psi})\otimes
      \sum_{l=0,1,2}(M^{\mathcal{W}}_l\ket{\Phi^{\mathcal{F}}[k]})
      \otimes\ket{\Phi^{\mathcal{W}}[l]} \\
  &=\sum_{k=0,1}(M^{\mathcal{F}}_k\ket{\Psi})\otimes\ket{\Phi^{\mathcal{F}}[k]}
      \otimes\ket{\Phi^{\mathcal{W}}[k]} \\
  &=\sum_{k=0,1}\,\sum_{l=0,1,2}(M_{k,l}\ket{\Psi})\otimes\ket{\Phi[k,l]} \\
  &=\sum_{(k,l)\in\Omega} (M_{k,l}\ket{\Psi})\otimes\ket{\Phi[k,l]}
\end{split}
\end{equation}
for each state $\ket{\Psi}\in\mathcal{H}_\mathcal{S}$ of the system $\mathcal{S}$,
where $\ket{\Phi_{\mathrm{init}}}:=
\ket{\Phi^{\mathcal{F}}_{\mathrm{init}}}\otimes\ket{\Phi^{\mathcal{W}}_{\mathrm{init}}}$
and
$\ket{\Phi[k,l]}:=\ket{\Phi^{\mathcal{F}}[k]}\otimes\ket{\Phi^{\mathcal{W}}[l]}$.

\subsection{Application of the principle of typicality}
\label{subsec:WF-Appl-POT}

Thus, in the setting of the Wigner's friend paradox,
the unitary operator $U_{\mathrm{whole}}$ applying to
the initial state $\ket{+}\otimes\ket{\Phi_{\mathrm{init}}}$ describes
the \emph{repeated once} of the infinite repetition of the measurements
where the succession of
the measurement $\mathcal{M}^\mathcal{F}$ and
the measurement $\mathcal{M}^\mathcal{W}$
is infinitely repeated.
It follows from \eqref{eq:WignerFriend-all} that
the application of $U_{\mathrm{whole}}$,
consisting of the sequential applications of $U_{\mathcal{F}}$ and $U_{\mathcal{W}}$ in this order,
can be regarded as a
\emph{single measurement} which is described by
the measurement operators $\{M_{k,l}\}_{(k,l)\in\Omega}$
and whose all possible outcomes form the set $\Omega$.

Hence, we can apply Definition~\ref{pmrpwst}
to this scenario of the setting of measurements.
Thus,
according to Definition~\ref{pmrpwst},
we can see that a \emph{world} is an infinite sequence over $\Omega$
and the probability measure induced by
the \emph{probability measure representation for the prefixes of worlds}
is a Bernoulli measure $\lambda_P$ on
$\Omega^\infty$,
where $P$ is a finite probability space on $\Omega$ such that
$P(k,l)$ is the square of the norm of the vector
\begin{equation*}
  (M_{k,l}\ket{+})\otimes\ket{\Phi[k,l]}
\end{equation*}
for every $(k,l)\in\Omega$.
Here $\Omega$ is the set of all possible records of
the observers $\mathcal{F}$ and $\mathcal{W}$
in the \emph{repeated once} of the experiments.
Let us calculate the explicit form of $P(k,l)$.
Then, using~\eqref{eq:POT-WF-Mkl=dlkMFk}, \eqref{eq:WF-MF0=00_MF1=11}, and \eqref{eq:WF-k+=f1s2k0+k1},
we have that
\begin{equation}\label{eq:WF-finitepsP}
  P(k,l)=\bra{+}M_{k,l}^\dag M_{k,l}\ket{+}\braket{\Phi[k,l]}{\Phi[k,l]}
  =\delta_{k,l}\bra{+}{M^{\mathcal{F}}_k}^\dag M^{\mathcal{F}}_k\ket{+}
  =\delta_{k,l}\abs{\braket{k}{+}}^2
  =\delta_{k,l}\abs{c_k}^2
\end{equation}
for each $(k,l)\in\Omega$.

Now, let us apply
Postulate~\ref{POT}, the \emph{principle of typicality},
to the setting of measurements
developed above.
Let $\omega$ be \emph{our world} in the infinite repetition of
the measurement described by the measurement operators $\{M_{k,l}\}_{(k,l)\in\Omega}$
in the above setting.
This $\omega$
is an infinite sequence over $\Omega$
consisting of records
in the two observers $\mathcal{F}$ and $\mathcal{W}$
which is being generated by the infinite repetition of the measurement
described by the measurement operators $\{M_{k,l}\}_{(k,l)\in\Omega}$
in the above setting.
Since the Bernoulli measure $\lambda_P$ on $\Omega^\infty$ is
the probability measure induced by the
probability
measure representation
for the prefixes of
worlds
in the above setting,
it follows from
Postulate~\ref{POT}
that $\omega$ is Martin-L\"{o}f random with respect to the Bernoulli measure $\lambda_{P}$, that is, \emph{$\omega$ is Martin-L\"of $P$-random}.

We use $\alpha$ and $\beta$
to denote the infinite sequences over $\{0,1\}$ and $\{0,1,2\}$, respectively,
such that $(\alpha(n),\beta(n))=\omega(n)$ for every $n\in\N^+$.
Then, since $\omega$ is Martin-L\"of $P$-random, it follows from \eqref{eq:WF-finitepsP} and
Corollary~\ref{cor:always-positive-probability}
that
\begin{equation}\label{eq:Wigner-bn=gn}
  \alpha(n)=\beta(n)
\end{equation}
for every $n\in\N^+$.
Thus, it follows from
Postulate~\ref{POT} and \eqref{eq:WignerFriend-all}
that, in our world $\omega$, for each $n\in\N^+$ the state of
the $n$th composite system $\mathcal{S}+\mathcal{F}+\mathcal{W}$,
i.e., the $n$th copy of a composite system consisting of
the system $\mathcal{S}$ and the two observers $\mathcal{F}$ and $\mathcal{W}$,
immediately after the measurement $\mathcal{M}^\mathcal{W}$ is given by
\begin{equation}\label{eq:WF-state-immafter-MW}
\begin{split}
  (M_{\alpha(n),\beta(n)}\ket{+})\otimes\ket{\Phi[\alpha(n),\beta(n)]}
  &=(M_{\alpha(n),\alpha(n)}\ket{+})\otimes\ket{\Phi^{\mathcal{F}}[\alpha(n)]}\otimes\ket{\Phi^{\mathcal{W}}[\alpha(n)]} \\
  &=(M^{\mathcal{F}}_{\alpha(n)}\ket{+})\otimes\ket{\Phi^{\mathcal{F}}[\alpha(n)]}\otimes\ket{\Phi^{\mathcal{W}}[\alpha(n)]} \\
  &=c_{\alpha(n)}\ket{\alpha(n)}\otimes\ket{\Phi^{\mathcal{F}}[\alpha(n)]}\otimes\ket{\Phi^{\mathcal{W}}[\alpha(n)]},
\end{split}
\end{equation}
up to the normalization factor,
where the second equality follows from \eqref{eq:POT-WF-Mkl=dlkMFk}, and the last equality follows from
\eqref{eq:WF-MF0=00_MF1=11} and \eqref{eq:WF-k+=f1s2k0+k1}.
Recall that,
for each $k=0,1$,
the state $\ket{\Phi^{\mathcal{F}}[k]}$ indicates that
\emph{the observer $\mathcal{F}$ records the value $k$},
and
the state $\ket{\Phi^{\mathcal{W}}[k]}$ indicates that
\emph{the observer $\mathcal{W}$ knows that
the observer $\mathcal{F}$ records the value $k$}.
Hence, we have that, in our world $\omega$,
\emph{immediately after the measurement $\mathcal{M}^\mathcal{W}$,
the value which the observer $\mathcal{F}$ records
equals the value which the observer $\mathcal{W}$ knows that
the observer $\mathcal{F}$ records,
in every repetition of the experiment}, independently of $n$,
as expected from the point of view of the conventional quantum mechanics.

Let us determine the states of
the composite system $\mathcal{S}+\mathcal{F}$
immediately after the measurement $\mathcal{M}^\mathcal{F}$
in our world $\omega$.
For that purpose, we can apply our general framework
to treat situations
where apparatuses perform measurements over other apparatuses as well as over a normal quantum system only being measured,
developed in Section~\ref{sec:confirming point},
to the setting of measurements in the Wigner's friend paradox,
developed above.
The current setting of measurements in the Wigner's friend paradox is treated within the general framework with $n=2$,
and the measurements $\mathcal{M}^\mathcal{F}$ and $\mathcal{M}^\mathcal{W}$ in the current setting serve as the measurements $\mathcal{M}^1$ and $\mathcal{M}^2$ in the general framework, respectively.
The single qubit system $\mathcal{S}$, and the observers $\mathcal{F}$ and $\mathcal{W}$ in the current setting serve as the system $\mathcal{S}$, and the apparatuses $\mathcal{A}_1$ and $\mathcal{A}_2$ in the general framework, respectively.
Then the measurement operators
$\{M^1_{m_1}\}_{m_1\in\Omega_1}$ and $\{M^2_{m_2}\}_{m_2\in\Omega_2}$
in the general framework are instantiated in the current setting
in the following manner:
\begin{enumerate}
  \item $\Omega_1=\{0,1\}$ and $\Omega_2=\{0,1,2\}$.
  \item $M^1_{m_1}=M^{\mathcal{F}}_{m_1}$ for every $m_1\in\Omega_1$, and
    $M^2_{m_2}=I_\mathcal{S}\otimes M^{\mathcal{W}}_{m_2}$ for every $m_2\in\Omega_2$.
\end{enumerate}
Therefore,
based on \eqref{eq:WF-MW_0_1_2},
it is easy to check that the condition~\eqref{eq:Ui-TSCP-MO} certainly holds
in the current setting of measurements in the Wigner's friend paradox.
Thus, based on \eqref{eq:WignerFriend-all}, we see that
the measurement operators $\{M^{1,2}_{m_1,m_2}\}_{(m_1,m_2)\in\Omega_1\times\Omega_2}$
satisfying \eqref{eq:Ui-TSCP-1tok} with $n=k=2$ in the general framework are instantiated in the current setting
in the manner that
\begin{equation*}
  M^{1,2}_{m_1,m_2}=M_{m_1,m_2}
\end{equation*}
for every $m_1\in\Omega_1$ and $m_2\in\Omega_2$,
where the measurement operators $\{M_{k,l}\}_{(k,l)\in\Omega}$ on the right-hand side are defined by \eqref{eq:POT-WF-Mkl=dlkMFk}.
The state $\ket{\Psi^0_{\mathrm{init}}}$ of the system $\mathcal{S}$
immediately before the measurement $\mathcal{M}^1$ in the general framework
materializes as
the initial state $\ket{+}$ of the single qubit system $\mathcal{S}$ in the current setting,
defined by \eqref{eq:WF-k+=f1s2k0+k1}.

Thus,
on the one hand,
according to Definition~\ref{def:unchanged-state},
it is easy to check that 
the states of the system $\mathcal{S}$ are
unchanged after the measurement by
the observer $\mathcal{W}$
and therefore the states of the system $\mathcal{S}$ are confirmed
before the measurement by the observer $\mathcal{W}$,
where ``apparatus'' in Definition~\ref{def:unchanged-state} should read ``observer''.
Using \eqref{eq:WF-state-immafter-MW}, we see that, in our world $\omega$,
for every $n\in\N^+$ the state of the system $\mathcal{S}$
immediately after the measurement $\mathcal{M}^\mathcal{W}$
is $\ket{\alpha(n)}$ in the $n$th composite system $\mathcal{S}+\mathcal{F}+\mathcal{W}$.
Thus, it follows from Postulate~\ref{CF}~(i) that, in our world $\omega$,
for every $n\in\N^+$ the state of the system $\mathcal{S}$
immediately after the measurement $\mathcal{M}^\mathcal{F}$ is
\begin{equation}\label{eq:WF-S-state-immafter-MF}
  \ket{\alpha(n)}
\end{equation}
in the $n$th composite system $\mathcal{S}+\mathcal{F}+\mathcal{W}$.

On the other hand,
according to Definition~\ref{def:unchanged-apparatus},
it is easy to check that 
the final states of the observer $\mathcal{F}$ are
unchanged after the measurement by
the observer $\mathcal{W}$
and therefore the final states of the observer $\mathcal{F}$ are confirmed
before the measurement by the observer $\mathcal{W}$,
where ``apparatus'' in Definition~\ref{def:unchanged-apparatus} should read ``observer''.
Using \eqref{eq:WF-state-immafter-MW}, we see that, in our world $\omega$,
for every $n\in\N^+$ the state of the observer $\mathcal{F}$
immediately after the measurement $\mathcal{M}^\mathcal{W}$
is $\ket{\Phi^{\mathcal{F}}[\alpha(n)]}$ in the $n$th composite system $\mathcal{S}+\mathcal{F}+\mathcal{W}$.
Thus, it follows from Postulate~\ref{CF}~(ii) that, in our world $\omega$,
for every $n\in\N^+$ the state of the observer $\mathcal{F}$
immediately after the measurement $\mathcal{M}^\mathcal{F}$ is
\begin{equation}\label{eq:WF-F-state-immafter-MF}
  \ket{\Phi^{\mathcal{F}}[\alpha(n)]}
\end{equation}
in the $n$th composite system $\mathcal{S}+\mathcal{F}+\mathcal{W}$.

Hence, it follows from
\eqref{eq:WF-S-state-immafter-MF} and \eqref{eq:WF-F-state-immafter-MF} that,
in our world $\omega$, for every $n\in\N^+$ 
the state of the composite system $\mathcal{S}+\mathcal{F}$
immediately after the measurement $\mathcal{M}^\mathcal{F}$ is given by
\begin{equation}\label{eq:WF-SF-state-immafter-MF}
  \ket{\alpha(n)}\otimes\ket{\Phi^{\mathcal{F}}[\alpha(n)]}
\end{equation}
in the $n$th composite system $\mathcal{S}+\mathcal{F}+\mathcal{W}$.
This state~\eqref{eq:WF-SF-state-immafter-MF}
is different from the `undecided' state~\eqref{eq:SFstateW} in Case~2.
Therefore,
in our world $\omega$,
for every $n\in\N^+$
the observer $\mathcal{F}$ surely
records
the value $\alpha(n)$
immediately after the measurement $\mathcal{M}^\mathcal{F}$
in the $n$th repetition of the experiment,
\emph{in accordance with our common sense}.

In addition,
since $\omega$ is Martin-L\"of $P$-random,
it follows from Theorem~\ref{contraction2} and \eqref{eq:WF-finitepsP}
that $\alpha$ is a Martin-L\"of $B$-random infinite binary sequence,
where $B$ is a finite probability space on $\{0,1\}$ such that
$B(0)=\abs{c_0}^2$ and $B(1)=\abs{c_1}^2$.
Therefore,
it follows from Theorem~\ref{FI} that for every $k\in\{0,1\}$ it holds that
\[
  \lim_{n\to\infty} \frac{N_k(\rest{\alpha}{n})}{n}
  =\abs{c_k}^2,
\]
where $N_k(\rest{\alpha}{n})$ denotes the number of the occurrences of $k$
in the prefix of $\alpha$ of length $n$,
as expected from the point of view of the conventional quantum mechanics.
Thus,
in our world $\omega$,
for each $k=0,1$
the following (i) and (ii) simultaneously occur
in a proportion of $\abs{c_k}^2$ out of the infinite repetitions of the experiment:
\begin{enumerate}
\item The observer $\mathcal{F}$ records the value $k$
immediately after the measurement $\mathcal{M}^\mathcal{F}$, and
\item the observer $\mathcal{F}$ records the value $k$ and
the observer $\mathcal{W}$ knows that the observer $\mathcal{F}$ records the value $k$
immediately after the measurement $\mathcal{M}^\mathcal{W}$,
\end{enumerate}
\emph{in accordance with our common sense}.

We can use Postulate~\ref{Recursive Use} to determine the state of
the composite system $\mathcal{S}+\mathcal{F}$
immediately after the measurement $\mathcal{M}^\mathcal{F}$
in our world $\omega$,
instead of using Postulate~\ref{CF}.
First note that
each of the measurement operators $\{M^{\mathcal{F}}_0,M^{\mathcal{F}}_1\}$ and
$\{I_\mathcal{S}\otimes M^{\mathcal{W}}_0,I_\mathcal{S}\otimes M^{\mathcal{W}}_1,I_\mathcal{S}\otimes M^{\mathcal{W}}_2\}$ forms a PVM.
On the one hand,
applying
$U_\mathcal{F}\otimes I_{\mathcal{W}}$ to the initial state
$\ket{+}\otimes\ket{\Phi^{\mathcal{F}}_{\mathrm{init}}}\otimes\ket{\Phi^{\mathcal{W}}_{\mathrm{init}}}$ in the repeated once of the experiments
results in a state $\ket{\Psi^{\mathcal{F}}_{\mathrm{Total}}}$ given by
\begin{equation}\label{eq:WF-virtual-state-Total-F}
\begin{split}
  \ket{\Psi^{\mathcal{F}}_{\mathrm{Total}}}
  &:=(U_\mathcal{F}\otimes I_{\mathcal{W}})(\ket{+}\otimes\ket{\Phi^{\mathcal{F}}_{\mathrm{init}}}\otimes\ket{\Phi^{\mathcal{W}}_{\mathrm{init}}})
  =\sum_{k=0,1}(M^{\mathcal{F}}_k\ket{+})\otimes\ket{\Phi^{\mathcal{F}}[k]}\otimes\ket{\Phi^{\mathcal{W}}_{\mathrm{init}}} \\
  &=(c_0\ket{0}\otimes\ket{\Phi^{\mathcal{F}}[0]}+c_1\ket{1}\otimes\ket{\Phi^{\mathcal{F}}[1]})\otimes\ket{\Phi^{\mathcal{W}}_{\mathrm{init}}},
\end{split}
\end{equation}
where the second equality follows from \eqref{eq:U-MF-S}, and the last equality follows from \eqref{eq:WF-MF0=00_MF1=11} and \eqref{eq:WF-k+=f1s2k0+k1}.
On the other hand,
the final states of the observer $\mathcal{F}$ are
unchanged after the measurement by
the observer $\mathcal{W}$, as we saw above.
Thus, it follows from Postulate~\ref{Recursive Use} and
\eqref{eq:WF-state-immafter-MW} that, in our world $\omega$,
for each $n\in\N^+$ the state of
the $n$th composite system $\mathcal{S}+\mathcal{F}+\mathcal{W}$
immediately after the measurement $\mathcal{M}^\mathcal{F}$ is given by
\begin{equation}\label{eq:WF-Proj-virtual-state-Total-F}
  P\left(\ket{\Phi^{\mathcal{F}}[\alpha(n)]}\right)\ket{\Psi^{\mathcal{F}}_{\mathrm{Total}}},
\end{equation}
up to the normalization factor, where
\[
  P\left(\ket{\Phi^{\mathcal{F}}[\alpha(n)]}\right)
  =I_{\mathcal{S}}\otimes
  \product{\Phi^{\mathcal{F}}[\alpha(n)]}{\Phi^{\mathcal{F}}[\alpha(n)]}\otimes I_{\mathcal{W}}.
\]
The vector \eqref{eq:WF-Proj-virtual-state-Total-F} equals
\[
  c_{\alpha(n)}\ket{\alpha(n)}\otimes\ket{\Phi^{\mathcal{F}}[\alpha(n)]}\otimes\ket{\Phi^{\mathcal{W}}_{\mathrm{init}}}
\]
for every $n\in\N^+$,
due to \eqref{eq:WF-virtual-state-Total-F}.
Thus, we have the same result as \eqref{eq:WF-SF-state-immafter-MF},
which we obtained based on Postulate~\ref{CF}.

\section{Deutsch's thought experiment}
\label{sec:Deutsch's thought experiment}

In 1985, modifying the Wigner's friend paradox, which is described in Section~\ref{sec:Wigner's friend} of this paper, Deutsch~\cite{Deu85}
proposed a thought experiment
which can, in principle, determine whether
the observation by
the observer $\mathcal{F}$ can lead to the reduction of state vector of the system $\mathcal{S}$ prepared initially in the state $\ket{+}$.
In what follows, we make an analysis of Deutsch's thought experiment
in the framework of quantum mechanics based on the principle of typicality.
Actually, in this paper,
we study a well-known variant of this thought experiment,
which is a simplification of the technique used in Frauchiger and Renner~\cite{FR18}.
We call this variant just \emph{Deutsch's thought experiment}.

In this section,
we review Deutsch's thought experiment
\emph{in the terminology of the conventional quantum mechanics}.
In the setting of Deutsch's thought experiment,
we consider a single qubit system $\mathcal{S}$ and
two observers $\mathcal{F}$ and $\mathcal{D}$,
as in the original Wigner's friend paradox described in Section~\ref{sec:Wigner's friend}.
The observer $\mathcal{F}$ performs a measurement $\mathcal{M}^\mathcal{F}$
over the system $\mathcal{S}$ in the exactly same manner
as in the Wigner's friend paradox.
However, in Deutsch's thought experiment~\cite{Deu85},
the observer $\mathcal{D}$ performs a measurement
over the composite system $\mathcal{S}+\mathcal{F}$,
and not over the observer $\mathcal{F}$ only,
like in the case of the original Wigner's friend paradox.
For completeness and clarity, in what follows,
we describe Deutsch's thought experiment
from scratch.

Consider a single qubit system $\mathcal{S}$ with state space
$\mathcal{H}_\mathcal{S}$.
An observer $\mathcal{F}$ performs a measurement
$\mathcal{M}^\mathcal{F}$
described
by the measurement operators
$\{M^{\mathcal{F}}_0,M^{\mathcal{F}}_1\}$
over the system $\mathcal{S}$
(according to Postulate~\ref{Born-rule}),
where
\begin{equation}\label{eq:DTE-MF0=00_MF1=11}
  M^{\mathcal{F}}_0:=\product{0}{0}\quad\text{and}\quad M^{\mathcal{F}}_1:=\product{1}{1},
\end{equation}
and $\ket{0}$ and $\ket{1}$
form
an orthonormal basis of
the state space $\mathcal{H}_\mathcal{S}$ of the
system $\mathcal{S}$.

Hereafter, we assume that
the observer $\mathcal{F}$ itself is
a quantum system with state space $\mathcal{H}_\mathcal{F}$.
Then we consider an observer $\mathcal{D}$ who is on the outside of
the composite system $\mathcal{S}+\mathcal{F}$
consisting of the system $\mathcal{S}$ and the observer $\mathcal{F}$. 
The observer $\mathcal{D}$ can also be regarded as
an
environment for
the composite system $\mathcal{S}+\mathcal{F}$.
Then, according to \eqref{single_measurement},
the measurement process of
the measurement
$\mathcal{M}^\mathcal{F}$
is described by a unitary operator $U_\mathcal{F}$
acting on
$\mathcal{H}_\mathcal{S}\otimes\mathcal{H}_\mathcal{F}$
such that
\begin{equation}\label{eq:UD-MF-S}
  U_\mathcal{F}(\ket{\Psi}\otimes\ket{\Phi^{\mathcal{F}}_{\mathrm{init}}})
  =\sum_{k=0,1}(M^{\mathcal{F}}_k\ket{\Psi})\otimes\ket{\Phi^{\mathcal{F}}[k]}
\end{equation}
for every $\ket{\Psi}\in\mathcal{H}_\mathcal{S}$.
The unitary operator
$U_\mathcal{F}$ describes the interaction between the system $\mathcal{S}$ and
the observer $\mathcal{F}$ as a quantum system.
The vector $\ket{\Phi^{\mathcal{F}}_{\mathrm{init}}}\in\mathcal{H}_\mathcal{F}$ is
the initial state of the
observer
$\mathcal{F}$, and $\ket{\Phi^{\mathcal{F}}[k]}\in\mathcal{H}_\mathcal{F}$ is
a final state of the
observer
$\mathcal{F}$ for each $k=0,1$,
with $\braket{\Phi^{\mathcal{F}}[k]}{\Phi^{\mathcal{F}}[k']}=\delta_{k,k'}$.
For each $k=0,1$,
the state $\ket{\Phi^{\mathcal{F}}[k]}$ indicates that
\emph{the observer $\mathcal{F}$ records the value $k$}.
The unitary time-evolution \eqref{eq:UD-MF-S} is the quantum mechanical description of
the measurement process of the measurement $\mathcal{M}^\mathcal{F}$
over
the composite system $\mathcal{S}+\mathcal{F}$,
from the point of view of the observer $\mathcal{D}$.

Let $c_0$ and $c_1$ be \emph{arbitrary} two non-zero complex numbers such that $\abs{c_0}^2+\abs{c_1}^2=1$.
We then define a state
$\ket{\Psi^{\mathcal{S}+\mathcal{F}}[+]}\in\mathcal{H}_\mathcal{S}\otimes\mathcal{H}_\mathcal{F}$
of the composite system $\mathcal{S}+\mathcal{F}$ by
\begin{equation}\label{eq:DTE-Def-kPS+F+}
  \ket{\Psi^{\mathcal{S}+\mathcal{F}}[+]}:=c_0\ket{0}\otimes\ket{\Phi^{\mathcal{F}}[0]}+c_1\ket{1}\otimes\ket{\Phi^{\mathcal{F}}[1]}.
\end{equation}
Immediately after the measurement $\mathcal{M}^\mathcal{F}$
performed
by the observer $\mathcal{F}$ over the system $\mathcal{S}$,
the observer $\mathcal{D}$ performs a measurement $\mathcal{M}^\mathcal{D}$
described by the measurement operators $\{M^{\mathcal{D}}_+,M^{\mathcal{D}}_-\}$
over
the composite system $\mathcal{S}+\mathcal{F}$
(according to Postulate~\ref{Born-rule}),
where
\begin{equation}\label{eq:DTE-Def-MW+_MW-}
\begin{split}
  M^{\mathcal{D}}_+&:=\product{\Psi^{\mathcal{S}+\mathcal{F}}[+]}{\Psi^{\mathcal{S}+\mathcal{F}}[+]},\\
  M^{\mathcal{D}}_-&:=\sqrt{I_{\mathcal{S}+\mathcal{F}}-{M^{\mathcal{D}}_+}^\dag M^{\mathcal{D}}_+}
  =I_{\mathcal{S}+\mathcal{F}}-M^{\mathcal{D}}_+,
\end{split}
\end{equation}
and $I_{\mathcal{S}+\mathcal{F}}$
denotes
the identity operator acting on $\mathcal{H}_\mathcal{S}\otimes\mathcal{H}_\mathcal{F}$.

Now, let us assume that
the system $\mathcal{S}$ is initially in a state $\ket{+}$ in the above setting,
where $\ket{+}$ is defined by
\begin{equation}\label{eq:DTE-k+=f1s2k0+k1}
  \ket{+}:=c_0\ket{0}+c_1\ket{1}.
\end{equation}
Then, we have the following two cases, Case~1 or Case~2,
depending on how we treat the measurement $\mathcal{M}^\mathcal{F}$.
Note that
we make the following argument regarding Case~1
and Case~2
in a \emph{heuristic} manner
in terms of the conventional quantum mechanics:

\paragraph{Case 1.}

\emph{Applying Postulate~\ref{Born-rule} to the measurement $\mathcal{M}^\mathcal{F}$},
either the following two possibilities (i) or (ii) occurs
immediately after the measurement $\mathcal{M}^\mathcal{F}$:
\begin{enumerate}
\item The system $\mathcal{S}$ and the observer $\mathcal{F}$ are in the states
  $\ket{0}$ and $\ket{\Phi^{\mathcal{F}}[0]}$, respectively.
  Therefore,
  the state of the composite system $\mathcal{S}+\mathcal{F}$
  is
  $\ket{0}\otimes\ket{\Phi^{\mathcal{F}}[0]}$.
\item The system $\mathcal{S}$ and the observer $\mathcal{F}$ are in the states
  $\ket{1}$ and $\ket{\Phi^{\mathcal{F}}[1]}$, respectively.
  Therefore,
  the state of the composite system
  $\mathcal{S}+\mathcal{F}$
  is
  $\ket{1}\otimes\ket{\Phi^{\mathcal{F}}[1]}$.
\end{enumerate}
Here, the possibilities (i) and (ii) occur with probabilities $\abs{c_0}^2$ and $\abs{c_1}^2$, respectively.
Then, applying Postulate~\ref{Born-rule} to the subsequent measurement $\mathcal{M}^\mathcal{D}$,
we see that
\emph{the
measurement $\mathcal{M}^\mathcal{D}$
gives
the outcomes $+$ with probabilities $\abs{c_0}^2$ and $\abs{c_1}^2$
if the possibilities (i) and (ii) above occur
in the measurement $\mathcal{M}^\mathcal{F}$, respectively}.
Note that,
in this Case~1,
according to Postulate~\ref{Born-rule}
we can consider the joint probability $P_{\mathrm{c1}}(k,l)$
that the measurement $\mathcal{M}^\mathcal{F}$ gives the outcome $k$ and then
the measurement $\mathcal{M}^\mathcal{D}$ gives the outcome $l$
for each $k=0,1$ and $l=+,-$.
We
then
have that
$P_{\mathrm{c1}}(k,+)=\abs{c_k}^2\abs{c_k}^2$ and $P_{\mathrm{c1}}(k,-)=\abs{c_k}^2\abs{c_{1-k}}^2$ for every $k=0,1$.
That is, we have that
\[
  P_{\mathrm{c1}}(k,l)=\abs{c_k}^2(\delta_{l,+}\abs{c_k}^2+\delta_{l,-}\abs{c_{1-k}}^2)
\]
for every $k=0,1$ and $l=+,-$.
Therefore,
the probability that the measurement $\mathcal{M}^\mathcal{D}$ gives the outcome $+$ is given by 
\[
  \sum_{k=0,1}P_{\mathrm{c1}}(k,+)
  =\abs{c_0}^4+\abs{c_1}^4,
\]
which is less than one since $\abs{c_0}^2+\abs{c_1}^2=1$ and $0<\abs{c_0}^2<1$.
Note that this probability takes the minimum value $1/2$
when $\abs{c_0}=\abs{c_1}=1/\sqrt{2}$.

\paragraph{Case 2.}

\emph{According to the unitary time-evolution~\eqref{eq:UD-MF-S}},
the state of the composite system $\mathcal{S}+\mathcal{F}$ is
\begin{equation}\label{eq:DTE-kPS+F+}
  \ket{\Psi^{\mathcal{S}+\mathcal{F}}[+]}
\end{equation}
immediately after the measurement $\mathcal{M}^\mathcal{F}$.
Thus, applying Postulate~\ref{Born-rule} to the subsequent measurement $\mathcal{M}^\mathcal{D}$,
we see that
\emph{the outcome $+$ occurs surely
in the
measurement $\mathcal{M}^\mathcal{D}$}.
Thus, for each $l=+,-$,
we can consider the probability $P_{\mathrm{c2}}(l)$
that the measurement $\mathcal{M}^\mathcal{D}$ gives the outcome $l$,
where $P_{\mathrm{c2}}(+)=1$ and $P_{\mathrm{c2}}(-)=0$.
On the other hand,
the notion of the probability
that the measurement $\mathcal{M}^\mathcal{F}$ gives a certain outcome
is meaningless
in this Case~2.

\bigskip\smallskip

Thus, we have a
difference
about the measurement results of the measurement $\mathcal{M}^\mathcal{D}$,
depending on how we treat the measurement $\mathcal{M}^\mathcal{F}$, i.e.,
depending on whether we are based on Postulate~\ref{Born-rule} (Case~1)
or the unitary time-evolution~\eqref{eq:UD-MF-S} (Case~2).
Note that
this
difference
can be testable, \emph{in principle}, by
repeating the experiment sufficiently many times
in the conventional quantum mechanics.

The above argument regarding Case~1 and Case~2 is a \emph{heuristic} one
in terms of the conventional quantum mechanics.
In the next section,
we refine and reformulate the description and analysis of Deutsch's thought experiment 
in our rigorous framework of quantum mechanics based on the principle of typicality.

\section{Analysis of Deutsch's thought experiment based on the principle of typicality}
\label{sec:Analysis of Deutsch's thought experiment}

In this section, we make an analysis of Deutsch's thought experiment,
which is described in the preceding section in the terminology of the conventional quantum mechanics,
in terms of our framework of quantum mechanics based on
Postulates~\ref{POT} and \ref{Recursive Use},
together with Postulates~\ref{state_space}, \ref{composition}, and \ref{evolution}.
We will then see that
the measurement $\mathcal{M}^\mathcal{D}$ has surely the outcomes $+$,
even if we fully treat the observer $\mathcal{F}$ as a measuring apparatus,
instead of treating
her
as a mere quantum system only being measured.

To proceed this program,
first of all
\emph{we have to implement everything, i.e.,
the
two
measurements $\mathcal{M}^\mathcal{F}$ and $\mathcal{M}^\mathcal{D}$
in the setting of Deutsch's thought experiment,
by unitary time-evolution}.
The measurement process of the measurement
$\mathcal{M}^\mathcal{F}$
is already described by the unitary operator $U_\mathcal{F}$
given in \eqref{eq:UD-MF-S}.
For applying the principle of typicality,
the observer $\mathcal{D}$
is regarded
as a quantum system with state space $\mathcal{H}_\mathcal{D}$.
Then, according to \eqref{single_measurement},
the measurement process of
the measurement
$\mathcal{M}^\mathcal{D}$
is described by a unitary operator $U_\mathcal{D}$
acting on
$\mathcal{H}_\mathcal{S}\otimes\mathcal{H}_\mathcal{F}\otimes\mathcal{H}_\mathcal{D}$
such that
\begin{equation}\label{eq:UD-MW-F}
  U_\mathcal{D}(\ket{\Psi}\otimes\ket{\Phi^{\mathcal{D}}_{\mathrm{init}}})
  =\sum_{l=+,-}(M^{\mathcal{D}}_l\ket{\Psi})\otimes\ket{\Phi^{\mathcal{D}}[l]}
\end{equation}
for every $\ket{\Psi}\in\mathcal{H}_\mathcal{S}\otimes\mathcal{H}_\mathcal{F}$.
The unitary operator
$U_\mathcal{D}$ describes the interaction
between
the composite system $\mathcal{S}+\mathcal{F}$
and the observer $\mathcal{D}$
where both
are regarded as a quantum system.
The vector $\ket{\Phi^{\mathcal{D}}_{\mathrm{init}}}\in\mathcal{H}_\mathcal{D}$ is
the initial state of the
observer
$\mathcal{D}$, and $\ket{\Phi^{\mathcal{D}}[l]}\in\mathcal{H}_\mathcal{D}$ is
a final state of the
observer
$\mathcal{D}$ for each $l=+,-$,
with $\braket{\Phi^{\mathcal{D}}[l]}{\Phi^{\mathcal{D}}[l']}=\delta_{l,l'}$.
For each $l=+,-$, the state $\ket{\Phi^{\mathcal{D}}[l]}$ indicates that
\emph{the observer $\mathcal{D}$ records the value $l$}.

Now, we set
$P_0:=\product{0}{0}$ and $P_1:=\product{1}{1}$.
Then, from \eqref{eq:UD-MF-S}, obviously we have that
\begin{equation}\label{eq:UD-MF-S_Pa}
U_{\mathcal{F}}(\ket{\Psi}\otimes\ket{\Phi^{\mathcal{F}}_\mathrm{init}})=\sum_{a=0,1} (P_a\ket{\Psi})\otimes\ket{\Phi^{\mathcal{F}}[a]}
\end{equation}
for every $\ket{\Psi}\in\mathcal{H}_\mathcal{S}$.
We define $\ket{-}\in\mathcal{H}_{\mathcal{S}}$ by
\begin{equation}\label{eq:DTE-k-}
  \ket{-}:=\overline{c_1}\ket{0}-\overline{c_0}\ket{1},
\end{equation}
and define
$\ket{\Psi^{\mathcal{S}+\mathcal{F}}[-]}\in\mathcal{H}_\mathcal{S}\otimes\mathcal{H}_\mathcal{F}$ by
\begin{equation}\label{eq:DTE-Def-kPS+F-}
  \ket{\Psi^{\mathcal{S}+\mathcal{F}}[-]}:=\overline{c_1}\ket{0}\otimes\ket{\Phi^{\mathcal{F}}[0]}-\overline{c_0}\ket{1}\otimes\ket{\Phi^{\mathcal{F}}[1]}.
\end{equation}
Then, it follows from \eqref{eq:UD-MF-S_Pa} that
\begin{equation}\label{eq:D-sPkoPFa=kPS+F+-}
  U_{\mathcal{F}}(\ket{l}\otimes\ket{\Phi^{\mathcal{F}}_\mathrm{init}})
  =\sum_{a=0,1} (P_a\ket{l})\otimes\ket{\Phi^{\mathcal{F}}[a]}
  =\ket{\Psi^{\mathcal{S}+\mathcal{F}}[l]}
\end{equation}
for each $l=+,-$,
where the last equality follows from \eqref{eq:DTE-k+=f1s2k0+k1}, \eqref{eq:DTE-Def-kPS+F+}, \eqref{eq:DTE-k-}, and \eqref{eq:DTE-Def-kPS+F-}.
On the other hand,
using \eqref{eq:DTE-Def-MW+_MW-}, \eqref{eq:DTE-Def-kPS+F+}, and \eqref{eq:DTE-Def-kPS+F-} we have that
\begin{equation}\label{eq:D-MDm-kPS+F=d-kPS+F}
  M^{\mathcal{D}}_l\ket{\Psi^{\mathcal{S}+\mathcal{F}}[l']}=\delta_{l,l'}\ket{\Psi^{\mathcal{S}+\mathcal{F}}[l']}
\end{equation}
for every $l,l'=+,-$.
Thus, using \eqref{eq:UD-MW-F} and \eqref{eq:D-sPkoPFa=kPS+F+-} we have that
\begin{equation}\label{eq:Deutsch-ket}
\begin{split}
U_{\mathcal{D}} ((U_{\mathcal{F}}(\ket{l}\otimes\ket{\Phi^{\mathcal{F}}_\mathrm{init}}))\otimes\ket{\Phi^{\mathcal{D}}_\mathrm{init}})
&=\sum_{b=+,-}(M^{\mathcal{D}}_b\ket{\Psi^{\mathcal{S}+\mathcal{F}}[l]})\otimes\ket{\Phi^{\mathcal{D}}[b]} \\
&=\ket{\Psi^{\mathcal{S}+\mathcal{F}}[l]})\otimes\ket{\Phi^{\mathcal{D}}[l]} \\
&=\sum_{a=0,1} (P_a\ket{l})\otimes\ket{\Phi^{\mathcal{F}}[a]}\otimes\ket{\Phi^{\mathcal{D}}[l]}
\end{split}
\end{equation}
for each $l=+,-$,
where the second equality follows from \eqref{eq:D-MDm-kPS+F=d-kPS+F},
and the last equality follows from \eqref{eq:DTE-Def-kPS+F+} and \eqref{eq:DTE-Def-kPS+F-}.
We set
$Q_+:=\product{+}{+}$ and $Q_-:=\product{-}{-}$.
Note that $\ket{+}$ and $\ket{-}$ form an orthonormal basis of $\mathcal{H}_\mathcal{S}$.
Therefore it follows that
\begin{equation}\label{eq:D-Q++Q-=I}
  Q_+ +Q_- =I_\mathcal{S},
\end{equation}
where $I_\mathcal{S}$ denotes the identity operator acting on $\mathcal{H}_\mathrm{S}$.
Then, based on \eqref{eq:Deutsch-ket}, we have that
\begin{equation}\label{eq:Deutsch-pre-all}
U_{\mathcal{D}} ((U_{\mathcal{F}}(\ket{\Psi}\otimes\ket{\Phi^{\mathcal{F}}_\mathrm{init}}))\otimes\ket{\Phi^{\mathcal{D}}_\mathrm{init}})=\sum_{a=0,1}\sum_{b=+,-} (P_a Q_b\ket{\Psi})\otimes\ket{\Phi^{\mathcal{F}}[a]}\otimes\ket{\Phi^{\mathcal{D}}[b]}
\end{equation}
for every $\ket{\Psi}\in\mathcal{H}_\mathcal{S}$.
Actually, using \eqref{eq:Deutsch-ket}
we can check that the equality \eqref{eq:Deutsch-pre-all} holds certainly
in the case where $\ket{\Psi}=\ket{+}$ or $\ket{\Psi}=\ket{-}$.
Therefore,
since $\ket{+}$ and $\ket{-}$ form an orthonormal basis of $\mathcal{H}_\mathcal{S}$,
due to linearity we have that the equality \eqref{eq:Deutsch-pre-all} holds for an arbitrary $\ket{\Psi}\in\mathcal{H}_\mathcal{S}$,
as desired.

We denote the set $\{0,1\}\times\{+,-\}$ by $\Omega$,
and we define a finite
collection
$\{M_{k,l}\}_{(k,l)\in\Omega}$ of operators acting on
the state space $\mathcal{H}_\mathcal{S}$ by
\begin{equation}\label{eq:DTE-Mkl=PkQl}
  M_{k,l}:=P_k Q_l.
\end{equation}
Then it follows from \eqref{eq:D-Q++Q-=I} that
\begin{equation*}
  \sum_{(k,l)\in\Omega} M_{k,l}^\dag M_{k,l}
  =\sum_{b=+,-}\,\sum_{a=0,1} Q_b P_a Q_b
  =\sum_{b=+,-}Q_b=I_\mathcal{S}.
\end{equation*}
Thus, the finite collection $\{M_{k,l}\}_{(k,l)\in\Omega}$ satisfies the \emph{completeness equation}, and therefore forms \emph{measurement operators}.
We denote the identity operator acting on $\mathcal{H}_\mathrm{D}$ by $I_\mathcal{D}$,
and we define a unitary operator $U_{\mathrm{whole}}$ acting on $\mathcal{H}_\mathcal{S}\otimes\mathcal{H}_\mathcal{F}\otimes\mathcal{H}_\mathcal{D}$
by
\[
  U_{\mathrm{whole}}:=U_{\mathcal{D}}\circ(U_{\mathcal{F}}\otimes I_\mathcal{D}).
\]
Then, using \eqref{eq:Deutsch-pre-all} and \eqref{eq:DTE-Mkl=PkQl} we see that, as a whole,
the sequential applications of $U_{\mathcal{F}}$ and $U_{\mathcal{D}}$
to a state $\ket{\Psi}\otimes\ket{\Phi_{\mathrm{init}}}$ of the composite system $\mathcal{S}+\mathcal{F}+\mathcal{D}$
result in:
\begin{equation}\label{eq:Deutsch-all}
  U_{\mathrm{whole}}(\ket{\Psi}\otimes\ket{\Phi_{\mathrm{init}}})
  =\sum_{(k,l)\in\Omega}(M_{k,l}\ket{\Psi})\otimes\ket{\Phi[k,l]}
\end{equation}
for every state $\ket{\Psi}\in\mathcal{H}_\mathcal{S}$ of the system $\mathcal{S}$,
where $\ket{\Phi_{\mathrm{init}}}:=
\ket{\Phi^{\mathcal{F}}_{\mathrm{init}}}\otimes\ket{\Phi^{\mathcal{D}}_{\mathrm{init}}}$
and
$\ket{\Phi[k,l]}:=\ket{\Phi^{\mathcal{F}}[k]}\otimes\ket{\Phi^{\mathcal{D}}[l]}$.

\subsection{Application of the principle of typicality}
\label{subsec:Deutsch_Application of the principle of typicality}

Thus, in the setting of Deutsch's thought experiment,
the unitary operator $U_{\mathrm{whole}}$ applying to
the initial state $\ket{+}\otimes\ket{\Phi_{\mathrm{init}}}$ describes
the \emph{repeated once} of the infinite repetition of the measurements
where the succession of
the measurement $\mathcal{M}^\mathcal{F}$ and
the measurement $\mathcal{M}^\mathcal{D}$
is infinitely repeated.
It follows from \eqref{eq:Deutsch-all} that
the application of $U_{\mathrm{whole}}$,
consisting of the sequential applications of $U_{\mathcal{F}}$ and $U_{\mathcal{D}}$ in this order,
can be regarded as a
\emph{single measurement} which is described by
the measurement operators $\{M_{k,l}\}_{(k,l)\in\Omega}$
and whose all possible outcomes form the set $\Omega$.

Hence, we can apply Definition~\ref{pmrpwst}
to this scenario of the setting of measurements.
Therefore, according to Definition~\ref{pmrpwst},
we can see that a \emph{world} is an infinite sequence over $\Omega$
and the probability measure induced by
the \emph{probability measure representation for the prefixes of worlds}
is a Bernoulli measure $\lambda_P$ on
$\Omega^\infty$,
where $P$ is a finite probability space on $\Omega$ such that
$P(k,l)$ is the square of the norm of the vector
\begin{equation*}
  (M_{k,l}\ket{+})\otimes\ket{\Phi[k,l]}
\end{equation*}
for every $(k,l)\in\Omega$.
Here $\Omega$ is the set of all possible records of
the observers $\mathcal{F}$ and $\mathcal{D}$
in the \emph{repeated once} of the experiments.
Let us calculate the explicit form of $P(k,l)$.
Then, using~\eqref{eq:DTE-Mkl=PkQl} and \eqref{eq:DTE-k+=f1s2k0+k1}, we have that
\begin{equation}\label{eq:Deutsch-finitepsP}
  P(k,l)=\bra{+}M_{k,l}^\dag M_{k,l}\ket{+}\braket{\Phi[k,l]}{\Phi[k,l]}
  =\bra{+}Q_l P_k Q_l\ket{+}
  =\delta_{l,+}\bra{+}P_k\ket{+}
  =\delta_{l,+}\abs{c_k}^2
\end{equation}
for each $(k,l)\in\Omega$.

Now, let us apply
Postulate~\ref{POT}, the \emph{principle of typicality},
to the setting of measurements
developed above.
Let $\omega$ be \emph{our world} in the infinite repetition of
the measurement described by the measurement operators $\{M_{k,l}\}_{(k,l)\in\Omega}$ in the above setting.
This $\omega$
is an infinite sequence over $\Omega$
consisting of records
in the
two
observers $\mathcal{F}$ and $\mathcal{D}$
which is being generated by the infinite repetition of the measurement
described by the measurement operators $\{M_{k,l}\}_{(k,l)\in\Omega}$
in the above setting.
Since the Bernoulli measure $\lambda_P$ on $\Omega^\infty$ is
the probability measure induced by the
probability
measure representation
for the prefixes of
worlds
in the above setting,
it follows from
Postulate~\ref{POT}
that $\omega$ is Martin-L\"{o}f random with respect to the Bernoulli measure $\lambda_{P}$, that is, \emph{$\omega$ is Martin-L\"of $P$-random}.

We use $\alpha$ and $\beta$
to denote the infinite sequences over $\{0,1\}$ and $\{+,-\}$, respectively,
such that $(\alpha(n),\beta(n))=\omega(n)$ for every $n\in\N^+$.
Then, since $\omega$ is Martin-L\"of $P$-random, it follows from \eqref{eq:Deutsch-finitepsP} and
Corollary~\ref{cor:always-positive-probability}
that
$\beta$
consists only of $+$, i.e., $\beta=++++++\dotsc\dotsc$.
Thus, it follows from Postulate~\ref{POT} and \eqref{eq:Deutsch-all}
that, in our world $\omega$, for each $n\in\N^+$ the state of
the $n$th composite system $\mathcal{S}+\mathcal{F}+\mathcal{D}$,
i.e., the $n$th copy of a composite system consisting of
the system $\mathcal{S}$ and the two observers $\mathcal{F}$ and $\mathcal{D}$,
immediately after the measurement $\mathcal{M}^\mathcal{D}$ is given by
\begin{equation*}
\begin{split}
  (M_{\alpha(n),\beta(n)}\ket{+})\otimes\ket{\Phi[\alpha(n),\beta(n)]}
  &=(M_{\alpha(n),+}\ket{+})\otimes\ket{\Phi^{\mathcal{F}}[\alpha(n)]}\otimes\ket{\Phi^{\mathcal{D}}[+]} \\
  &=c_{\alpha(n)}\ket{\alpha(n)}\otimes\ket{\Phi^{\mathcal{F}}[\alpha(n)]}\otimes\ket{\Phi^{\mathcal{D}}[+]},
\end{split}
\end{equation*}
up to the normalization factor,
where the last equality follows from \eqref{eq:DTE-Mkl=PkQl}
and \eqref{eq:DTE-k+=f1s2k0+k1}.
Recall that the state $\ket{\Phi^{\mathcal{D}}[+]}$ indicates that \emph{the observer $\mathcal{D}$ records the value $+$}.
Hence, we have that, in our world $\omega$, \emph{the observer $\mathcal{D}$ always records the value $+$
immediately
after the measurement $\mathcal{M}^\mathcal{D}$ in every repetition
of the experiment},
independently of $n$.

In addition,
since $\omega$ is Martin-L\"of $P$-random,
it follows from Theorem~\ref{contraction2} and \eqref{eq:Deutsch-finitepsP}
that
$\alpha$ is a Martin-L\"of $B$-random infinite binary sequence,
where $B$ is a finite probability space on $\{0,1\}$ such that
$B(0)=\abs{c_0}^2$ and $B(1)=\abs{c_1}^2$.
Therefore,
it follows from Theorem~\ref{FI} that for every $k\in\{0,1\}$ it holds that
\[
  \lim_{n\to\infty} \frac{N_k(\rest{\alpha}{n})}{n}
  =\abs{c_k}^2,
\]
where $N_k(\rest{\alpha}{n})$ denotes the number of the occurrences of $k$
in the prefix of $\alpha$ of length $n$.
Thus,
in our world $\omega$,
the following (i) and (ii) hold:
\begin{enumerate}
\item 
For each $k=0,1$, the observer $\mathcal{F}$ records the value $k$
immediately after the measurement $\mathcal{M}^\mathcal{D}$
in a proportion of $\abs{c_k}^2$ out of the infinite repetitions of the experiment.
\item The observer $\mathcal{D}$ always records the value $+$
immediately
after the measurement $\mathcal{M}^\mathcal{D}$
in every repetition of the experiment,
whichever value the observer $\mathcal{F}$ records in the
corresponding
repetition.
\end{enumerate}
We can then see that
\emph{the above situation of (i) and (ii) is different from both Case~1 and Case~2
in Section~\ref{sec:Deutsch's thought experiment}},
although the argument regarding Case~1 and Case~2 given in Section~\ref{sec:Deutsch's thought experiment} is heuristic
and the statements of
both Case~1 and Case~2 in Section~\ref{sec:Deutsch's thought experiment} are
at least operationally vague
since they are described
in terms of the conventional quantum mechanics
where the operational characterization of the notion of probability is not given.

Let us determine the state of
the composite system $\mathcal{S}+\mathcal{F}$
immediately after the measurement $\mathcal{M}^\mathcal{F}$
in our world $\omega$.
For that purpose, we can apply our general framework
to treat situations
where apparatuses perform measurements over other apparatuses as well as over a normal quantum system only being measured,
developed in Section~\ref{sec:confirming point},
to the setting of measurements in Deutsch's thought experiment,
developed above.
The current setting of measurements in Deutsch's thought experiment is treated within the general framework with $n=2$,
and the measurements $\mathcal{M}^\mathcal{F}$ and $\mathcal{M}^\mathcal{D}$ in the current setting serve as the measurements $\mathcal{M}^1$ and $\mathcal{M}^2$ in the general framework, respectively.
The single qubit system $\mathcal{S}$, and the observers $\mathcal{F}$ and $\mathcal{D}$ in the current setting serve as the system $\mathcal{S}$, and the apparatuses $\mathcal{A}_1$ and $\mathcal{A}_2$ in the general framework, respectively.
Then the measurement operators
$\{M^1_{m_1}\}_{m_1\in\Omega_1}$ and $\{M^2_{m_2}\}_{m_2\in\Omega_2}$
in the general framework are instantiated in the current setting
in the following manner:
\begin{enumerate}
  \item $\Omega_1=\{0,1\}$ and $\Omega_2=\{+,-\}$.
  \item $M^1_{m_1}=M^{\mathcal{F}}_{m_1}$ for every $m_1\in\Omega_1$, and
    $M^2_{m_2}=M^{\mathcal{D}}_{m_2}$ for every $m_2\in\Omega_2$.
\end{enumerate}
Therefore,
based on \eqref{eq:DTE-Def-MW+_MW-},
it is easy to check that the condition~\eqref{eq:Ui-TSCP-MO} certainly holds
in the current setting of measurements in Deutsch's thought experiment.
Thus, based on \eqref{eq:Deutsch-all}, we see that
the measurement operators $\{M^{1,2}_{m_1,m_2}\}_{(m_1,m_2)\in\Omega_1\times\Omega_2}$
satisfying \eqref{eq:Ui-TSCP-1tok} with $n=k=2$ in the general framework are instantiated in the current setting
in the manner that
\begin{equation*}
  M^{1,2}_{m_1,m_2}=M_{m_1,m_2}
\end{equation*}
for every $m_1\in\Omega_1$ and $m_2\in\Omega_2$,
where the measurement operators $\{M_{k,l}\}_{(k,l)\in\Omega}$ on the right-hand side are defined by \eqref{eq:DTE-Mkl=PkQl}.
The state $\ket{\Psi^0_{\mathrm{init}}}$ of the system $\mathcal{S}$
immediately before the measurement $\mathcal{M}^1$ in the general framework
materializes as
the initial state $\ket{+}$ of the single qubit system $\mathcal{S}$ in the current setting,
defined by \eqref{eq:DTE-k+=f1s2k0+k1}.

Thus,
we can use Postulate~\ref{Recursive Use} to determine the state of
the composite system $\mathcal{S}+\mathcal{F}$
immediately after the measurement $\mathcal{M}^\mathcal{F}$
in our world $\omega$.
First note that each of the measurement operators $\{M^{\mathcal{F}}_0,M^{\mathcal{F}}_1\}$ and $\{M^{\mathcal{D}}_+,M^{\mathcal{D}}_-\}$ forms a PVM.
On the one hand, applying $U_\mathcal{F}\otimes I_{\mathcal{D}}$
to the initial state
$\ket{+}\otimes\ket{\Phi^{\mathcal{F}}_{\mathrm{init}}}\otimes\ket{\Phi^{\mathcal{D}}_{\mathrm{init}}}$ in the repeated once of the experiments
results in a state $\ket{\Psi^{\mathcal{F}}_{\mathrm{Total}}}$ given by
\begin{equation}\label{eq:DTE-virtual-state-Total-F}
\begin{split}
  \ket{\Psi^{\mathcal{F}}_{\mathrm{Total}}}
  &:=(U_\mathcal{F}\otimes I_{\mathcal{D}})(\ket{+}\otimes\ket{\Phi^{\mathcal{F}}_{\mathrm{init}}}\otimes\ket{\Phi^{\mathcal{D}}_{\mathrm{init}}})
  =\sum_{k=0,1}(M^{\mathcal{F}}_k\ket{+})\otimes\ket{\Phi^{\mathcal{F}}[k]}\otimes\ket{\Phi^{\mathcal{D}}_{\mathrm{init}}} \\
  &=(c_0\ket{0}\otimes\ket{\Phi^{\mathcal{F}}[0]}+c_1\ket{1}\otimes\ket{\Phi^{\mathcal{F}}[1]})\otimes\ket{\Phi^{\mathcal{D}}_{\mathrm{init}}} \\
  &=\ket{\Psi^{\mathcal{S}+\mathcal{F}}[+]}\otimes\ket{\Phi^{\mathcal{D}}_{\mathrm{init}}},
\end{split}
\end{equation}
where the second equality follows from \eqref{eq:UD-MF-S},
the third equality follows from \eqref{eq:DTE-MF0=00_MF1=11} and \eqref{eq:DTE-k+=f1s2k0+k1},
and the last equality follows from \eqref{eq:DTE-Def-kPS+F+}.
On the other hand,
it is easy to see that, according to Definition~\ref{def:unchanged-apparatus},
the final states of the observer $\mathcal{F}$ are \emph{not} unchanged after the measurement by
the observer $\mathcal{D}$,
where ``apparatus'' in Definition~\ref{def:unchanged-apparatus} should read ``observer''.
Thus, it follows from Postulate~\ref{Recursive Use}
that, in our world $\omega$, for each $n\in\N^+$ the state of
the $n$th composite system $\mathcal{S}+\mathcal{F}+\mathcal{D}$
immediately after the measurement $\mathcal{M}^\mathcal{F}$ is given by
\begin{equation}\label{eq:DTE-Proj-virtual-state-Total-F}
  P\left(\right)\ket{\Psi^{\mathcal{F}}_{\mathrm{Total}}},
\end{equation}
up to the normalization factor, where
$P\left(\right)$
equals
the identity operator acting on the state space
$\mathcal{H}_\mathcal{S}\otimes\mathcal{H}_\mathcal{F}\otimes\mathcal{H}_\mathcal{D}$
of the composite system $\mathcal{S}+\mathcal{F}+\mathcal{D}$.
Hence, it follows from
\eqref{eq:DTE-Proj-virtual-state-Total-F} and \eqref{eq:DTE-virtual-state-Total-F} that,
in our world $\omega$, for every $n\in\N^+$ 
the state of the composite system $\mathcal{S}+\mathcal{F}$
immediately after the measurement $\mathcal{M}^\mathcal{F}$ is given by
\begin{equation*}%
  \ket{\Psi^{\mathcal{S}+\mathcal{F}}[+]}
\end{equation*}
in the $n$th composite system $\mathcal{S}+\mathcal{F}+\mathcal{D}$.

\section{\boldmath Analysis of Deutsch's thought experiment, with the friend $\mathcal{F}$ being a mere quantum system only measured}
\label{sec:Analysis of DTE with a mere quantum system F}

In the previous section, we have treated the observer $\mathcal{F}$
as a measuring apparatus
in the analysis of Deutsch's thought experiment
in
our rigorous framework of quantum mechanics based on the principle of typicality.
What happens
if
we regard the observer $\mathcal{F}$ as a mere quantum system
and not as a measurement apparatus
in the analysis of Deutsch's thought experiment in our framework?
We can make an analysis of this case still more easily in our framework,
as follows.
We do not need to use either Postulate~\ref{CF} or Postulate~\ref{Recursive Use} in this case.
Note that this case corresponds to
Case~2 in Section~\ref{sec:Deutsch's thought experiment},
which is described in terms of the conventional quantum mechanics in Section~\ref{sec:Deutsch's thought experiment}.

In this case,
according to \eqref{eq:UD-MF-S},
the state of the composite system $\mathcal{S}+\mathcal{F}$
immediately after the `measurement' $\mathcal{M}^\mathcal{F}$ is given
by
$U_{\mathcal{F}}(\ket{+}\otimes\ket{\Phi^{\mathcal{F}}_\mathrm{init}})$.
Thus, the unitary operator $U_{\mathcal{D}}$ given by \eqref{eq:UD-MW-F} applying to
the initial state
$$(U_{\mathcal{F}}(\ket{+}\otimes\ket{\Phi^{\mathcal{F}}_\mathrm{init}}))
\otimes\ket{\Phi^{\mathcal{D}}_\mathrm{init}}$$
describes the \emph{repeated once} of the infinite repetition of the measurement
(experiment)
where the measurement $\mathcal{M}^\mathcal{D}$ is infinitely repeated.
Let $\Theta:=\{+,-\}$.
The application of $U_{\mathcal{D}}$ itself, of course, forms a \emph{single measurement},
which is described by
the measurement operators $\{M^{\mathcal{D}}_l\}_{l\in\Theta}$
and whose all possible outcomes form the set $\Theta$.

Hence, we can apply Definition~\ref{pmrpwst}
to this scenario of the setting of measurements.
Therefore, according to Definition~\ref{pmrpwst},
we can see that a \emph{world} is an infinite sequence over $\Theta$
and the probability measure induced by
the \emph{probability measure representation for the prefixes of worlds}
is a Bernoulli measure $\lambda_R$ on
$\Theta^\infty$,
where $R$ is a finite probability space on $\Theta$ such that
$R(l)$ is the square of the norm of the vector
\begin{equation*}
  (M^{\mathcal{D}}_l(U_{\mathcal{F}}(\ket{+}\otimes\ket{\Phi^{\mathcal{F}}_\mathrm{init}})))
  \otimes\ket{\Phi^{\mathcal{D}}[l]}
\end{equation*}
for every $l\in\Theta$.
Here $\Theta$ is the set of all possible records of
the observer $\mathcal{D}$
in the \emph{repeated once} of the measurements
(experiments).
Let us calculate the explicit form of $R(l)$.
Recall that
\begin{equation}\label{eq:DTE-UIk+okPFi=kPS+F+}
U_{\mathcal{F}}(\ket{+}\otimes\ket{\Phi^{\mathcal{F}}_\mathrm{init}})
=\ket{\Psi^{\mathcal{S}+\mathcal{F}}[+]},
\end{equation}
which follows
from \eqref{eq:UD-MF-S}, \eqref{eq:DTE-MF0=00_MF1=11}, \eqref{eq:DTE-k+=f1s2k0+k1},
and \eqref{eq:DTE-Def-kPS+F+}.
Thus,
using~\eqref{eq:DTE-Def-MW+_MW-} we have that
\begin{equation}\label{eq:Deutsch-finitepsP-semi}
  R(l)=\bra{\Psi^{\mathcal{S}+\mathcal{F}}[+]}{M^{\mathcal{D}}_l}^\dag M^{\mathcal{D}}_l\ket{\Psi^{\mathcal{S}+\mathcal{F}}[+]}\braket{\Phi^{\mathcal{D}}[l]}{\Phi^{\mathcal{D}}[l]}
  =\delta_{l,+}
\end{equation}
for each $l\in\Theta$.

Now, let us apply
Postulate~\ref{POT}, the \emph{principle of typicality},
to the setting of measurements
developed above.
Let $\omega$ be \emph{our world} in the infinite repetition of the measurement $\mathcal{M}^\mathcal{D}$ described by the measurement operators $\{M^{\mathcal{D}}_l\}_{l\in\Theta}$
in the above setting.
This $\omega$
is an infinite sequence over $\Theta$
consisting of records
in the observer $\mathcal{D}$
which is being generated by the infinite repetition of
the measurement $\mathcal{M}^\mathcal{D}$
described by the measurement operators $\{M^{\mathcal{D}}_l\}_{l\in\Theta}$
in the above setting.
Since the Bernoulli measure $\lambda_R$ on $\Theta^\infty$ is
the probability measure induced by the
probability
measure representation
for the prefixes of
worlds
in the above setting,
it follows from
Postulate~\ref{POT}
that $\omega$ is Martin-L\"{o}f random with respect to the Bernoulli measure $\lambda_{R}$, that is, \emph{$\omega$ is Martin-L\"of $R$-random}.

Then, since $\omega$ is Martin-L\"of $R$-random, it follows from \eqref{eq:Deutsch-finitepsP-semi} and
Corollary~\ref{cor:always-positive-probability}
that
$\omega$
consists only of $+$, i.e., $\omega=++++++\dotsc\dotsc$.
Thus, it follows from Postulate~\ref{POT} and \eqref{eq:UD-MW-F}
that, in our world $\omega$, for each $n\in\N^+$ the state of
the $n$th composite system $\mathcal{S}+\mathcal{F}+\mathcal{D}$,
i.e., the $n$th copy of a composite system consisting of
the system $\mathcal{S}$ and the `observer' $\mathcal{F}$ both being only measured,
and the real observer $\mathcal{D}$ measuring them,
immediately after the measurement $\mathcal{M}^\mathcal{D}$ is given by
\begin{align*}
  (M^{\mathcal{D}}_{\omega(n)}(U_{\mathcal{F}}(\ket{+}\otimes\ket{\Phi^{\mathcal{F}}_\mathrm{init}})))
  \otimes\ket{\Phi^{\mathcal{D}}[\omega(n)]}
  &=(M^{\mathcal{D}}_{+}(U_{\mathcal{F}}(\ket{+}\otimes\ket{\Phi^{\mathcal{F}}_\mathrm{init}})))
  \otimes\ket{\Phi^{\mathcal{D}}[+]} \\
  &=\ket{\Psi^{\mathcal{S}+\mathcal{F}}[+]}\otimes\ket{\Phi^{\mathcal{D}}[+]},
\end{align*}
where the last equality follows from
\eqref{eq:DTE-Def-MW+_MW-} and \eqref{eq:DTE-UIk+okPFi=kPS+F+}.
Recall that the state $\ket{\Phi^{\mathcal{D}}[+]}$ indicates that \emph{the observer $\mathcal{D}$ records the value $+$}.
Hence,
it follows that, in our world $\omega$,
\emph{the observer $\mathcal{D}$ always records the value $+$
immediately after the measurement $\mathcal{M}^\mathcal{D}$ in every repetition of the experiment},
independently of $n$.
Thus, this result
recovers in a \emph{refined manner}
Case~2 in Section~\ref{sec:Deutsch's thought experiment},
which is described and analyzed in terms of the conventional quantum mechanics
in Section~\ref{sec:Deutsch's thought experiment}.
Regarding the results of the measurement $\mathcal{M}^\mathcal{D}$ performed by the observer $\mathcal{D}$,
we have the same results as
in Section~\ref{sec:Analysis of Deutsch's thought experiment},
where both of the treatments are based on the principle of typicality.
The important point here is that
neither the result of Section~\ref{sec:Analysis of Deutsch's thought experiment}
nor the result of this section
corresponds to Case~1 in Section~\ref{sec:Deutsch's thought experiment},
which is described and heuristically analyzed in terms of the conventional quantum mechanics in a vague manner in Section~\ref{sec:Deutsch's thought experiment}.

\section{Wigner collaborates with Deutsch}
\label{sec:Wigner collaborates with Deutsch}

In our general framework developed in Section~\ref{sec:confirming point},
we can analyze more complicated situations where apparatuses perform measurements over other apparatuses,
than considered in the previous sections.
In this section, as such an example, we consider a combination of
Wigner's friend and Deutsch's thought experiment.
The system $\mathcal{S}$ and
the three observers $\mathcal{F}$, $\mathcal{W}$, and $\mathcal{D}$
participate in the experiment of this combination.
The three observers make their own measurements as in the preceding sections,
but in the order of the measurement $\mathcal{M}^\mathcal{W}$ and
then the measurement $\mathcal{M}^\mathcal{D}$,
where the observer $\mathcal{D}$ do nothing to the observer $\mathcal{W}$.
We call this combination the \emph{Wigner-Deutsch collaboration}.
In this section, we describe the Wigner-Deutsch collaboration
in detail
\emph{in the terminology of the conventional quantum mechanics}
from scratch.

Consider a single qubit system $\mathcal{S}$ with state space
$\mathcal{H}_\mathcal{S}$.
An observer $\mathcal{F}$ performs a measurement
$\mathcal{M}^\mathcal{F}$
described
by the measurement operators
$\{M^{\mathcal{F}}_0,M^{\mathcal{F}}_1\}$
over the system $\mathcal{S}$
(according to Postulate~\ref{Born-rule}),
where
\begin{equation}\label{eq:EDTE-MF0=00_MF1=11}
  M^{\mathcal{F}}_0:=\product{0}{0}\quad\text{and}\quad M^{\mathcal{F}}_1:=\product{1}{1},
\end{equation}
and $\ket{0}$ and $\ket{1}$
form
an orthonormal basis of
the state space $\mathcal{H}_\mathcal{S}$ of the
system $\mathcal{S}$.

Hereafter, we assume that
the observer $\mathcal{F}$ itself is
a quantum system with state space $\mathcal{H}_\mathcal{F}$.
Then we consider observers $\mathcal{W}$ and $\mathcal{D}$ who are on the outside of
the composite system $\mathcal{S}+\mathcal{F}$
consisting of the system $\mathcal{S}$ and the observer $\mathcal{F}$. 
The observers $\mathcal{W}$ and $\mathcal{D}$ together can also be regarded as
an
environment for
the composite system $\mathcal{S}+\mathcal{F}$.
Then, according to \eqref{single_measurement},
the measurement process of
the measurement
$\mathcal{M}^\mathcal{F}$
is described by a unitary operator $U_\mathcal{F}$
acting on
$\mathcal{H}_\mathcal{S}\otimes\mathcal{H}_\mathcal{F}$
such that
\begin{equation}\label{eq:EUD-MF-S}
  U_\mathcal{F}(\ket{\Psi}\otimes\ket{\Phi^{\mathcal{F}}_{\mathrm{init}}})
  =\sum_{k=0,1}(M^{\mathcal{F}}_k\ket{\Psi})\otimes\ket{\Phi^{\mathcal{F}}[k]}
\end{equation}
for every $\ket{\Psi}\in\mathcal{H}_\mathcal{S}$.
The unitary operator
$U_\mathcal{F}$ describes the interaction between the system $\mathcal{S}$ and
the observer $\mathcal{F}$ as a quantum system.
The vector $\ket{\Phi^{\mathcal{F}}_{\mathrm{init}}}\in\mathcal{H}_\mathcal{F}$ is
the initial state of the
observer
$\mathcal{F}$, and $\ket{\Phi^{\mathcal{F}}[k]}\in\mathcal{H}_\mathcal{F}$ is
a final state of the
observer
$\mathcal{F}$ for each $k=0,1$,
with $\braket{\Phi^{\mathcal{F}}[k]}{\Phi^{\mathcal{F}}[k']}=\delta_{k,k'}$.
For each $k=0,1$,
the state $\ket{\Phi^{\mathcal{F}}[k]}$ indicates that
\emph{the observer $\mathcal{F}$ records the value $k$}.
The unitary time-evolution \eqref{eq:EUD-MF-S} is the quantum mechanical description of
the measurement process of the measurement $\mathcal{M}^\mathcal{F}$
over
the composite system $\mathcal{S}+\mathcal{F}$,
from the point of view of the observers $\mathcal{W}$ and $\mathcal{D}$.

Immediately after the measurement $\mathcal{M}^\mathcal{F}$
performed by the observer $\mathcal{F}$ over the system $\mathcal{S}$,
the observer $\mathcal{W}$
performs a measurement $\mathcal{M}^\mathcal{W}$
described by the measurement operators
$\{M^{\mathcal{W}}_0,M^{\mathcal{W}}_1,M^{\mathcal{W}}_2\}$
over the system $\mathcal{F}$
(according to Postulate~\ref{Born-rule}),
where
\begin{equation}\label{eq:WDc-MW_0_1_2}
\begin{split}
  M^{\mathcal{W}}_0&:=\product{\Phi^{\mathcal{F}}[0]}{\Phi^{\mathcal{F}}[0]},\\
  M^{\mathcal{W}}_1&:=\product{\Phi^{\mathcal{F}}[1]}{\Phi^{\mathcal{F}}[1]},\\
  M^{\mathcal{W}}_2&:=\sqrt{I_\mathcal{F}-{M^{\mathcal{W}}_0}^\dag M^{\mathcal{W}}_0-{M^{\mathcal{W}}_1}^\dag M^{\mathcal{W}}_1}
  =I_\mathcal{F}-M^{\mathcal{W}}_0-M^{\mathcal{W}}_1,
\end{split}
\end{equation}
and $I_\mathcal{F}$
denotes
the identity operator acting on $\mathcal{H}_\mathcal{F}$.

Let $c_0$ and $c_1$ be \emph{arbitrary} two non-zero complex numbers such that $\abs{c_0}^2+\abs{c_1}^2=1$.
We then define a state
$\ket{\Psi^{\mathcal{S}+\mathcal{F}}[+]}\in\mathcal{H}_\mathcal{S}\otimes\mathcal{H}_\mathcal{F}$
of the composite system $\mathcal{S}+\mathcal{F}$ by
\begin{equation}\label{eq:EDTE-Def-kPS+F+}
  \ket{\Psi^{\mathcal{S}+\mathcal{F}}[+]}:=c_0\ket{0}\otimes\ket{\Phi^{\mathcal{F}}[0]}+c_1\ket{1}\otimes\ket{\Phi^{\mathcal{F}}[1]}.
\end{equation}
Immediately after the measurement $\mathcal{M}^\mathcal{W}$
performed by the observer $\mathcal{W}$ on the observer $\mathcal{F}$,
the observer $\mathcal{D}$ performs a measurement $\mathcal{M}^\mathcal{D}$
described by the measurement operators $\{M^{\mathcal{D}}_+,M^{\mathcal{D}}_-\}$
over
the composite system $\mathcal{S}+\mathcal{F}$
(according to Postulate~\ref{Born-rule}),
where
\begin{equation}\label{eq:EDTE-Def-MW+_MW-}
\begin{split}
  M^{\mathcal{D}}_+&:=\product{\Psi^{\mathcal{S}+\mathcal{F}}[+]}{\Psi^{\mathcal{S}+\mathcal{F}}[+]},\\
  M^{\mathcal{D}}_-&:=\sqrt{I_{\mathcal{S}+\mathcal{F}}-{M^{\mathcal{D}}_+}^\dag M^{\mathcal{D}}_+}
  =I_{\mathcal{S}+\mathcal{F}}-M^{\mathcal{D}}_+,
\end{split}
\end{equation}
and $I_{\mathcal{S}+\mathcal{F}}$
denotes
the identity operator acting on $\mathcal{H}_\mathcal{S}\otimes\mathcal{H}_\mathcal{F}$.

Now, let us assume that
the system $\mathcal{S}$ is initially in a state $\ket{+}$ in the above setting,
where $\ket{+}$ is defined by
\begin{equation}\label{eq:EDTE-k+=f1s2k0+k1}
  \ket{+}:=c_0\ket{0}+c_1\ket{1}.
\end{equation}
Then, we have the following two cases, Case~1 or Case~2,
depending on how we treat the measurement $\mathcal{M}^\mathcal{F}$.
Note that
we make the following argument regarding both Case~1 and Case~2
in a \emph{heuristic} manner
in terms of the conventional quantum mechanics:

\paragraph{Case 1.}

\emph{Applying Postulate~\ref{Born-rule} to the measurement $\mathcal{M}^\mathcal{F}$},
either the following two possibilities (i) or (ii) occurs
immediately after the measurement $\mathcal{M}^\mathcal{F}$:
\begin{enumerate}
\item The system $\mathcal{S}$ and the observer $\mathcal{F}$ are in the states
  $\ket{0}$ and $\ket{\Phi^{\mathcal{F}}[0]}$, respectively.
  Therefore,
  the state of the composite system $\mathcal{S}+\mathcal{F}$
  is
  $\ket{0}\otimes\ket{\Phi^{\mathcal{F}}[0]}$.
\item The system $\mathcal{S}$ and the observer $\mathcal{F}$ are in the states
  $\ket{1}$ and $\ket{\Phi^{\mathcal{F}}[1]}$, respectively.
  Therefore,
  the state of the composite system
  $\mathcal{S}+\mathcal{F}$
  is
  $\ket{1}\otimes\ket{\Phi^{\mathcal{F}}[1]}$.
\end{enumerate}
Here, the possibilities (i) and (ii) occur with probabilities $\abs{c_0}^2$ and $\abs{c_1}^2$, respectively.
Then, applying Postulate~\ref{Born-rule} to the subsequent measurement $\mathcal{M}^\mathcal{W}$,
we see that \emph{the
measurement $\mathcal{M}^\mathcal{W}$ gives
the outcome $0$ with certainty in the case of the occurrence of the possibility~(i)},
and \emph{it gives the outcome $1$ with certainty in the case of the occurrence of the possibility~(ii)}.
The
measurement $\mathcal{M}^\mathcal{W}$ does not change the state of
the composite system $\mathcal{S}+\mathcal{F}$,
whichever of the two possibilities (i) or (ii) above occurs
in the measurement $\mathcal{M}^\mathcal{F}$.
Then, applying Postulate~\ref{Born-rule} to
the further subsequent measurement $\mathcal{M}^\mathcal{D}$,
we see that
\emph{the
measurement $\mathcal{M}^\mathcal{D}$ gives the outcomes $+$ with probabilities $\abs{c_0}^2$ and $\abs{c_1}^2$
if the possibilities (i) and (ii) above occur
in the measurement $\mathcal{M}^\mathcal{F}$, respectively}.

Thus, in this Case~1,
we can consider the joint probability $P_{\mathrm{c1}}(k,l,m)$
that the measurement $\mathcal{M}^\mathcal{F}$ gives the outcome $k$ and then
the measurement $\mathcal{M}^\mathcal{W}$ gives the outcome $l$ and then
the measurement $\mathcal{M}^\mathcal{D}$ gives the outcome $m$
for each $k=0,1$, $l=0,1,2$, and $m=+,-$.
We then have that
$P_{\mathrm{c1}}(k,l,+)=\abs{c_k}^2\delta_{k,l}\abs{c_k}^2$ and
$P_{\mathrm{c1}}(k,l,-)=\abs{c_k}^2\delta_{k,l}\abs{c_{1-k}}^2$
for every $k=0,1$ and $l=0,1,2$.
That is, we have that
\begin{equation}\label{eq:WDc-Pc1klm}
  P_{\mathrm{c1}}(k,l,m)=\abs{c_k}^2\delta_{k,l}(\delta_{m,+}\abs{c_k}^2+\delta_{m,-}\abs{c_{1-k}}^2)
\end{equation}
for every $k=0,1$, $l=0,1,2$, and $m=+,-$.

Then, whichever of the two possibilities (i) or (ii) above occurs in the measurement $\mathcal{M}^\mathcal{F}$, the
state of the composite system $\mathcal{S}+\mathcal{F}$
immediately after the measurement $\mathcal{M}^\mathcal{D}$ is given by
\[
  \ket{\Psi^{\mathcal{S}+\mathcal{F}}[m]}
\]
if the measurement $\mathcal{M}^\mathcal{D}$ gives an outcome $m\in\{+,-\}$,
where the state
$\ket{\Psi^{\mathcal{S}+\mathcal{F}}[-]}\in\mathcal{H}_\mathcal{S}\otimes\mathcal{H}_\mathcal{F}$
of the composite system $\mathcal{S}+\mathcal{F}$ is defined by
\begin{equation}\label{eq:EDTE-Def-kPS+F-}
  \ket{\Psi^{\mathcal{S}+\mathcal{F}}[-]}:=\overline{c_1}\ket{0}\otimes\ket{\Phi^{\mathcal{F}}[0]}-\overline{c_0}\ket{1}\otimes\ket{\Phi^{\mathcal{F}}[1]}.
\end{equation}

\paragraph{Case 2.}

\emph{According to the unitary time-evolution~\eqref{eq:EUD-MF-S}},
immediately after the measurement $\mathcal{M}^\mathcal{F}$,
the state of the composite system $\mathcal{S}+\mathcal{F}$ is
\begin{equation*}%
  \ket{\Psi^{\mathcal{S}+\mathcal{F}}[+]}.
\end{equation*}
Then,
applying Postulate~\ref{Born-rule} to
the subsequent measurement $\mathcal{M}^\mathcal{W}$,
we see that the measurement $\mathcal{M}^\mathcal{W}$ results in
either the following two possibilities (i) or (ii):
\begin{enumerate}
\item
 The measurement $\mathcal{M}^\mathcal{W}$ gives the outcome $0$,
  and the post-measurement state of the composite system $\mathcal{S}+\mathcal{F}$ is
  $\ket{0}\otimes\ket{\Phi^{\mathcal{F}}[0]}$.
\item
 The measurement $\mathcal{M}^\mathcal{W}$ gives the outcome $1$,
  and the post-measurement state of the composite system $\mathcal{S}+\mathcal{F}$ is
  $\ket{1}\otimes\ket{\Phi^{\mathcal{F}}[1]}$.
\end{enumerate}
Here, the possibilities (i) and (ii) occur with probabilities $\abs{c_0}^2$ and $\abs{c_1}^2$, respectively.
Then, applying Postulate~\ref{Born-rule} to
the further subsequent measurement $\mathcal{M}^\mathcal{D}$,
we see that
\emph{the measurement $\mathcal{M}^\mathcal{D}$ gives the outcomes $+$ with probabilities $\abs{c_0}^2$ and $\abs{c_1}^2$
if the possibilities (i) and (ii) above occur
in the measurement $\mathcal{M}^\mathcal{W}$, respectively}.

Thus,
in this Case~2,
we can consider the joint probability $P_{\mathrm{c2}}(l,m)$
that the measurement $\mathcal{M}^\mathcal{W}$ gives the outcome $l$ and then
the measurement $\mathcal{M}^\mathcal{D}$ gives the outcome $m$
for each $l=0,1,2$ and $m=+,-$.
We then have that
$P_{\mathrm{c2}}(l,+)=\abs{c_l}^2\abs{c_l}^2$ and
$P_{\mathrm{c2}}(l,-)=\abs{c_l}^2\abs{c_{1-l}}^2$
for every $l=0,1$.
That is, we have that
\begin{equation}\label{eq:WDc-Pc2lm}
  P_{\mathrm{c2}}(l,m)=\abs{c_l}^2(\delta_{m,+}\abs{c_l}^2+\delta_{m,-}\abs{c_{1-l}}^2)
\end{equation}
for every $l=0,1$ and $m=+,-$.
Moreover, we have that $P_{\mathrm{c2}}(2,m)=0$ for every $m=+,-$.
On the other hand, the notion of the probability that
the measurement $\mathcal{M}^\mathcal{F}$ gives a certain outcome
is meaningless in this Case~2.

In the same manner as in Case~1, whichever of the two possibilities (i) or (ii) above occurs in the measurement $\mathcal{M}^\mathcal{W}$, the
state of the composite system $\mathcal{S}+\mathcal{F}$
immediately after the measurement $\mathcal{M}^\mathcal{D}$ is given by
\[
  \ket{\Psi^{\mathcal{S}+\mathcal{F}}[m]}
\]
if the measurement $\mathcal{M}^\mathcal{D}$ gives an outcome $m\in\{+,-\}$.

\bigskip\smallskip

Thus,
if we ignore the measurement results of the measurement $\mathcal{M}^\mathcal{F}$
in Case~1,
there is no difference between Case~1 and Case~2
regarding the measurements $\mathcal{M}^\mathcal{W}$ and $\mathcal{M}^\mathcal{D}$.
Actually, we can see that
\[
  \sum_{k=0,1}P_{\mathrm{c1}}(k,l,m)=\abs{c_l}^2(\delta_{m,+}\abs{c_l}^2+\delta_{m,-}\abs{c_{1-l}}^2)=P_{\mathrm{c2}}(l,m)
\]
for every $l=0,1$ and $m=+,-$,
and that
\[
  \sum_{k=0,1}P_{\mathrm{c1}}(k,2,m)=0=P_{\mathrm{c2}}(2,m)
\]
for every $m=+,-$.

The above argument regarding Case~1 and Case~2 is a \emph{heuristic} one
in terms of the conventional quantum mechanics.
In the next section,
we refine and reformulate the description and analysis of the Wigner-Deutsch collaboration
in our rigorous framework of quantum mechanics based on the principle of typicality.

\section{Analysis of the Wigner-Deutsch collaboration based on the principle of typicality}
\label{sec:Analysis of the Wigner-Deutsch collaboration}

In this section, we make an analysis of
the Wigner-Deutsch collaboration,
which is
described in the preceding section
in the terminology of the conventional quantum mechanics,
in terms of our framework of quantum mechanics based on
Postulates~\ref{POT}, \ref{CF}, and \ref{Recursive Use},
together with Postulates~\ref{state_space}, \ref{composition}, and \ref{evolution}.

To proceed this program,
first of all
\emph{we have to implement everything, i.e.,
the
three
measurements $\mathcal{M}^\mathcal{F}$, $\mathcal{M}^\mathcal{W}$,
and $\mathcal{M}^\mathcal{D}$
in the setting of the Wigner-Deutsch collaboration,
by unitary time-evolution}.
The measurement process of the measurement
$\mathcal{M}^\mathcal{F}$
is already described by the unitary operator $U_\mathcal{F}$
given in \eqref{eq:EUD-MF-S}.
For applying the principle of typicality,
the observers $\mathcal{W}$ and $\mathcal{D}$
are regarded as quantum systems
with state spaces $\mathcal{H}_\mathcal{W}$ and $\mathcal{H}_\mathcal{D}$, respectively.
Then, on the one hand, according to \eqref{single_measurement},
the measurement process of
the measurement
$\mathcal{M}^\mathcal{W}$
is described by a unitary operator $U_\mathcal{W}$
acting on
$\mathcal{H}_\mathcal{F}\otimes\mathcal{H}_\mathcal{W}$
such that
\begin{equation}\label{eq:UE-MW-F}
  U_\mathcal{W}(\ket{\Psi}\otimes\ket{\Phi^{\mathcal{W}}_{\mathrm{init}}})
  =\sum_{l=0,1,2}(M^{\mathcal{W}}_l\ket{\Psi})\otimes\ket{\Phi^{\mathcal{W}}[l]}
\end{equation}
for every $\ket{\Psi}\in\mathcal{H}_\mathcal{F}$.
The unitary operator
$U_\mathcal{W}$ describes the interaction
between the observer $\mathcal{F}$ and the observer $\mathcal{W}$
where both the observers are regarded as a quantum system.
The vector $\ket{\Phi^{\mathcal{W}}_{\mathrm{init}}}\in\mathcal{H}_\mathcal{W}$ is
the initial state of the
observer
$\mathcal{W}$, and $\ket{\Phi^{\mathcal{W}}[l]}\in\mathcal{H}_\mathcal{W}$ is
a final state of the
observer
$\mathcal{W}$ for each $l=0,1,2$,
with $\braket{\Phi^{\mathcal{W}}[l]}{\Phi^{\mathcal{W}}[l']}=\delta_{l,l'}$.
For each $l=0,1$,
the state $\ket{\Phi^{\mathcal{W}}[l]}$ indicates that
\emph{the observer $\mathcal{W}$ knows that
the observer $\mathcal{F}$ records the value $l$}.
On the other hand,
according to \eqref{single_measurement} again,
the measurement process of
the measurement
$\mathcal{M}^\mathcal{D}$
is described by a unitary operator $U_\mathcal{D}$
acting on
$\mathcal{H}_\mathcal{S}\otimes\mathcal{H}_\mathcal{F}\otimes\mathcal{H}_\mathcal{W}\otimes\mathcal{H}_\mathcal{D}$
such that
\begin{equation}\label{eq:UED-MW-F}
  U_\mathcal{D}(\ket{\Psi}\otimes\ket{\Phi}\otimes\ket{\Phi^{\mathcal{D}}_{\mathrm{init}}})
  =\sum_{m=+,-}(M^{\mathcal{D}}_m\ket{\Psi})\otimes\ket{\Phi}\otimes\ket{\Phi^{\mathcal{D}}[m]}
\end{equation}
for every $\ket{\Psi}\in\mathcal{H}_\mathcal{S}\otimes\mathcal{H}_\mathcal{F}$
and $\ket{\Phi}\in\mathcal{H}_\mathcal{W}$.
The unitary operator
$U_\mathcal{D}$ describes
the interaction
between
the composite system $\mathcal{S}+\mathcal{F}$
and the observer $\mathcal{D}$
where both
are regarded as a quantum system.
The vector $\ket{\Phi^{\mathcal{D}}_{\mathrm{init}}}\in\mathcal{H}_\mathcal{D}$ is
the initial state of the
observer
$\mathcal{D}$, and $\ket{\Phi^{\mathcal{D}}[m]}\in\mathcal{H}_\mathcal{D}$ is
a final state of the
observer
$\mathcal{D}$ for each $m=+,-$,
with $\braket{\Phi^{\mathcal{D}}[m]}{\Phi^{\mathcal{D}}[m']}=\delta_{m,m'}$.
For each $m=+,-$, the state $\ket{\Phi^{\mathcal{D}}[m]}$ indicates that
\emph{the observer $\mathcal{D}$ records the value $m$}.

Now,
for each $\ket{\Psi}\in\mathcal{H}_\mathcal{S}$,
the sequential applications of $U_{\mathcal{F}}$ and $U_{\mathcal{W}}$
to the state $\ket{\Psi}\otimes\ket{\Phi^{\mathcal{F}}_{\mathrm{init}}}\otimes\ket{\Phi^{\mathcal{W}}_{\mathrm{init}}}$
of the composite system $\mathcal{S}+\mathcal{F}+\mathcal{W}$,
consisting of the system $\mathcal{S}$ and the two observers $\mathcal{F}$ and $\mathcal{W}$,
result in
the state:
\begin{equation}\label{eq:ED-WignerFriend-all}
\begin{split}
  &((I_\mathcal{S}\otimes U_{\mathcal{W}})\circ(U_{\mathcal{F}}\otimes I_\mathcal{W}))
  (\ket{\Psi}\otimes\ket{\Phi^{\mathcal{F}}_{\mathrm{init}}}\otimes\ket{\Phi^{\mathcal{W}}_{\mathrm{init}}}) \\
  &=\sum_{k=0,1}(M^{\mathcal{F}}_k\ket{\Psi})\otimes
      \sum_{l=0,1,2}(M^{\mathcal{W}}_l\ket{\Phi^{\mathcal{F}}[k]})
      \otimes\ket{\Phi^{\mathcal{W}}[l]} \\
  &=\sum_{k=0,1}(M^{\mathcal{F}}_k\ket{\Psi})\otimes\ket{\Phi^{\mathcal{F}}[k]}
      \otimes\ket{\Phi^{\mathcal{W}}[k]},
\end{split}
\end{equation}
where the first equality follows from \eqref{eq:EUD-MF-S} and \eqref{eq:UE-MW-F},
the last equality follows from \eqref{eq:WDc-MW_0_1_2},
and $I_\mathcal{S}$ and $I_\mathcal{W}$ denote the identity operators acting on
$\mathcal{H}_\mathcal{S}$ and $\mathcal{H}_\mathcal{W}$, respectively.
Therefore, using \eqref{eq:EDTE-MF0=00_MF1=11} and \eqref{eq:UED-MW-F} we have that
\begin{equation}\label{eq:WDc-ketk}
\begin{split}
  &(U_{\mathcal{D}}\circ(I_\mathcal{S}\otimes U_{\mathcal{W}}\otimes I_\mathcal{D})\circ(U_{\mathcal{F}}\otimes I_\mathcal{W}\otimes I_\mathcal{D}))
  (\ket{k}\otimes\ket{\Phi^{\mathcal{F}}_{\mathrm{init}}}\otimes\ket{\Phi^{\mathcal{W}}_{\mathrm{init}}}\otimes\ket{\Phi^{\mathcal{D}}_{\mathrm{init}}}) \\
  &=U_{\mathcal{D}}(\ket{k}\otimes\ket{\Phi^{\mathcal{F}}[k]}
      \otimes\ket{\Phi^{\mathcal{W}}[k]}\otimes\ket{\Phi^{\mathcal{D}}_{\mathrm{init}}}) \\
  &=\sum_{m=+,-}(M^{\mathcal{D}}_m(\ket{k}\otimes\ket{\Phi^{\mathcal{F}}[k]}))\otimes\ket{\Phi^{\mathcal{W}}[k]}\otimes\ket{\Phi^{\mathcal{D}}[m]}
\end{split}
\end{equation}
for each $k=0,1$,
where $I_\mathcal{D}$ denotes the identity operator acting on $\mathcal{H}_\mathrm{D}$.
On the other hand,
based on
\eqref{eq:EDTE-Def-kPS+F+} and \eqref{eq:EDTE-Def-kPS+F-} we note that
\begin{align}
  \ket{0}\otimes\ket{\Phi^{\mathcal{F}}[0]}
  &=\overline{c_0}\ket{\Psi^{\mathcal{S}+\mathcal{F}}[+]}+c_1\ket{\Psi^{\mathcal{S}+\mathcal{F}}[-]}, \label{eq:EDTE-kPS+F-inv-0} \\
  \ket{1}\otimes\ket{\Phi^{\mathcal{F}}[1]}
  &=\overline{c_1}\ket{\Psi^{\mathcal{S}+\mathcal{F}}[+]}-c_0\ket{\Psi^{\mathcal{S}+\mathcal{F}}[-]}. \label{eq:EDTE-kPS+F-inv-1}
\end{align}
Based on
\eqref{eq:EDTE-Def-MW+_MW-}, \eqref{eq:EDTE-Def-kPS+F+}, and \eqref{eq:EDTE-Def-kPS+F-} we also note that
\begin{equation}\label{eq:WDc-MDm-kPS+F=d-kPS+F}
  M^{\mathcal{D}}_m\ket{\Psi^{\mathcal{S}+\mathcal{F}}[m']}=\delta_{m,m'}\ket{\Psi^{\mathcal{S}+\mathcal{F}}[m']}
\end{equation}
for every $m,m'=+,-$.
Thus,
we have that
\begin{equation}\label{eq:Wigner-Deutsch-ket0}
\begin{split}
  &(U_{\mathcal{D}}\circ(I_\mathcal{S}\otimes U_{\mathcal{W}}\otimes I_\mathcal{D})\circ(U_{\mathcal{F}}\otimes I_\mathcal{W}\otimes I_\mathcal{D}))
  (\ket{0}\otimes\ket{\Phi^{\mathcal{F}}_{\mathrm{init}}}\otimes\ket{\Phi^{\mathcal{W}}_{\mathrm{init}}}\otimes\ket{\Phi^{\mathcal{D}}_{\mathrm{init}}}) \\
  &=\overline{c_0}\ket{\Psi^{\mathcal{S}+\mathcal{F}}[+]}\otimes\ket{\Phi^{\mathcal{W}}[0]}\otimes\ket{\Phi^{\mathcal{D}}[+]}+c_1\ket{\Psi^{\mathcal{S}+\mathcal{F}}[-]}\otimes\ket{\Phi^{\mathcal{W}}[0]}\otimes\ket{\Phi^{\mathcal{D}}[-]} \\
  &=
  \overline{c_0}c_0\ket{0}\otimes\ket{\Phi^{\mathcal{F}}[0]}\otimes\ket{\Phi^{\mathcal{W}}[0]}\otimes\ket{\Phi^{\mathcal{D}}[+]}
  +\overline{c_0}c_1\ket{1}\otimes\ket{\Phi^{\mathcal{F}}[1]}\otimes\ket{\Phi^{\mathcal{W}}[0]}\otimes\ket{\Phi^{\mathcal{D}}[+]} \\
  &\quad+c_1\overline{c_1}\ket{0}\otimes\ket{\Phi^{\mathcal{F}}[0]}\otimes\ket{\Phi^{\mathcal{W}}[0]}\otimes\ket{\Phi^{\mathcal{D}}[-]}
  -c_1\overline{c_0}\ket{1}\otimes\ket{\Phi^{\mathcal{F}}[1]}\otimes\ket{\Phi^{\mathcal{W}}[0]}\otimes\ket{\Phi^{\mathcal{D}}[-]},
\end{split}
\end{equation}
where the first equality follows from \eqref{eq:WDc-ketk},
\eqref{eq:EDTE-kPS+F-inv-0}, and \eqref{eq:WDc-MDm-kPS+F=d-kPS+F},
and the last equality follows from \eqref{eq:EDTE-Def-kPS+F+} and \eqref{eq:EDTE-Def-kPS+F-}.
Similarly,
we have that
\begin{equation}\label{eq:Wigner-Deutsch-ket1}
\begin{split}
  &(U_{\mathcal{D}}\circ(I_\mathcal{S}\otimes U_{\mathcal{W}}\otimes I_\mathcal{D})\circ(U_{\mathcal{F}}\otimes I_\mathcal{W}\otimes I_\mathcal{D}))
  (\ket{1}\otimes\ket{\Phi^{\mathcal{F}}_{\mathrm{init}}}\otimes\ket{\Phi^{\mathcal{W}}_{\mathrm{init}}}\otimes\ket{\Phi^{\mathcal{D}}_{\mathrm{init}}})\\
  &=\overline{c_1}\ket{\Psi^{\mathcal{S}+\mathcal{F}}[+]}\otimes\ket{\Phi^{\mathcal{W}}[1]}\otimes\ket{\Phi^{\mathcal{D}}[+]}-c_0\ket{\Psi^{\mathcal{S}+\mathcal{F}}[-]}\otimes\ket{\Phi^{\mathcal{W}}[1]}\otimes\ket{\Phi^{\mathcal{D}}[-]} \\
  &=
  \overline{c_1}c_0\ket{0}\otimes\ket{\Phi^{\mathcal{F}}[0]}\otimes\ket{\Phi^{\mathcal{W}}[1]}\otimes\ket{\Phi^{\mathcal{D}}[+]}
  +\overline{c_1}c_1\ket{1}\otimes\ket{\Phi^{\mathcal{F}}[1]}\otimes\ket{\Phi^{\mathcal{W}}[1]}\otimes\ket{\Phi^{\mathcal{D}}[+]} \\
  &\quad-c_0\overline{c_1}\ket{0}\otimes\ket{\Phi^{\mathcal{F}}[0]}\otimes\ket{\Phi^{\mathcal{W}}[1]}\otimes\ket{\Phi^{\mathcal{D}}[-]}
  +c_0\overline{c_0}\ket{1}\otimes\ket{\Phi^{\mathcal{F}}[1]}\otimes\ket{\Phi^{\mathcal{W}}[1]}\otimes\ket{\Phi^{\mathcal{D}}[-]},
\end{split}
\end{equation}
where the first equality follows from \eqref{eq:WDc-ketk},
\eqref{eq:EDTE-kPS+F-inv-1}, and \eqref{eq:WDc-MDm-kPS+F=d-kPS+F},
and the last equality follows from \eqref{eq:EDTE-Def-kPS+F+} and \eqref{eq:EDTE-Def-kPS+F-}.
Then, on the one hand, we set
\begin{equation}\label{eq:WDc-def-P0-P1}
  P_0:=\product{0}{0}\quad\text{and}\quad P_1:=\product{1}{1}.
\end{equation}
On the other hand, we define $\ket{-}\in\mathcal{H}_{\mathcal{S}}$ by
\begin{equation}\label{eq:WDc-def-ket-}
  \ket{-}:=\overline{c_1}\ket{0}-\overline{c_0}\ket{1},
\end{equation}
and then we set
\begin{equation}\label{eq:WDc-def-Q+-Q-}
  Q_+:=\product{+}{+}\quad\text{and}\quad Q_-:=\product{-}{-}.
\end{equation}
Obviously, we then see that $\ket{+}$ and $\ket{-}$ form an orthonormal basis of $\mathcal{H}_\mathcal{S}$, and that
\begin{equation}\label{eq:WDc-P0+P1=I-Q++Q-=I}
  P_0+P_1=I_\mathcal{S}\quad\text{and}\quad Q_+ +Q_- =I_\mathcal{S}.
\end{equation}
It follows from \eqref{eq:WDc-def-Q+-Q-}, \eqref{eq:EDTE-k+=f1s2k0+k1}, and \eqref{eq:WDc-def-ket-} that
\begin{equation}\label{eq:WDc-kkQ+kl-kkQ-kl}
  \bra{k}Q_+\ket{l}=c_k\overline{c_l}\quad\text{and}\quad\bra{k}Q_-\ket{l}=(-1)^{k+l}\overline{c_{1-k}}c_{1-l}
\end{equation}
for every $k=0,1$ and $l=0,1$.
That is, we have that
\begin{equation}\label{eq:WDc-kkQmkl}
  \bra{k}Q_m\ket{l}=\delta_{m,+}c_k\overline{c_l}+\delta_{m,-}(-1)^{k+l}\overline{c_{1-k}}c_{1-l}
\end{equation}
for every $k=0,1$, $l=0,1$, and $m=+,-$.
Thus,
using \eqref{eq:Wigner-Deutsch-ket0}, \eqref{eq:WDc-kkQ+kl-kkQ-kl}, and \eqref{eq:WDc-def-P0-P1} we have that
\begin{equation}\label{eq:Wigner-Deutsch-ket0-2}
\begin{split}
  &(U_{\mathcal{D}}\circ(I_\mathcal{S}\otimes U_{\mathcal{W}}\otimes I_\mathcal{D})\circ(U_{\mathcal{F}}\otimes I_\mathcal{W}\otimes I_\mathcal{D}))
  (\ket{0}\otimes\ket{\Phi^{\mathcal{F}}_{\mathrm{init}}}\otimes\ket{\Phi^{\mathcal{W}}_{\mathrm{init}}}\otimes\ket{\Phi^{\mathcal{D}}_{\mathrm{init}}}) \\
  &=
  P_0Q_+\ket{0}\otimes\ket{\Phi^{\mathcal{F}}[0]}\otimes\ket{\Phi^{\mathcal{W}}[0]}\otimes\ket{\Phi^{\mathcal{D}}[+]}
  +P_1Q_+\ket{0}\otimes\ket{\Phi^{\mathcal{F}}[1]}\otimes\ket{\Phi^{\mathcal{W}}[0]}\otimes\ket{\Phi^{\mathcal{D}}[+]} \\
  &\quad+P_0Q_-\ket{0}\otimes\ket{\Phi^{\mathcal{F}}[0]}\otimes\ket{\Phi^{\mathcal{W}}[0]}\otimes\ket{\Phi^{\mathcal{D}}[-]}
  +P_1Q_-\ket{0}\otimes\ket{\Phi^{\mathcal{F}}[1]}\otimes\ket{\Phi^{\mathcal{W}}[0]}\otimes\ket{\Phi^{\mathcal{D}}[-]}.
\end{split}
\end{equation}
Similarly, using \eqref{eq:Wigner-Deutsch-ket1}, \eqref{eq:WDc-kkQ+kl-kkQ-kl}, and \eqref{eq:WDc-def-P0-P1} we have that
\begin{equation}\label{eq:Wigner-Deutsch-ket1-2}
\begin{split}
  &(U_{\mathcal{D}}\circ(I_\mathcal{S}\otimes U_{\mathcal{W}}\otimes I_\mathcal{D})\circ(U_{\mathcal{F}}\otimes I_\mathcal{W}\otimes I_\mathcal{D}))
  (\ket{1}\otimes\ket{\Phi^{\mathcal{F}}_{\mathrm{init}}}\otimes\ket{\Phi^{\mathcal{W}}_{\mathrm{init}}}\otimes\ket{\Phi^{\mathcal{D}}_{\mathrm{init}}})\\
  &=
  P_0Q_+\ket{1}\otimes\ket{\Phi^{\mathcal{F}}[0]}\otimes\ket{\Phi^{\mathcal{W}}[1]}\otimes\ket{\Phi^{\mathcal{D}}[+]}
  +P_1Q_+\ket{1}\otimes\ket{\Phi^{\mathcal{F}}[1]}\otimes\ket{\Phi^{\mathcal{W}}[1]}\otimes\ket{\Phi^{\mathcal{D}}[+]} \\
  &\quad +P_0Q_-\ket{1}\otimes\ket{\Phi^{\mathcal{F}}[0]}\otimes\ket{\Phi^{\mathcal{W}}[1]}\otimes\ket{\Phi^{\mathcal{D}}[-]}
  +P_1Q_-\ket{1}\otimes\ket{\Phi^{\mathcal{F}}[1]}\otimes\ket{\Phi^{\mathcal{W}}[1]}\otimes\ket{\Phi^{\mathcal{D}}[-]}.
\end{split}
\end{equation}
Hence, it follows from \eqref{eq:Wigner-Deutsch-ket0-2}, \eqref{eq:Wigner-Deutsch-ket1-2}, and \eqref{eq:WDc-def-P0-P1} that
\begin{equation}\label{eq:Wigner-Deutsch-pre-all}
\begin{split}
  &(U_{\mathcal{D}}\circ(I_\mathcal{S}\otimes U_{\mathcal{W}}\otimes I_\mathcal{D})\circ(U_{\mathcal{F}}\otimes I_\mathcal{W}\otimes I_\mathcal{D}))
  (\ket{\Psi}\otimes\ket{\Phi^{\mathcal{F}}_{\mathrm{init}}}\otimes\ket{\Phi^{\mathcal{W}}_{\mathrm{init}}}\otimes\ket{\Phi^{\mathcal{D}}_{\mathrm{init}}})\\
  &=\sum_{k=0,1}\sum_{l=0,1}\sum_{m=+,-}
  (P_k Q_m P_l\ket{\Psi})\otimes\ket{\Phi^{\mathcal{F}}[k]}\otimes\ket{\Phi^{\mathcal{W}}[l]}\otimes\ket{\Phi^{\mathcal{D}}[m]}
\end{split}
\end{equation}
for every $\ket{\Psi}\in\mathcal{H}_\mathcal{S}$.
Actually, we can see from \eqref{eq:Wigner-Deutsch-ket0-2} and \eqref{eq:WDc-def-P0-P1} that the equality \eqref{eq:Wigner-Deutsch-pre-all} holds certainly
in the case of $\ket{\Psi}=\ket{0}$,
and
we can
also
see from \eqref{eq:Wigner-Deutsch-ket1-2} and \eqref{eq:WDc-def-P0-P1} that the equality \eqref{eq:Wigner-Deutsch-pre-all} holds certainly
in the case of $\ket{\Psi}=\ket{1}$.
Thus,
since $\ket{0}$ and $\ket{1}$ form an orthonormal basis of
$\mathcal{H}_\mathcal{S}$,
due to linearity we have that \eqref{eq:Wigner-Deutsch-pre-all} holds for an arbitrary $\ket{\Psi}\in\mathcal{H}_\mathcal{S}$,
as desired.

We denote the set $\{0,1\}\times\{0,1,2\}\times\{+,-\}$ by $\Omega$,
and we define a finite
collection
$\{M_{k,l,m}\}_{(k,l,m)\in\Omega}$ of operators acting on
the state space $\mathcal{H}_\mathcal{S}$ by
\begin{equation}\label{eq:EDTE-Mkl=PkQl}
  M_{k,l,m}:=P_k Q_m P_l
\end{equation}
if $l\neq 2$ and
\begin{equation}\label{eq:EDTE-Mkl=0}
M_{k,l,m}:=0
\end{equation}
otherwise.
Then it follows from \eqref{eq:WDc-P0+P1=I-Q++Q-=I} that
\begin{equation*}
  \sum_{(k,l,m)\in\Omega} M_{k,l,m}^\dag M_{k,l,m}
  =\sum_{l=0,1}\,\sum_{m=+,-}\,\sum_{k=0,1} P_l Q_m P_k Q_m P_l
  =\sum_{l=0,1}\,\sum_{m=+,-} P_l Q_m P_l
  =\sum_{l=0,1} P_l
  =I_\mathcal{S}.
\end{equation*}
Thus, the finite collection $\{M_{k,l,m}\}_{(k,l,m)\in\Omega}$ satisfies the \emph{completeness equation}, and therefore forms \emph{measurement operators}.
We define a unitary operator $U_{\mathrm{whole}}$ acting on $\mathcal{H}_\mathcal{S}\otimes\mathcal{H}_\mathcal{F}\otimes\mathcal{H}_\mathcal{W}\otimes\mathcal{H}_\mathcal{D}$
by
\[
  U_{\mathrm{whole}}:=U_{\mathcal{D}}\circ(I_\mathcal{S}\otimes U_{\mathcal{W}}\otimes I_\mathcal{D})\circ(U_{\mathcal{F}}\otimes I_\mathcal{W}\otimes I_\mathcal{D}).
\]
Then, using \eqref{eq:Wigner-Deutsch-pre-all}, \eqref{eq:EDTE-Mkl=PkQl}, and \eqref{eq:EDTE-Mkl=0} we see that, as a whole,
the sequential applications of $U_{\mathcal{F}}$, $U_{\mathcal{W}}$, and $U_{\mathcal{D}}$
to a state $\ket{\Psi}\otimes\ket{\Phi_{\mathrm{init}}}$ of the composite system $\mathcal{S}+\mathcal{F}+\mathcal{W}+\mathcal{D}$
result in:
\begin{equation}\label{eq:Wigner-Deutsch-all}
\begin{split}
  U_{\mathrm{whole}}(\ket{\Psi}\otimes\ket{\Phi_{\mathrm{init}}})
  &=(U_{\mathcal{D}}\circ(I_\mathcal{S}\otimes U_{\mathcal{W}}\otimes I_\mathcal{D})\circ(U_{\mathcal{F}}\otimes I_\mathcal{W}\otimes I_\mathcal{D}))
    (\ket{\Psi}\otimes\ket{\Phi_{\mathrm{init}}})\\
  &=\sum_{(k,l,m)\in\Omega}(M_{k,l,m}\ket{\Psi})\otimes\ket{\Phi[k,l,m]}
\end{split}
\end{equation}
for each state $\ket{\Psi}\in\mathcal{H}_\mathcal{S}$ of the system $\mathcal{S}$,
where $\ket{\Phi_{\mathrm{init}}}:=
\ket{\Phi^{\mathcal{F}}_{\mathrm{init}}}\otimes\ket{\Phi^{\mathcal{W}}_{\mathrm{init}}}\otimes\ket{\Phi^{\mathcal{D}}_{\mathrm{init}}}$
and
$\ket{\Phi[k,l,m]}:=\ket{\Phi^{\mathcal{F}}[k]}\otimes\ket{\Phi^{\mathcal{W}}[l]}\otimes\ket{\Phi^{\mathcal{D}}[m]}$.

\subsection{Application of the principle of typicality}
\label{subsec:the Wigner-Deutsch collaboration_Application of the principle of typicality}

Thus, in the setting of the Wigner-Deutsch collaboration,
the unitary operator $U_{\mathrm{whole}}$ applying to
the initial state $\ket{+}\otimes\ket{\Phi_{\mathrm{init}}}$ describes
the \emph{repeated once} of the infinite repetition of the measurements
where the succession of
the measurements $\mathcal{M}^\mathcal{F}$, $\mathcal{M}^\mathcal{W}$,
and $\mathcal{M}^\mathcal{D}$ in this order
is infinitely repeated.
It follows from \eqref{eq:Wigner-Deutsch-all} that
the application of $U_{\mathrm{whole}}$,
consisting of the sequential applications of $U_{\mathcal{F}}$, $U_{\mathcal{W}}$, and $U_{\mathcal{D}}$ in this order,
can be regarded as a
\emph{single measurement} which is described by
the measurement operators $\{M_{k,l,m}\}_{(k,l,m)\in\Omega}$
and whose all possible outcomes form the set $\Omega$.

Hence, we can apply Definition~\ref{pmrpwst}
to this scenario of the setting of measurements.
Therefore, according to Definition~\ref{pmrpwst},
we can see that a \emph{world} is an infinite sequence over $\Omega$
and the probability measure induced by
the \emph{probability measure representation for the prefixes of worlds}
is a Bernoulli measure $\lambda_P$ on
$\Omega^\infty$,
where $P$ is a finite probability space on $\Omega$ such that
$P(k,l,m)$ is the square of the norm of the vector
\begin{equation*}
  (M_{k,l,m}\ket{+})\otimes\ket{\Phi[k,l,m]}
\end{equation*}
for every $(k,l,m)\in\Omega$.
Here $\Omega$ is the set of all possible records of
the observers $\mathcal{F}$, $\mathcal{W}$, and $\mathcal{D}$
in the \emph{repeated once} of the experiments.
Let us calculate the explicit form of $P(k,l,m)$.
We use $\Omega_{c}$ to denote the set $\{0,1\}\times\{0,1\}\times\{+,-\}$,
which is a proper subset of $\Omega$.
First, using \eqref {eq:EDTE-Mkl=0} we see that
\begin{equation}\label{eq:Wigner-Deutsch_Mklm+Pklm=fklmoss2-0}
  (M_{k,l,m}\ket{+})\otimes\ket{\Phi[k,l,m]}=0
\end{equation}
for each $(k,l,m)\in\Omega\setminus\Omega_c$,
and using \eqref{eq:EDTE-Mkl=PkQl}, \eqref{eq:WDc-def-P0-P1}, \eqref{eq:WDc-kkQmkl}, and \eqref{eq:EDTE-k+=f1s2k0+k1} we see that
\begin{equation}\label{eq:Wigner-Deutsch_Mklm+Pklm=fklmoss2}
\begin{split}
  (M_{k,l,m}\ket{+})\otimes\ket{\Phi[k,l,m]}
  &=(P_kQ_mP_l\ket{+})\otimes\ket{\Phi^{\mathcal{F}}[k]}\otimes\ket{\Phi^{\mathcal{W}}[l]}\otimes\ket{\Phi^{\mathcal{D}}[m]} \\
  &=(\ket{k}\bra{k}Q_m\ket{l}\braket{l}{+})\otimes\ket{\Phi^{\mathcal{F}}[k]}\otimes\ket{\Phi^{\mathcal{W}}[l]}\otimes\ket{\Phi^{\mathcal{D}}[m]} \\
  &=f(k,l,m)\ket{k}\otimes\ket{\Phi^{\mathcal{F}}[k]}\otimes\ket{\Phi^{\mathcal{W}}[l]}\otimes\ket{\Phi^{\mathcal{D}}[m]}
\end{split}
\end{equation}
for each $(k,l,m)\in\Omega_c$,
where
\[
  f(k,l,m):=c_l\bra{k}Q_m\ket{l}=c_l(\delta_{m,+}c_k\overline{c_l}+\delta_{m,-}(-1)^{k+l}\overline{c_{1-k}}c_{1-l}).
\]
Thus, using~\eqref{eq:Wigner-Deutsch_Mklm+Pklm=fklmoss2-0} and
\eqref{eq:Wigner-Deutsch_Mklm+Pklm=fklmoss2} we have that
\begin{equation}\label{eq:Wigner-Deutsch-finitepsP-0}
  P(k,l,m)=0
\end{equation}
for every $(k,l,m)\in\Omega\setminus\Omega_c$, and 
\begin{equation}\label{eq:Wigner-Deutsch-finitepsP}
  P(k,l,m)=\abs{f(k,l,m)}^2=\abs{c_l}^2(\delta_{m,+}\abs{c_k}^2\abs{c_l}^2+\delta_{m,-}\abs{c_{1-k}}^2\abs{c_{1-l}}^2)
\end{equation}
for every $(k,l,m)\in\Omega_c$.

Now, let us apply
Postulate~\ref{POT}, the \emph{principle of typicality},
to the setting of measurements
developed above.
Let $\omega$ be \emph{our world} in the infinite repetition of
the measurement described by the measurement operators $\{M_{k,l,m}\}_{(k,l,m)\in\Omega}$ in the above setting.
This $\omega$
is an infinite sequence over $\Omega$
consisting of records
in the
three
observers $\mathcal{F}$, $\mathcal{W}$, and $\mathcal{D}$
which is being generated by the infinite repetition of the measurement
described by the measurement operators $\{M_{k,l,m}\}_{(k,l,m)\in\Omega}$
in the above setting.
Since the Bernoulli measure $\lambda_P$ on $\Omega^\infty$ is
the probability measure induced by the
probability
measure representation
for the prefixes of
worlds
in the above setting,
it follows from
Postulate~\ref{POT}
that $\omega$ is Martin-L\"{o}f random with respect to the Bernoulli measure $\lambda_{P}$, that is, \emph{$\omega$ is Martin-L\"of $P$-random}.

We use $\alpha$, $\beta$, and $\gamma$
to denote the infinite sequences over $\{0,1\}$, $\{0,1,2\}$, and $\{+,-\}$, respectively,
such that $(\alpha(n),\beta(n),\gamma(n))=\omega(n)$ for every $n\in\N^+$.
Then, since $\omega$ is Martin-L\"of $P$-random, it follows from \eqref{eq:Wigner-Deutsch-finitepsP-0} and
Corollary~\ref{cor:always-positive-probability}
that
\[
  \omega(n)\in\Omega_c
\]
for every $n\in\N^+$,
and therefore
\[
  \beta(n)\neq 2
\]
for every $n\in\N^+$, i.e., $\beta$ is an infinite binary sequence.
Thus, it follows from Postulate~\ref{POT} and \eqref{eq:Wigner-Deutsch-all}
that, in our world $\omega$, for each $n\in\N^+$ the state of
the $n$th composite system $\mathcal{S}+\mathcal{F}+\mathcal{W}+\mathcal{D}$,
i.e., the $n$th copy of a composite system consisting of
the system $\mathcal{S}$ and the three observers $\mathcal{F}$, $\mathcal{W}$, and $\mathcal{D}$,
immediately after the measurement $\mathcal{M}^\mathcal{D}$ is given by
\begin{equation}\label{eq:WDc-state-immafter-MD}
\begin{split}
  &(M_{\alpha(n),\beta(n),\gamma(n)}\ket{+})\otimes\ket{\Phi[\alpha(n),\beta(n),\gamma(n)]} \\
  &=f(\alpha(n),\beta(n),\gamma(n))\ket{\alpha(n)}\otimes\ket{\Phi^{\mathcal{F}}[\alpha(n)]}\otimes\ket{\Phi^{\mathcal{W}}[\beta(n)]}\otimes\ket{\Phi^{\mathcal{D}}[\gamma(n)]},
\end{split}
\end{equation}
up to the normalization factor,
where the equality follows from
\eqref{eq:Wigner-Deutsch_Mklm+Pklm=fklmoss2}.
Since $\omega$ is Martin-L\"of $P$-random,
it follows from Theorem~\ref{FI} and \eqref{eq:Wigner-Deutsch-finitepsP} that
for every $(k,l,m)\in\Omega_c$ it holds that
\begin{equation}\label{eq:WDc-imafterW-by-PoT-LLN}
  \lim_{n\to\infty} \frac{N_{k,l,m}(\rest{\omega}{n})}{n}
  =\abs{c_l}^2(\delta_{m,+}\abs{c_k}^2\abs{c_l}^2+\delta_{m,-}\abs{c_{1-k}}^2\abs{c_{1-l}}^2),
\end{equation}
where $N_{k,l,m}(\rest{\omega}{n})$ denotes the number of the occurrences of $(k,l,m)$
in the prefix of $\omega$ of length $n$.
Thus, it follows from \eqref{eq:WDc-state-immafter-MD} and \eqref{eq:WDc-imafterW-by-PoT-LLN} that, in our world $\omega$, \emph{the following holds for each $(k,l,m)\in\Omega_c$:
In a proportion of $$\abs{c_l}^2(\delta_{m,+}\abs{c_k}^2\abs{c_l}^2+\delta_{m,-}\abs{c_{1-k}}^2\abs{c_{1-l}}^2)$$ out of the infinite repetitions of the experiment,
the state of $\mathcal{S}$ is $\ket{k}$ and
the observers $\mathcal{F}$, $\mathcal{W}$, and $\mathcal{D}$ record
the values $k$, $l$, and $m$, respectively,
immediately after the measurement $\mathcal{M}^\mathcal{D}$}.
Compared with \eqref{eq:WDc-Pc1klm} and \eqref{eq:WDc-Pc2lm}, we see that
\emph{this situation is different from both Case~1 and Case~2
in Section~\ref{sec:Wigner collaborates with Deutsch}},
although the argument regarding both Case~1 and Case~2 given in Section~\ref{sec:Wigner collaborates with Deutsch} is heuristic
and the statements of
both Case~1 and Case~2 in Section~\ref{sec:Wigner collaborates with Deutsch} are
at least operationally vague
since they are described
in terms of the conventional quantum mechanics
where the operational characterization of the notion of probability is not given.

Let us determine the states of
the composite system $\mathcal{S}+\mathcal{F}+\mathcal{W}$
immediately after the measurement $\mathcal{M}^\mathcal{W}$ and
moreover
immediately after the measurement $\mathcal{M}^\mathcal{F}$
in our world $\omega$.
For that purpose, we can apply our general framework
to treat situations
where apparatuses perform measurements over other apparatuses as well as over a normal quantum system only being measured,
developed in Section~\ref{sec:confirming point},
to the setting of measurements in the Wigner-Deutsch collaboration,
developed above.
The current setting of measurements in the Wigner-Deutsch collaboration is treated within the general framework with $n=3$,
and the measurements $\mathcal{M}^\mathcal{F}$, $\mathcal{M}^\mathcal{W}$, and $\mathcal{M}^\mathcal{D}$ in the current setting serve as the measurements $\mathcal{M}^1$, $\mathcal{M}^2$, and $\mathcal{M}^3$ in the general framework, respectively.
The single qubit system $\mathcal{S}$, and the observers $\mathcal{F}$, $\mathcal{W}$, and $\mathcal{D}$ in the current setting serve as the system $\mathcal{S}$, and the apparatuses $\mathcal{A}_1$, $\mathcal{A}_2$, and $\mathcal{A}_3$ in the general framework, respectively.
Then the measurement operators
$\{M^1_{m_1}\}_{m_1\in\Omega_1}$, $\{M^2_{m_2}\}_{m_2\in\Omega_2}$, and $\{M^3_{m_3}\}_{m_3\in\Omega_3}$
in the general framework are instantiated in the current setting
in the following manner:
\begin{enumerate}
  \item $\Omega_1=\{0,1\}$, $\Omega_2=\{0,1,2\}$, and $\Omega_3=\{+,-\}$.
  \item $M^1_{m_1}=M^{\mathcal{F}}_{m_1}$ for every $m_1\in\Omega_1$, $M^2_{m_2}=I_\mathcal{S}\otimes M^{\mathcal{W}}_{m_2}$ for every $m_2\in\Omega_2$, and
    $M^3_{m_3}=M^{\mathcal{D}}_{m_3}\otimes I_\mathcal{W}$ for every $m_3\in\Omega_3$.
\end{enumerate}
Therefore,
based on \eqref{eq:WDc-MW_0_1_2} and \eqref{eq:EDTE-Def-MW+_MW-},
it is easy to check that the condition~\eqref{eq:Ui-TSCP-MO} certainly holds
in the current setting of measurements in the Wigner-Deutsch collaboration.
Thus, based on \eqref{eq:Wigner-Deutsch-all}, we see that
the measurement operators $\{M^{1,2,3}_{m_1,m_2,m_3}\}_{(m_1,m_2,m_3)\in\Omega_1\times\Omega_2\times\Omega_3}$
satisfying \eqref{eq:Ui-TSCP-1tok} with $n=k=3$ in the general framework are instantiated in the current setting
in the manner that
\begin{equation*}
  M^{1,2,3}_{m_1,m_2,m_3}=M_{m_1,m_2,m_3}
\end{equation*}
for every $m_1\in\Omega_1$, $m_2\in\Omega_2$, and $m_3\in\Omega_3$,
where the measurement operators $\{M_{k,l,m}\}_{(k,l,m)\in\Omega}$ on the right-hand side are defined by \eqref{eq:EDTE-Mkl=PkQl} and \eqref{eq:EDTE-Mkl=0}.
Based on \eqref{eq:ED-WignerFriend-all}, we also see that
the measurement operators $\{M^{1,2}_{m_1,m_2}\}_{(m_1,m_2)\in\Omega_1\times\Omega_2}$
satisfying \eqref{eq:Ui-TSCP-1tok} with $n=3$ and $k=2$ in the general framework are instantiated in the current setting
in the manner that
\begin{equation*}
  M^{1,2}_{m_1,m_2}=\delta_{m_1,m_2}M^{\mathcal{F}}_{m_1}
\end{equation*}
for every $m_1\in\Omega_1$ and $m_2\in\Omega_2$,
where the measurement operators $\{M^{\mathcal{F}}_0,M^{\mathcal{F}}_1\}$ on the right-hand side are defined by \eqref{eq:EDTE-MF0=00_MF1=11}.
The state $\ket{\Psi^0_{\mathrm{init}}}$ of the system $\mathcal{S}$
immediately before the measurement $\mathcal{M}^1$ in the general framework
materializes as
the initial state $\ket{+}$ of the single qubit system $\mathcal{S}$ in the current setting,
defined by \eqref{eq:EDTE-k+=f1s2k0+k1}.

First, we use Postulate~\ref{CF} to
determine the state of the observer $\mathcal{W}$
immediately after the measurement $\mathcal{M}^\mathcal{W}$
in our world $\omega$.
According to Definition~\ref{def:unchanged-apparatus},
it is easy to check that 
the final states of the observer $\mathcal{W}$ are
unchanged after the measurement by
the observer $\mathcal{D}$
and therefore the final states of the observer $\mathcal{W}$ are confirmed
before the measurement by the observer $\mathcal{D}$,
where ``apparatus'' in Definition~\ref{def:unchanged-apparatus} should read ``observer''.
Using \eqref{eq:WDc-state-immafter-MD}, we see that, in our world $\omega$,
for every $n\in\N^+$ the state of the observer $\mathcal{W}$
immediately after the measurement $\mathcal{M}^\mathcal{D}$
is $\ket{\Phi^{\mathcal{W}}[\beta(n)]}$ in the $n$th composite system $\mathcal{S}+\mathcal{F}+\mathcal{W}+\mathcal{D}$.
Thus, it follows from Postulate~\ref{CF}~(ii) that, in our world $\omega$,
for every $n\in\N^+$ the state of the observer $\mathcal{W}$
immediately after the measurement $\mathcal{M}^\mathcal{W}$ is
\begin{equation*}%
  \ket{\Phi^{\mathcal{W}}[\beta(n)]}
\end{equation*}
in the $n$th composite system $\mathcal{S}+\mathcal{F}+\mathcal{W}+\mathcal{D}$.

We can use Postulate~\ref{Recursive Use} to determine the \emph{whole} state of
the composite system $\mathcal{S}+\mathcal{F}+\mathcal{W}$
immediately after the measurement $\mathcal{M}^\mathcal{W}$
in our world $\omega$,
instead of using Postulate~\ref{CF}.
First note that
each of the measurement operators $\{M^{\mathcal{F}}_0,M^{\mathcal{F}}_1\}$,
$\{I_\mathcal{S}\otimes M^{\mathcal{W}}_0,I_\mathcal{S}\otimes M^{\mathcal{W}}_1,I_\mathcal{S}\otimes M^{\mathcal{W}}_2\}$, and
$\{M^{\mathcal{D}}_+\otimes I_\mathcal{W},M^{\mathcal{D}}_-\otimes I_\mathcal{W}\}$ forms a PVM.
On the one hand,
applying $(I_\mathcal{S}\otimes U_{\mathcal{W}}\otimes I_\mathcal{D})\circ(U_{\mathcal{F}}\otimes I_\mathcal{W}\otimes I_\mathcal{D})$
to the initial state
$\ket{+}\otimes\ket{\Phi^{\mathcal{F}}_{\mathrm{init}}}\otimes\ket{\Phi^{\mathcal{W}}_{\mathrm{init}}}\otimes\ket{\Phi^{\mathcal{D}}_{\mathrm{init}}}$
of the repeated once of the experiments
results in a state $\ket{\Psi^{\mathcal{W}}_{\mathrm{Total}}}$ given by
\begin{equation}\label{eq:WDc-virtual-state-Total-W}
\begin{split}
  \ket{\Psi^{\mathcal{W}}_{\mathrm{Total}}}
  &:=((I_\mathcal{S}\otimes U_{\mathcal{W}}\otimes I_\mathcal{D})\circ(U_{\mathcal{F}}\otimes I_\mathcal{W}\otimes I_\mathcal{D}))(\ket{+}\otimes\ket{\Phi^{\mathcal{F}}_{\mathrm{init}}}\otimes\ket{\Phi^{\mathcal{W}}_{\mathrm{init}}}\otimes\ket{\Phi^{\mathcal{D}}_{\mathrm{init}}}) \\
  &=\sum_{k=0,1}(M^{\mathcal{F}}_k\ket{+})\otimes\ket{\Phi^{\mathcal{F}}[k]}
      \otimes\ket{\Phi^{\mathcal{W}}[k]}\otimes\ket{\Phi^{\mathcal{D}}_{\mathrm{init}}} \\
  &=\sum_{k=0,1}c_k\ket{k}\otimes\ket{\Phi^{\mathcal{F}}[k]}
      \otimes\ket{\Phi^{\mathcal{W}}[k]}\otimes\ket{\Phi^{\mathcal{D}}_{\mathrm{init}}},
\end{split}
\end{equation}
where the second equality follows from \eqref{eq:ED-WignerFriend-all}, and the last equality follows from \eqref{eq:EDTE-MF0=00_MF1=11} and \eqref{eq:EDTE-k+=f1s2k0+k1}.
On the other hand,
the final states of the observer $\mathcal{W}$ are unchanged after the measurement by
the observer $\mathcal{D}$, as we saw above.
But
it is easy to see that, according to Definition~\ref{def:unchanged-apparatus},
the final states of the observer $\mathcal{F}$ is \emph{not} unchanged after the measurement by
the observer $\mathcal{D}$.
Here ``apparatus'' in Definition~\ref{def:unchanged-apparatus} should read ``observer''.
Thus, it follows from Postulate~\ref{Recursive Use} and
\eqref{eq:WDc-state-immafter-MD} that, in our world $\omega$,
for each $n\in\N^+$ the state of
the $n$th composite system $\mathcal{S}+\mathcal{F}+\mathcal{W}+\mathcal{D}$
immediately after the measurement $\mathcal{M}^\mathcal{W}$ is given by
\begin{equation}\label{eq:WDc-Proj-virtual-state-Total-W}
  P\left(\ket{\Phi^{\mathcal{W}}[\beta(n)]}\right)\ket{\Psi^{\mathcal{W}}_{\mathrm{Total}}},
\end{equation}
up to the normalization factor, where
\[
  P\left(\ket{\Phi^{\mathcal{W}}[\beta(n)]}\right)
  =I_{\mathcal{S}}\otimes I_{\mathcal{F}}\otimes
  \product{\Phi^{\mathcal{W}}[\beta(n)]}{\Phi^{\mathcal{W}}[\beta(n)]}\otimes I_{\mathcal{D}}.
\]
The vector \eqref{eq:WDc-Proj-virtual-state-Total-W} equals
\begin{equation}\label{eq:WDc-Proj-virtual-state-Total-W-result}
  c_{\beta(n)}\ket{\beta(n)}\otimes\ket{\Phi^{\mathcal{F}}[\beta(n)]}\otimes\ket{\Phi^{\mathcal{W}}[\beta(n)]}\otimes\ket{\Phi^{\mathcal{D}}_{\mathrm{init}}}
\end{equation}
for every $n\in\N^+$,
due to \eqref{eq:WDc-virtual-state-Total-W}.
Hence, in our world $\omega$, for every $n\in\N^+$ 
the state of the composite system $\mathcal{S}+\mathcal{F}+\mathcal{W}$
immediately after the measurement $\mathcal{M}^\mathcal{W}$ is given by
\begin{equation}\label{eq:WDc-SFW-state-immafter-MW}
  \ket{\beta(n)}\otimes\ket{\Phi^{\mathcal{F}}[\beta(n)]}\otimes\ket{\Phi^{\mathcal{W}}[\beta(n)]}
\end{equation}
in the $n$th composite system $\mathcal{S}+\mathcal{F}+\mathcal{W}+\mathcal{D}$.

Now, we can apply
Postulate~\ref{Recursive Use}
\emph{recursively},
based on \eqref{eq:WDc-SFW-state-immafter-MW}.
Namely,
we can use Postulate~\ref{Recursive Use} to determine the state of
the composite system $\mathcal{S}+\mathcal{F}$
immediately after the measurement $\mathcal{M}^\mathcal{F}$
in our world $\omega$.
Recall that
each of the measurement operators $\{M^{\mathcal{F}}_0,M^{\mathcal{F}}_1\}$ and
$\{I_\mathcal{S}\otimes M^{\mathcal{W}}_0,I_\mathcal{S}\otimes M^{\mathcal{W}}_1,I_\mathcal{S}\otimes M^{\mathcal{W}}_2\}$ forms a PVM.
On the one hand, applying $U_\mathcal{F}\otimes I_{\mathcal{W}}$
to the initial state
$\ket{+}\otimes\ket{\Phi^{\mathcal{F}}_{\mathrm{init}}}\otimes\ket{\Phi^{\mathcal{W}}_{\mathrm{init}}}$ in the repeated once of the experiments
where the observer $\mathcal{D}$ is ignored
results in a state $\ket{\Psi^{\mathcal{F}}_{\mathrm{Total}}}$ given by
\begin{equation}\label{eq:WDc-virtual-state-Total-F}
\begin{split}
  \ket{\Psi^{\mathcal{F}}_{\mathrm{Total}}}
  &:=(U_\mathcal{F}\otimes I_{\mathcal{W}})(\ket{+}\otimes\ket{\Phi^{\mathcal{F}}_{\mathrm{init}}}\otimes\ket{\Phi^{\mathcal{W}}_{\mathrm{init}}})
  =\sum_{k=0,1}(M^{\mathcal{F}}_k\ket{+})\otimes\ket{\Phi^{\mathcal{F}}[k]}\otimes\ket{\Phi^{\mathcal{W}}_{\mathrm{init}}} \\
  &=(c_0\ket{0}\otimes\ket{\Phi^{\mathcal{F}}[0]}+c_1\ket{1}\otimes\ket{\Phi^{\mathcal{F}}[1]})\otimes\ket{\Phi^{\mathcal{W}}_{\mathrm{init}}},
\end{split}
\end{equation}
where the second equality follows from \eqref{eq:EUD-MF-S},
and the last equality follows from \eqref{eq:EDTE-MF0=00_MF1=11} and \eqref{eq:EDTE-k+=f1s2k0+k1}.
On the other hand,
according to Definition~\ref{def:unchanged-apparatus},
it is easy to check that 
the final states of the observer $\mathcal{F}$ are
unchanged after the measurement by
the observer $\mathcal{W}$
and therefore the final states of the observer $\mathcal{F}$ are confirmed
before the measurement by the observer $\mathcal{W}$,
where ``apparatus'' in Definition~\ref{def:unchanged-apparatus} should read ``observer''.
Thus, it follows from Postulate~\ref{Recursive Use} and
\eqref{eq:WDc-SFW-state-immafter-MW} that, in our world $\omega$,
for each $n\in\N^+$ the state of
the $n$th composite system $\mathcal{S}+\mathcal{F}+\mathcal{W}$,
i.e., the $n$th copy of a composite system consisting of
the system $\mathcal{S}$ and the two observers $\mathcal{F}$ and $\mathcal{W}$,
immediately after the measurement $\mathcal{M}^\mathcal{F}$ is given by
\begin{equation}\label{eq:WDc-Proj-virtual-state-Total-F}
  P\left(\ket{\Phi^{\mathcal{F}}[\beta(n)]}\right)\ket{\Psi^{\mathcal{F}}_{\mathrm{Total}}},
\end{equation}
up to the normalization factor, where
\[
  P\left(\ket{\Phi^{\mathcal{F}}[\beta(n)]}\right)
  =I_{\mathcal{S}}\otimes
  \product{\Phi^{\mathcal{F}}[\beta(n)]}{\Phi^{\mathcal{F}}[\beta(n)]}\otimes I_{\mathcal{W}}.
\]
The vector \eqref{eq:WDc-Proj-virtual-state-Total-F} equals
\[
  c_{\beta(n)}\ket{\beta(n)}\otimes\ket{\Phi^{\mathcal{F}}[\beta(n)]}\otimes\ket{\Phi^{\mathcal{W}}_{\mathrm{init}}}
\]
for every $n\in\N^+$,
due to \eqref{eq:WDc-virtual-state-Total-F}.
Thus, in our world $\omega$, for every $n\in\N^+$ 
the state of the composite system $\mathcal{S}+\mathcal{F}$
immediately after the measurement $\mathcal{M}^\mathcal{F}$ is given by
\begin{equation}\label{eq:WDc-SF-state-immafter-MF}
  \ket{\beta(n)}\otimes\ket{\Phi^{\mathcal{F}}[\beta(n)]}
\end{equation}
in the $n$th composite system $\mathcal{S}+\mathcal{F}+\mathcal{W}+\mathcal{D}$.
Hence,
in our world $\omega$,
for every $n\in\N^+$
the system $\mathcal{S}$ is in the state $\ket{\beta(n)}$ and
the observer $\mathcal{F}$ surely records the value $\beta(n)$
immediately after the measurement $\mathcal{M}^\mathcal{F}$
in the $n$th repetition of the experiment.
Note that we can also derive the result~\eqref{eq:WDc-SF-state-immafter-MF}
from \eqref{eq:WDc-SFW-state-immafter-MW}
using Postulate~\ref{CF} in the same manner as
we derived the result~\eqref{eq:WF-SF-state-immafter-MF}
in Section~\ref{subsec:WF-Appl-POT},
instead of using Postulate~\ref{Recursive Use}.

It follows from \eqref{eq:WDc-imafterW-by-PoT-LLN} that
\begin{equation*}
\begin{split}
  \lim_{n\to\infty} \frac{N_l(\rest{\beta}{n})}{n}
  &=\lim_{n\to\infty} \frac{1}{n}\sum_{k=0,1}\sum_{m=+,-}N_{k,l,m}(\rest{\omega}{n})
  =\sum_{k=0,1}\sum_{m=+,-}\lim_{n\to\infty} \frac{N_{k,l,m}(\rest{\omega}{n})}{n} \\
  &=\sum_{k=0,1}\sum_{m=+,-}\abs{c_l}^2(\delta_{m,+}\abs{c_k}^2\abs{c_l}^2+\delta_{m,-}\abs{c_{1-k}}^2\abs{c_{1-l}}^2) \\
  &=\abs{c_l}^2(\abs{c_l}^2+\abs{c_{1-l}}^2) \\
  &=\abs{c_l}^2
\end{split}
\end{equation*}
for each $l\in\{0,1\}$,
where $N_l(\sigma)$ denotes the number of the occurrences of $l$ in $\sigma$ for every $l\in\{0,1\}$ and every $\sigma\in\X$.
Thus, the forms of the states \eqref{eq:WDc-SF-state-immafter-MF} and
\eqref{eq:WDc-SFW-state-immafter-MW} imply that,
in our world $\omega$,
for each $l=0,1$
the following (i) and (ii) simultaneously occur
in a proportion of $\abs{c_l}^2$ out of the infinite repetitions of the experiment:
\begin{enumerate}
\item
The state of $\mathcal{S}$ is $\ket{l}$ and
the observer $\mathcal{F}$ records this value $l$,
immediately after the measurement $\mathcal{M}^\mathcal{F}$.
\item
The state of $\mathcal{S}$ is $\ket{l}$ and
both
the observers $\mathcal{F}$ and $\mathcal{W}$ record
this value $l$,
immediately after the measurement $\mathcal{M}^\mathcal{W}$.
\end{enumerate}
We can then see that,
\emph{as far as the measurements $\mathcal{M}^\mathcal{F}$ and $\mathcal{M}^\mathcal{W}$ are concerned, this situation
coincides with Case~1 in Section~\ref{sec:Wigner collaborates with Deutsch}},
although the argument regarding Case~1 given in Section~\ref{sec:Wigner collaborates with Deutsch} is heuristic and
the statements of Case~1 in Section~\ref{sec:Wigner collaborates with Deutsch} are
at least operationally vague
since they are described
in terms of the conventional quantum mechanics
where the operational characterization of the notion of probability is not given.

\begin{remark}[The Validity I of Postulate~\ref{Recursive Use}]\label{rem:Validity1-PRU}
Let us consider the validity of Postulate~\ref{Recursive Use}.
In the statement of Postulate~\ref{Recursive Use},
when acted on the vector $\ket{\Psi^{n-1}_{\mathrm{Total}}}$,
the projector
\begin{equation}\label{eq:WDc-Rem1-Proj-original}
  P(\ket{\Phi^{i_1}[m_{i_1}]},\dots,\ket{\Phi^{i_k}[m_{i_k}]})
\end{equation}
determines the state of the total system $\mathcal{S}_{\mathrm{Total}}$
immediately after the measurement $\mathcal{M}^{n-1}$
in a manner that
the apparatuses $\mathcal{A}_{i_1},\dots,\mathcal{A}_{i_k}$ are
in the states $\ket{\Phi^{i_1}[m_{i_1}]},\dots,\ket{\Phi^{i_k}[m_{i_k}]}$, respectively,
immediately after the measurement $\mathcal{M}^{n-1}$.
Recall that
these apparatuses $\mathcal{A}_{i_1},\dots,\mathcal{A}_{i_k}$ forms precisely the set of
apparatuses among $\mathcal{A}_{1},\dots,\mathcal{A}_{n-1}$ whose final states are
unchanged after the measurement by the apparatus $\mathcal{A}_n$.
If we could replace the projector
$P(\ket{\Phi^{i_1}[m_{i_1}]},\dots,\ket{\Phi^{i_k}[m_{i_k}]})$ by
\begin{equation}\label{eq:WDc-Rem1-Proj}
P(\ket{\Phi^{1}[m_{1}]},\dots,\ket{\Phi^{n-1}[m_{n-1}]})
\end{equation}
in Postulate~\ref{Recursive Use},
we could determine all states of the apparatuses $\mathcal{A}_{1},\dots,\mathcal{A}_{n-1}$
immediately after the measurement $\mathcal{M}^{n-1}$.
However, this is impossible, as we will see in what follows.

In the above analysis of the Wigner-Deutsch collaboration,
the state of the composite system $\mathcal{S}+\mathcal{F}+\mathcal{W}+\mathcal{D}$
immediately after the measurement $\mathcal{M}^\mathcal{W}$
is given by \eqref{eq:WDc-Proj-virtual-state-Total-W} in our world $\omega$.
This is because
the final states of the observer $\mathcal{W}$ is unchanged after the measurement by
the observer $\mathcal{D}$
while the final states of the observer $\mathcal{F}$ is
not unchanged after the measurement by the observer $\mathcal{D}$, as we saw above.
Consequently,
\eqref{eq:WDc-Rem1-Proj-original} materializes as the projector $P\left(\ket{\Phi^{\mathcal{W}}[\beta(n)]}\right)$
in the above analysis of the Wigner-Deutsch collaboration
to determine the whole state of the composite system $\mathcal{S}+\mathcal{F}+\mathcal{W}$ immediately after the measurement $\mathcal{M}^\mathcal{W}$ in our world $\omega$.
In contrast,
the variant~\eqref{eq:WDc-Rem1-Proj} corresponds to the projector $P\left(\ket{\Phi^{\mathcal{F}}[\alpha(n)]},\ket{\Phi^{\mathcal{W}}[\beta(n)]}\right)$
in this analysis of the Wigner-Deutsch collaboration
to determine this whole state.

Let us calculate $P\left(\ket{\Phi^{\mathcal{F}}[\alpha(n)]},\ket{\Phi^{\mathcal{W}}[\beta(n)]}\right)\ket{\Psi^{\mathcal{W}}_{\mathrm{Total}}}$.
Then using \eqref{eq:WDc-virtual-state-Total-W} we have that
\begin{equation}\label{eq:WDc-immafterW-by-PoT-False}
  P\left(\ket{\Phi^{\mathcal{F}}[\alpha(n)]},\ket{\Phi^{\mathcal{W}}[\beta(n)]}\right)\ket{\Psi^{\mathcal{W}}_{\mathrm{Total}}}
  =c_{\alpha(n)}\delta_{\beta(n),\alpha(n)}\ket{\alpha(n)}\otimes\ket{\Phi^{\mathcal{F}}[\alpha(n)]}\otimes\ket{\Phi^{\mathcal{W}}[\beta(n)]}\otimes\ket{\Phi^{\mathcal{D}}_{\mathrm{init}}}
\end{equation}
for every $n\in\N^+$.
However,
it follows from \eqref{eq:WDc-imafterW-by-PoT-LLN} that
$\alpha(n)\neq\beta(n)$
holds for a proportion of $2\abs{c_0}^2\abs{c_1}^2$ out of all $n\in\N^+$.
This is because using \eqref{eq:WDc-imafterW-by-PoT-LLN} we have that
\begin{align*}
  &\lim_{n\to\infty} \frac{1}{n}\sum_{k=0,1}\sum_{m=+,-}N_{k,1-k,m}(\rest{\omega}{n}) \\
  &=\sum_{k=0,1}\sum_{m=+,-}\lim_{n\to\infty} \frac{N_{k,1-k,m}(\rest{\omega}{n})}{n} \\
  &=\sum_{k=0,1}\sum_{m=+,-}\abs{c_{1-k}}^2(\delta_{m,+}\abs{c_k}^2\abs{c_{1-k}}^2+\delta_{m,-}\abs{c_{1-k}}^2\abs{c_{k}}^2) \\
  &=2\sum_{k=0,1}\abs{c_{1-k}}^2\abs{c_k}^2\abs{c_{1-k}}^2 \\
  &=2\abs{c_0}^2\abs{c_1}^2
  >0.
\end{align*}
Therefore, the vector \eqref{eq:WDc-immafterW-by-PoT-False} becomes $0$
for such infinitely many $n$.
Hence, the vector \eqref{eq:WDc-immafterW-by-PoT-False} does not become
the state of the composite system $\mathcal{S}+\mathcal{F}+\mathcal{W}+\mathcal{D}$
immediately after the measurement $\mathcal{M}^\mathcal{W}$ for such infinitely many $n$ in our world $\omega$,
since the state vector must not be a zero-vector.
In contrast,
the vector
$P\left(\ket{\Phi^{\mathcal{W}}[\beta(n)]}\right)\ket{\Psi^{\mathcal{W}}_{\mathrm{Total}}}$,
which the original Postulate~\ref{Recursive Use} requires to be
the state of the composite system $\mathcal{S}+\mathcal{F}+\mathcal{W}+\mathcal{D}$
immediately after the measurement $\mathcal{M}^\mathcal{W}$ in our world $\omega$,
is non-zero for all $n\in\N^+$,
since it equals \eqref{eq:WDc-Proj-virtual-state-Total-W-result} for all $n\in\N^+$.
This observation supports the validity of Postulate~\ref{Recursive Use}.
\qed
\end{remark}

\begin{remark}[The Validity II of Postulate~\ref{Recursive Use}]\label{rem:Validity2-PRU}
Let us consider the validity of Postulate~\ref{Recursive Use}
further.
Postulate~\ref{Recursive Use} can be applied
on the premise that
the measurement operators $\{M^i_{m_i}\}_{m_i\in\Omega_i}$ form a PVM
in $\mathcal{H}_{\mathcal{S}}\otimes\overline{\mathcal{H}}_1\otimes\dots\otimes\overline{\mathcal{H}}_{i-1}$
for every $i=1,\dots,n$.
We will demonstrate the necessity of this premise,
based on the Wigner-Deutsch collaboration.
We
call
a variant of
Postulate~\ref{Recursive Use} where this premise is eliminated
the \emph{naive Postulate~\ref{Recursive Use}}.
We will show that the naive Postulate~\ref{Recursive Use} is too strong
and leads to a contradiction.

On the one hand,
using the naive Postulate~\ref{Recursive Use}
we can derive the result~\eqref{eq:WDc-SF-state-immafter-MF}
from \eqref{eq:WDc-state-immafter-MD}
in the same manner as we derived the result~\eqref{eq:WDc-SF-state-immafter-MF}
from \eqref{eq:WDc-state-immafter-MD} above
using Postulate~\ref{Recursive Use}.
That is, we have that, in our world $\omega$, for every $n\in\N^+$ 
the state of the composite system $\mathcal{S}+\mathcal{F}$
immediately after the measurement $\mathcal{M}^\mathcal{F}$ is given by
\begin{equation}\label{eq:WDc-Rem1-SF-state-immafter-MF}
  \ket{\beta(n)}\otimes\ket{\Phi^{\mathcal{F}}[\beta(n)]}
\end{equation}
in the $n$th composite system $\mathcal{S}+\mathcal{F}+\mathcal{W}+\mathcal{D}$.

On the other hand,
in the Wigner-Deutsch collaboration
we can regard the combination of
the measurements $\mathcal{M}^{\mathcal{W}}$ and $\mathcal{M}^{\mathcal{D}}$ as a \emph{single} measurement $\mathcal{M}^{\mathcal{W}+\mathcal{D}}$,
where we regard
the combination of the observers $\mathcal{W}$ and $\mathcal{D}$
as a \emph{single} observer $\mathcal{W}+\mathcal{D}$.
Actually,
the measurement process of the measurement $\mathcal{M}^{\mathcal{W}+\mathcal{D}}$ is described by
a unitary operator
$U_{\mathcal{W}+\mathcal{D}}$ defined by
\[
  U_{\mathcal{W}+\mathcal{D}}:=U_{\mathcal{D}}\circ(I_{\mathcal{S}}\otimes U_{\mathcal{W}}\otimes I_\mathcal{D}).
\]
We denote the set $\{0,1,2\}\times\{+,-\}$ by $\Theta$,
and we define
a finite collection $\{M^{\mathcal{W}+\mathcal{D}}_{l,m}\}_{(l,m)\in\Theta}$
of operators acting on
the state space $\mathcal{H}_\mathcal{S}\otimes\mathcal{H}_\mathcal{F}$ by
\begin{equation}\label{eq:WDc-Rem1-def-MWD}
  M^{\mathcal{W}+\mathcal{D}}_{l,m}:=M^{\mathcal{D}}_m\circ(I_\mathcal{S}\otimes M^{\mathcal{W}}_l).
\end{equation}
Then using \eqref{eq:UE-MW-F} and \eqref{eq:UED-MW-F} we have that
\begin{equation*}%
  U_{\mathcal{W}+\mathcal{D}}(\ket{\Psi}\otimes\ket{\Phi^{\mathcal{W}}_{\mathrm{init}}}\otimes\ket{\Phi^{\mathcal{D}}_{\mathrm{init}}})
  =\sum_{(l,m)\in\Theta}(M^{\mathcal{W}+\mathcal{D}}_{l,m}\ket{\Psi})\otimes\ket{\Phi^{\mathcal{W}}[l]}\otimes\ket{\Phi^{\mathcal{D}}[m]}
\end{equation*}
for every $\ket{\Psi}\in\mathcal{H}_\mathcal{S}\otimes\mathcal{H}_\mathcal{F}$,
in the form corresponding to \eqref{eq:Ui-TSCP} with $n=k=2$.
Here, it is easy to see that the finite collection
$\{M^{\mathcal{W}+\mathcal{D}}_{l,m}\}_{(l,m)\in\Theta}$ satisfies the completeness equation, and therefore forms measurement operators
(see \eqref{eq:mopsMW+D_lm-cmpleq} in Section~\ref{sec:Analysis of WDc with a mere quantum system F} below).
Moreover, we can check that the measurement operators $\{M^{\mathcal{W}+\mathcal{D}}_{l,m}\}_{(l,m)\in\Theta}$ satisfy the condition~\eqref{eq:Ui-TSCP-MO} with $n=k=2$.

However, the measurement operators $\{M^{\mathcal{W}+\mathcal{D}}_{l,m}\}_{(l,m)\in\Theta}$ do not form a PVM
in $\mathcal{H}_\mathcal{S}\otimes\mathcal{H}_\mathcal{F}$.
To see this, we first note that
\begin{equation}\label{eq:MWDl+=f1sq2kPS+F+blotbPFl}
  M^{\mathcal{W}+\mathcal{D}}_{l,+}
  =\overline{c_l}\ket{\Psi^{\mathcal{S}+\mathcal{F}}[+]}\bra{l}\otimes\bra{\Phi^{\mathcal{F}}[l]}
\end{equation}
holds for each $l=0,1$,
which follows immediately from
\eqref{eq:WDc-Rem1-def-MWD},
\eqref{eq:EDTE-Def-MW+_MW-}, \eqref{eq:WDc-MW_0_1_2}, and \eqref{eq:EDTE-Def-kPS+F+}.
Assume contrarily that the measurement operators $\{M^{\mathcal{W}+\mathcal{D}}_{l,m}\}_{(l,m)\in\Theta}$ form a PVM
in $\mathcal{H}_\mathcal{S}\otimes\mathcal{H}_\mathcal{F}$.
Then $M^{\mathcal{W}+\mathcal{D}}_{0,+}$, in particular, is a projector on $\mathcal{H}_\mathcal{S}\otimes\mathcal{H}_\mathcal{F}$.
Since $M^{\mathcal{W}+\mathcal{D}}_{0,+}\neq 0$
due to \eqref{eq:MWDl+=f1sq2kPS+F+blotbPFl},
there is a normalized vector $\ket{\psi_0}\in\mathcal{H}_\mathcal{S}\otimes\mathcal{H}_\mathcal{F}$ such that $M^{\mathcal{W}+\mathcal{D}}_{0,+}\ket{\psi_0}=\ket{\psi_0}$.
Thus, using \eqref{eq:MWDl+=f1sq2kPS+F+blotbPFl} we have that
\begin{align*}
  1=\bra{\psi_0}M^{\mathcal{W}+\mathcal{D}}_{0,+}\ket{\psi_0}
  =\overline{c_0}\braket{\psi_0}{\Psi^{\mathcal{S}+\mathcal{F}}[+]}(\bra{0}\otimes\bra{\Phi^{\mathcal{F}}[0]})\ket{\psi_0}.
\end{align*}
Therefore, since $\abs{\braket{\psi_0}{\Psi^{\mathcal{S}+\mathcal{F}}[+]}}, \abs{(\bra{0}\otimes\bra{\Phi^{\mathcal{F}}[0]})\ket{\psi_0}}\le 1$,
it follows that $1\le \abs{c_0}$.
However, this contradicts
our
choice of $c_0$ such that $0<\abs{c_0}^2<1$, and
hence
the measurement operators $\{M^{\mathcal{W}+\mathcal{D}}_{l,m}\}_{(l,m)\in\Theta}$ cannot form a PVM
in $\mathcal{H}_\mathcal{S}\otimes\mathcal{H}_\mathcal{F}$.

Now, we can use the naive Postulate~\ref{Recursive Use} to determine the whole states of
the composite system $\mathcal{S}+\mathcal{F}$
immediately after the measurement $\mathcal{M}^\mathcal{F}$
in our world $\omega$,
under the situation that
the combination of
the measurements $\mathcal{M}^{\mathcal{W}}$ and $\mathcal{M}^{\mathcal{D}}$
is regarded as the \emph{single} measurement $\mathcal{M}^{\mathcal{W}+\mathcal{D}}$.
On the one hand,
in the Wigner-Deutsch collaboration where the combination of
the measurements $\mathcal{M}^{\mathcal{W}}$ and $\mathcal{M}^{\mathcal{D}}$ is regarded as the single measurement $\mathcal{M}^{\mathcal{W}+\mathcal{D}}$,
applying $U_\mathcal{F}\otimes I_{\mathcal{W}}\otimes I_{\mathcal{D}}$
to the initial state
$\ket{+}\otimes\ket{\Phi^{\mathcal{F}}_{\mathrm{init}}}\otimes\ket{\Phi^{\mathcal{W}}_{\mathrm{init}}}\otimes\ket{\Phi^{\mathcal{D}}_{\mathrm{init}}}$
in the repeated once of the experiments
results in a state $\ket{\Psi^{\mathcal{F}(\mathcal{W}+\mathcal{D})}_{\mathrm{Total}}}$
given by
\begin{equation}\label{eq:WDc-Rem1-virtual-state-Total-F}
\begin{split}
  \ket{\Psi^{\mathcal{F}(\mathcal{W}+\mathcal{D})}_{\mathrm{Total}}}
  &:=(U_\mathcal{F}\otimes I_{\mathcal{W}}\otimes I_{\mathcal{D}})(\ket{+}\otimes\ket{\Phi^{\mathcal{F}}_{\mathrm{init}}}\otimes\ket{\Phi^{\mathcal{W}}_{\mathrm{init}}}\otimes\ket{\Phi^{\mathcal{D}}_{\mathrm{init}}}) \\
  &=\sum_{k=0,1}(M^{\mathcal{F}}_k\ket{+})\otimes\ket{\Phi^{\mathcal{F}}[k]}\otimes\ket{\Phi^{\mathcal{W}}_{\mathrm{init}}}\otimes\ket{\Phi^{\mathcal{D}}_{\mathrm{init}}} \\
  &=(c_0\ket{0}\otimes\ket{\Phi^{\mathcal{F}}[0]}+c_1\ket{1}\otimes\ket{\Phi^{\mathcal{F}}[1]})\otimes\ket{\Phi^{\mathcal{W}}_{\mathrm{init}}}\otimes\ket{\Phi^{\mathcal{D}}_{\mathrm{init}}},
\end{split}
\end{equation}
where the second equality follows from \eqref{eq:EUD-MF-S}, and
the last equality follows from \eqref{eq:EDTE-MF0=00_MF1=11} and \eqref{eq:EDTE-k+=f1s2k0+k1}.
On the other hand,
it is easy to check that, according to Definition~\ref{def:unchanged-apparatus},
the final states of the observer $\mathcal{F}$ are \emph{not} unchanged after the measurement by
the observer $\mathcal{W}+\mathcal{D}$,
where ``apparatus'' in Definition~\ref{def:unchanged-apparatus} should read ``observer''.
Thus, it follows from the naive Postulate~\ref{Recursive Use}
that, in our world $\omega$, for each $n\in\N^+$ the state of
the $n$th composite system $\mathcal{S}+\mathcal{F}+\mathcal{W}+\mathcal{D}$
immediately after the measurement $\mathcal{M}^\mathcal{F}$ is given by
\begin{equation}\label{eq:WDc-Rem1-Proj-virtual-state-Total-F}
  P\left(\right)\ket{\Psi^{\mathcal{F}(\mathcal{W}+\mathcal{D})}_{\mathrm{Total}}},
\end{equation}
up to the normalization factor, where
$P\left(\right)$
equals
the identity operator acting on the
whole
state space
$\mathcal{H}_\mathcal{S}\otimes\mathcal{H}_\mathcal{F}\otimes\mathcal{H}_\mathcal{W}\otimes\mathcal{H}_\mathcal{D}$
of the composite system $\mathcal{S}+\mathcal{F}+\mathcal{W}+\mathcal{D}$.
Hence, it follows from \eqref{eq:WDc-Rem1-Proj-virtual-state-Total-F} and
\eqref{eq:WDc-Rem1-virtual-state-Total-F} that,
in our world $\omega$, for every $n\in\N^+$ 
the state of the composite system $\mathcal{S}+\mathcal{F}$
immediately after the measurement $\mathcal{M}^\mathcal{F}$ is given by
\begin{equation*}%
  c_0\ket{0}\otimes\ket{\Phi^{\mathcal{F}}[0]}+c_1\ket{1}\otimes\ket{\Phi^{\mathcal{F}}[1]}
\end{equation*}
in the $n$th composite system $\mathcal{S}+\mathcal{F}+\mathcal{W}+\mathcal{D}$.
However, this contradicts the result~\eqref{eq:WDc-Rem1-SF-state-immafter-MF},
which
implies
that the naive Postulate~\ref{Recursive Use} is too strong
to be consistent.

In contrast, we can avoid this contradiction for
the \emph{original}
Postulate~\ref{Recursive Use},
since the measurement operators $\{M^{\mathcal{W}+\mathcal{D}}_{l,m}\}_{(l,m)\in\Theta}$
do not form a PVM in $\mathcal{H}_\mathcal{S}\otimes\mathcal{H}_\mathcal{F}$ and therefore
the above argument does not work for Postulate~\ref{Recursive Use}.
Hence, in Postulate~\ref{Recursive Use},
the premise that the measurement operators $\{M^i_{m_i}\}_{m_i\in\Omega_i}$ form a PVM
in $\mathcal{H}_{\mathcal{S}}\otimes\overline{\mathcal{H}}_1\otimes\dots\otimes\overline{\mathcal{H}}_{i-1}$
for every $i=1,\dots,n$ is necessary.
\qed
\end{remark}

\section{\boldmath Analysis of the Wigner-Deutsch collaboration, with the friend $\mathcal{F}$ being a mere quantum system only measured}
\label{sec:Analysis of WDc with a mere quantum system F}

In the previous section, we have treated the observer $\mathcal{F}$
as a measuring apparatus
in the analysis of the Wigner-Deutsch collaboration
in
our rigorous framework of quantum mechanics based on the principle of typicality.
What happens
if
we regard the observer $\mathcal{F}$ as a mere quantum system
and not as a measurement apparatus
in the analysis of the Wigner-Deutsch collaboration in our framework?
We can make an analysis of this case
more easily in our framework,
as follows.
Based on the result of the analysis,
we will see that this case corresponds to
Case~2 in Section~\ref{sec:Wigner collaborates with Deutsch},
which is described and heuristically analyzed in terms of the conventional quantum mechanics in a vague manner in Section~\ref{sec:Wigner collaborates with Deutsch}.

In this case,
according to \eqref{eq:EUD-MF-S},
the state of the composite system $\mathcal{S}+\mathcal{F}$
immediately after the `measurement' $\mathcal{M}^\mathcal{F}$ is given
by
$U_{\mathcal{F}}(\ket{+}\otimes\ket{\Phi^{\mathcal{F}}_\mathrm{init}})$.
The succession of
the measurement $\mathcal{M}^\mathcal{W}$ and
then the measurement $\mathcal{M}^\mathcal{D}$
performed over the composite system $\mathcal{S}+\mathcal{F}$ in this state
forms
the \emph{repeated once} of the infinite repetition of the measurements of the Wigner-Deutsch collaboration in this case.
Thus, 
for applying Postulate~\ref{POT}, the principle of typicality,
first of all
we have to implement
the measurements $\mathcal{M}^\mathcal{W}$ and $\mathcal{M}^\mathcal{D}$
by unitary time-evolution.
This is already done in the form \eqref{eq:UE-MW-F} of $U_\mathcal{W}$ and
the form \eqref{eq:UED-MW-F} of $U_\mathcal{D}$.

Now, it follows from \eqref{eq:UE-MW-F} and \eqref{eq:UED-MW-F} that
\begin{equation}\label{eq:QSF-ED-WignerFriend-all}
\begin{split}
  &(U_{\mathcal{D}}\circ(I_\mathcal{S}\otimes U_{\mathcal{W}}\otimes I_\mathcal{D}))(\ket{\Psi}\otimes\ket{\Phi^{\mathcal{W}}_{\mathrm{init}}}\otimes\ket{\Phi^{\mathcal{D}}_{\mathrm{init}}}) \\
  &=U_{\mathcal{D}}\sum_{l=0,1,2}((I_\mathcal{S}\otimes M^{\mathcal{W}}_l)\ket{\Psi})\otimes\ket{\Phi^{\mathcal{W}}[l]}\otimes\ket{\Phi^{\mathcal{D}}_{\mathrm{init}}} \\
  &=\sum_{l=0,1,2}\,\sum_{m=+,-}((M^{\mathcal{D}}_m\circ(I_\mathcal{S}\otimes M^{\mathcal{W}}_l))\ket{\Psi})\otimes\ket{\Phi^{\mathcal{W}}[l]}\otimes\ket{\Phi^{\mathcal{D}}[m]}
\end{split}
\end{equation}
for each $\ket{\Psi}\in\mathcal{H}_\mathcal{S}\otimes\mathcal{H}_\mathcal{F}$.

We denote the set $\{0,1,2\}\times\{+,-\}$ by $\Theta$, and we define
a finite collection $\{M^{\mathcal{W}+\mathcal{D}}_{l,m}\}_{(l,m)\in\Theta}$ of operators acting on
the state space $\mathcal{H}_\mathcal{S}\otimes\mathcal{H}_\mathcal{F}$
of the composite system $\mathcal{S}+\mathcal{F}$ by
\begin{equation}\label{eq:QWD-POT-WF-Mkl=dlkMFk}
  M^{\mathcal{W}+\mathcal{D}}_{l,m}:=M^{\mathcal{D}}_m\circ(I_\mathcal{S}\otimes M^{\mathcal{W}}_l).
\end{equation}
Then it follows from \eqref{eq:EDTE-Def-MW+_MW-} and \eqref{eq:WDc-MW_0_1_2}
that
\begin{equation}\label{eq:mopsMW+D_lm-cmpleq}
\begin{split}
  \sum_{(l,m)\in\Theta} {M^{\mathcal{W}+\mathcal{D}}_{l,m}}^\dag M^{\mathcal{W}+\mathcal{D}}_{l,m}
  &=\sum_{l=0,1,2}(I_\mathcal{S}\otimes {M^{\mathcal{W}}_l}^\dag)
  \left[\sum_{m=+,-}{M^{\mathcal{D}}_m}^\dag M^{\mathcal{D}}_m\right]
  (I_\mathcal{S}\otimes M^{\mathcal{W}}_l) \\
  &=I_\mathcal{S}\otimes \sum_{l=0,1,2}{M^{\mathcal{W}}_l}^\dag M^{\mathcal{W}}_l
  =I_{\mathcal{S}+\mathcal{F}}.
\end{split}
\end{equation}
Thus, the finite collection $\{M^{\mathcal{W}+\mathcal{D}}_{l,m}\}_{(l,m)\in\Theta}$ satisfies the \emph{completeness equation}, and therefore forms \emph{measurement operators}.
We define a unitary operator $U_{\mathrm{whole}}$ acting on $\mathcal{H}_\mathcal{S}\otimes\mathcal{H}_\mathcal{F}\otimes\mathcal{H}_\mathcal{W}\otimes\mathcal{H}_\mathcal{D}$
by
\begin{equation}\label{eq:WDcfmqsom-def-Uwhole}
  U_{\mathrm{whole}}:=U_{\mathcal{D}}\circ(I_\mathcal{S}\otimes U_{\mathcal{W}}\otimes I_\mathcal{D}).
\end{equation}
Then, using \eqref{eq:QSF-ED-WignerFriend-all} and \eqref{eq:QWD-POT-WF-Mkl=dlkMFk} we see that, as a whole,
the sequential applications of $U_{\mathcal{W}}$ and $U_{\mathcal{D}}$
to a state $\ket{\Psi}\otimes\ket{\Phi^{\mathcal{W}+\mathcal{D}}_{\mathrm{init}}}$
of the composite system $\mathcal{S}+\mathcal{F}+\mathcal{W}+\mathcal{D}$
result in the state:
\begin{equation}\label{eq:Wigner-Deutsch-Q-all}
\begin{split}
  U_{\mathrm{whole}}(\ket{\Psi}\otimes\ket{\Phi^{\mathcal{W}+\mathcal{D}}_{\mathrm{init}}})
  &=(U_{\mathcal{D}}\circ(I_\mathcal{S}\otimes U_{\mathcal{W}}\otimes I_\mathcal{D}))(\ket{\Psi}\otimes\ket{\Phi^{\mathcal{W}+\mathcal{D}}_{\mathrm{init}}}) \\
  &=\sum_{(l,m)\in\Theta}(M^{\mathcal{W}+\mathcal{D}}_{l,m}\ket{\Psi})\otimes\ket{\Phi^{\mathcal{W}+\mathcal{D}}[l,m]}
\end{split}
\end{equation}
for each state $\ket{\Psi}\in\mathcal{H}_\mathcal{S}\otimes\mathcal{H}_\mathcal{F}$
of the composite system $\mathcal{S}+\mathcal{F}$,
where $\ket{\Phi^{\mathcal{W}+\mathcal{D}}_{\mathrm{init}}}:=
\ket{\Phi^{\mathcal{W}}_{\mathrm{init}}}\otimes\ket{\Phi^{\mathcal{D}}_{\mathrm{init}}}$
and
$\ket{\Phi^{\mathcal{W}+\mathcal{D}}[l,m]}:=
\ket{\Phi^{\mathcal{W}}[l]}\otimes\ket{\Phi^{\mathcal{D}}[m]}$.

In this case,
the state of the composite system $\mathcal{S}+\mathcal{F}$
immediately after the `measurement' $\mathcal{M}^\mathcal{F}$ is given
by
$U_{\mathcal{F}}(\ket{+}\otimes\ket{\Phi^{\mathcal{F}}_\mathrm{init}})$,
as we stated above.
Thus, the unitary operator $U_{\mathrm{whole}}$ given by \eqref{eq:WDcfmqsom-def-Uwhole} applying to
the initial state
$$(U_{\mathcal{F}}(\ket{+}\otimes\ket{\Phi^{\mathcal{F}}_\mathrm{init}}))
\otimes\ket{\Phi^{\mathcal{W}+\mathcal{D}}_\mathrm{init}}$$
describes the \emph{repeated once} of the infinite repetition of the measurements
(experiment)
where the succession of the measurements $\mathcal{M}^\mathcal{W}$
and $\mathcal{M}^\mathcal{D}$ in this order is infinitely repeated.
It follows from \eqref{eq:Wigner-Deutsch-Q-all} that
the application of $U_{\mathrm{whole}}$
forms a \emph{single measurement} which is described by
the measurement operators $\{M^{\mathcal{W}+\mathcal{D}}_{l,m}\}_{(l,m)\in\Theta}$
and whose all possible outcomes form the set $\Theta$.

Hence, we can apply Definition~\ref{pmrpwst}
to this scenario of the setting of measurements.
Therefore, according to Definition~\ref{pmrpwst},
we can see that a \emph{world} is an infinite sequence over $\Theta$
and the probability measure induced by
the \emph{probability measure representation for the prefixes of worlds}
is a Bernoulli measure $\lambda_R$ on
$\Theta^\infty$,
where $R$ is a finite probability space on $\Theta$ such that
$R(l,m)$ is the square of the norm of the vector
\begin{equation*}
  (M^{\mathcal{W}+\mathcal{D}}_{l,m}(U_{\mathcal{F}}(\ket{+}\otimes\ket{\Phi^{\mathcal{F}}_\mathrm{init}})))
  \otimes\ket{\Phi^{\mathcal{W}+\mathcal{D}}[l,m]}
\end{equation*}
for every $(l,m)\in\Theta$.
Here $\Theta$ is the set of all possible records of
the observers $\mathcal{W}$ and $\mathcal{D}$
in the \emph{repeated once} of the measurements
(experiments).
Let us calculate the explicit form of $R(l,m)$.
First, it follows from \eqref{eq:Wigner-Deutsch-Q-all} and \eqref{eq:Wigner-Deutsch-all} that
\begin{equation*}
\begin{split}
  &\sum_{(l,m)\in\Theta}(M^{\mathcal{W}+\mathcal{D}}_{l,m}(U_{\mathcal{F}}(\ket{+}\otimes\ket{\Phi^{\mathcal{F}}_\mathrm{init}})))\otimes\ket{\Phi^{\mathcal{W}+\mathcal{D}}[l,m]} \\
  &=(U_{\mathcal{D}}\circ(I_\mathcal{S}\otimes U_{\mathcal{W}}\otimes I_\mathcal{D}))((U_{\mathcal{F}}(\ket{+}\otimes\ket{\Phi^{\mathcal{F}}_\mathrm{init}}))\otimes\ket{\Phi^{\mathcal{W}+\mathcal{D}}_{\mathrm{init}}}) \\
  &=(U_{\mathcal{D}}\circ(I_\mathcal{S}\otimes U_{\mathcal{W}}\otimes I_\mathcal{D})\circ(U_{\mathcal{F}}\otimes I_\mathcal{W}\otimes I_\mathcal{D}))
    (\ket{+}\otimes\ket{\Phi_{\mathrm{init}}})\\
  &=\sum_{(k,l,m)\in\Omega}(M_{k,l,m}\ket{+})\otimes\ket{\Phi[k,l,m]} \\
  &=\sum_{(l,m)\in\Theta}\left\{\sum_{k=0,1}(M_{k,l,m}\ket{+})\otimes\ket{\Phi^{\mathcal{F}}[k]}
  \right\}\otimes\ket{\Phi^{\mathcal{W}+\mathcal{D}}[l,m]}.
\end{split}
\end{equation*}
Thus, we have that
\begin{equation}\label{eq:WDc-components}
  M^{\mathcal{W}+\mathcal{D}}_{l,m}(U_{\mathcal{F}}(\ket{+}\otimes\ket{\Phi^{\mathcal{F}}_\mathrm{init}}))
  =\sum_{k=0,1}(M_{k,l,m}\ket{+})\otimes\ket{\Phi^{\mathcal{F}}[k]}
\end{equation}
for every $(l,m)\in\Theta$.
We use $\Theta_{c}$ to denote the set $\{0,1\}\times\{+,-\}$,
which is a proper subset of $\Theta$.
Then, on the one hand,
using~\eqref{eq:WDc-components} and \eqref{eq:EDTE-Mkl=0} we have that
\begin{equation*}
  (M^{\mathcal{W}+\mathcal{D}}_{l,m}(U_{\mathcal{F}}(\ket{+}\otimes\ket{\Phi^{\mathcal{F}}_\mathrm{init}})))
  \otimes\ket{\Phi^{\mathcal{W}+\mathcal{D}}[l,m]}=0
\end{equation*}
for every $(l,m)\in\Theta\setminus\Theta_c$.
Therefore, it follows that
\begin{equation}\label{eq:WDc-finitepsR=0}
  R(l,m)=0
\end{equation}
for every $(l,m)\in\Theta\setminus\Theta_c$.
On the other hand, for each $(l,m)\in\Theta_c$,
using~\eqref{eq:WDc-components}, \eqref{eq:Wigner-Deutsch_Mklm+Pklm=fklmoss2},
\eqref{eq:EDTE-Def-kPS+F+}, and \eqref{eq:EDTE-Def-kPS+F-}
we have that
\begin{equation}\label{eq:WDc-QF-resultant-state-l01}
\begin{split}
  &(M^{\mathcal{W}+\mathcal{D}}_{l,m}(U_{\mathcal{F}}(\ket{+}\otimes\ket{\Phi^{\mathcal{F}}_\mathrm{init}})))
  \otimes\ket{\Phi^{\mathcal{W}+\mathcal{D}}[l,m]} \\
  &=\sum_{k=0,1}(M_{k,l,m}\ket{+})\otimes\ket{\Phi[k,l,m]} \\
  &=\sum_{k=0,1}f(k,l,m)\ket{k}\otimes\ket{\Phi^{\mathcal{F}}[k]}\otimes\ket{\Phi^{\mathcal{W}}[l]}\otimes\ket{\Phi^{\mathcal{D}}[m]} \\
  &=c_l(\delta_{m,+}\overline{c_l}\ket{\Psi^{\mathcal{S}+\mathcal{F}}[+]}+\delta_{m,-}(-1)^{l}c_{1-l}\ket{\Psi^{\mathcal{S}+\mathcal{F}}[-]})\otimes\ket{\Phi^{\mathcal{W}}[l]}\otimes\ket{\Phi^{\mathcal{D}}[m]} \\
  &=g(l,m)\ket{\Psi^{\mathcal{S}+\mathcal{F}}[m]}\otimes\ket{\Phi^{\mathcal{W}}[l]}\otimes\ket{\Phi^{\mathcal{D}}[m]},
\end{split}
\end{equation}
where
\[
  g(l,m):=c_l(\delta_{m,+}\overline{c_l}+\delta_{m,-}(-1)^{l}c_{1-l}).
\]
Hence, using~\eqref{eq:WDc-QF-resultant-state-l01} we have that
\begin{equation}\label{eq:WDc-QF-finitepsP}
  R(l,m)
  =\abs{g(l,m)}^2\braket{\Psi^{\mathcal{S}+\mathcal{F}}[m]}{\Psi^{\mathcal{S}+\mathcal{F}}[m]}
  =\abs{c_l}^2(\delta_{m,+}\abs{c_l}^2+\delta_{m,-}\abs{c_{1-l}}^2)
\end{equation}
for each $(l,m)\in\Theta_c$.

Now, let us apply
Postulate~\ref{POT}, the \emph{principle of typicality},
to the setting of measurements developed above.
Let $\omega$ be \emph{our world} in the infinite repetition of the measurement described by the measurement operators $\{M^{\mathcal{W}+\mathcal{D}}_{l,m}\}_{(l,m)\in\Theta}$
in the above setting.
This $\omega$ is an infinite sequence over $\Theta$ consisting of records
in the two observers $\mathcal{W}$ and $\mathcal{D}$
which is being generated by the infinite repetition of
the measurement
described by the measurement operators $\{M^{\mathcal{W}+\mathcal{D}}_{l,m}\}_{(l,m)\in\Theta}$
in the above setting.
Since the Bernoulli measure $\lambda_R$ on $\Theta^\infty$ is
the probability measure induced by the
probability
measure representation
for the prefixes of
worlds
in the above setting,
it follows from
Postulate~\ref{POT}
that $\omega$ is Martin-L\"{o}f random with respect to the Bernoulli measure $\lambda_{R}$, that is, \emph{$\omega$ is Martin-L\"of $R$-random}.

We use $\beta$ and $\gamma$
to denote the infinite sequences over $\{0,1,2\}$ and $\{+,-\}$, respectively,
such that $(\beta(n),\gamma(n))=\omega(n)$ for every $n\in\N^+$.
Then, since $\omega$ is Martin-L\"of $R$-random, it follows from \eqref{eq:WDc-finitepsR=0} and
Corollary~\ref{cor:always-positive-probability}
that
\[
  \omega(n)\in\Theta_c
\]
for every $n\in\N^+$,
and therefore
\begin{equation}\label{eq:WDc-QF-bnneq2}
  \beta(n)\neq 2
\end{equation}
for every $n\in\N^+$, i.e., $\beta$ is an infinite binary sequence.
Thus, it follows from Postulate~\ref{POT} and \eqref{eq:Wigner-Deutsch-Q-all}
that, in our world $\omega$, for each $n\in\N^+$ the state of
the $n$th composite system $\mathcal{S}+\mathcal{F}+\mathcal{W}+\mathcal{D}$,
i.e., the $n$th copy of a composite system consisting of
the system $\mathcal{S}$ and the three observers $\mathcal{F}$, $\mathcal{W}$, and $\mathcal{D}$,
immediately after the measurement $\mathcal{M}^\mathcal{D}$ is given by
\begin{equation}\label{eq:WDc-QF-resultant-state}
\begin{split}
  &(M^{\mathcal{W}+\mathcal{D}}_{\beta(n),\gamma(n)}(U_{\mathcal{F}}(\ket{+}\otimes\ket{\Phi^{\mathcal{F}}_\mathrm{init}})))
  \otimes\ket{\Phi^{\mathcal{W}+\mathcal{D}}[\beta(n),\gamma(n)]} \\
  &=g(\beta(n),\gamma(n))\ket{\Psi^{\mathcal{S}+\mathcal{F}}[\gamma(n)]}\otimes\ket{\Phi^{\mathcal{W}}[\beta(n)]}\otimes\ket{\Phi^{\mathcal{D}}[\gamma(n)]},
\end{split}
\end{equation}
up to the normalization factor,
where the equality follows from \eqref{eq:WDc-QF-resultant-state-l01}.
Since $\omega$ is Martin-L\"of $R$-random,
it follows from Theorem~\ref{FI} and \eqref{eq:WDc-QF-finitepsP} that
for every $(l,m)\in\Theta_c$ it holds that
\begin{equation}\label{eq:WDc-QF-imafterW-by-PoT-LLN}
  \lim_{n\to\infty} \frac{N_{l,m}(\rest{\omega}{n})}{n}
  =\abs{c_l}^2(\delta_{m,+}\abs{c_l}^2+\delta_{m,-}\abs{c_{1-l}}^2),
\end{equation}
where $N_{l,m}(\rest{\omega}{n})$ denotes the number of the occurrences of $(l,m)$
in the prefix of $\omega$ of length $n$.
Thus, it follows from \eqref{eq:WDc-QF-resultant-state} and \eqref{eq:WDc-QF-imafterW-by-PoT-LLN} that, in our world $\omega$, \emph{the following holds for each $(l,m)\in\Theta_c$:
In a proportion of $$\abs{c_l}^2(\delta_{m,+}\abs{c_l}^2+\delta_{m,-}\abs{c_{1-l}}^2)$$ out of the infinite repetitions of the experiment,
the state of the composite system $\mathcal{S}+\mathcal{F}$ is
$\ket{\Psi^{\mathcal{S}+\mathcal{F}}[m]}$ and
the observers $\mathcal{W}$ and $\mathcal{D}$ record
the values $l$ and $m$, respectively,
immediately after the measurement $\mathcal{M}^\mathcal{D}$}.
Compared to \eqref{eq:WDc-Pc2lm}, we see that
\emph{this situation coincides with that of Case~2
in Section~\ref{sec:Wigner collaborates with Deutsch}
immediately after the measurement $\mathcal{M}^\mathcal{D}$},
although the argument regarding Case~2 given in Section~\ref{sec:Wigner collaborates with Deutsch} is heuristic
and the statements of
Case~2 in Section~\ref{sec:Wigner collaborates with Deutsch} are
at least operationally vague
since they are described in terms of the conventional quantum mechanics
where the operational characterization of the notion of probability is not given.

Let us determine the states of
the composite system $\mathcal{S}+\mathcal{F}+\mathcal{W}$
immediately after the measurement $\mathcal{M}^\mathcal{W}$
in our world $\omega$.
For that purpose, we can apply our general framework
to treat situations
where apparatuses perform measurements over other apparatuses as well as over a normal quantum system only being measured,
developed in Section~\ref{sec:confirming point},
to the setting of measurements in the Wigner-Deutsch collaboration
where the observer $\mathcal{F}$ is regarded as a mere quantum system and not as a measurement apparatus,
developed above.
The current setting of measurements in the Wigner-Deutsch collaboration is treated within the general framework with $n=2$,
and the measurements $\mathcal{M}^\mathcal{W}$ and $\mathcal{M}^\mathcal{D}$ in the current setting serve as the measurements $\mathcal{M}^1$ and $\mathcal{M}^2$ in the general framework, respectively.
The composite system $\mathcal{S}+\mathcal{F}$, and the observers $\mathcal{W}$ and $\mathcal{D}$ in the current setting serve as the system $\mathcal{S}$, and the apparatuses $\mathcal{A}_1$ and $\mathcal{A}_2$ in the general framework, respectively.
Then the measurement operators
$\{M^1_{m_1}\}_{m_1\in\Omega_1}$ and $\{M^2_{m_2}\}_{m_2\in\Omega_2}$
in the general framework are instantiated in the current setting
in the following manner:
\begin{enumerate}
  \item $\Omega_1=\{0,1,2\}$ and $\Omega_2=\{+,-\}$.
  \item $M^1_{m_1}=I_\mathcal{S}\otimes M^{\mathcal{W}}_{m_1}$ for every $m_1\in\Omega_1$, and
    $M^2_{m_2}=M^{\mathcal{D}}_{m_2}\otimes I_\mathcal{W}$ for every $m_2\in\Omega_2$.
\end{enumerate}
Therefore,
based on \eqref{eq:EDTE-Def-MW+_MW-},
it is easy to check that the condition~\eqref{eq:Ui-TSCP-MO} certainly holds
in the current setting of measurements in the Wigner-Deutsch collaboration
where the observer $\mathcal{F}$ is regarded as a mere quantum system and not as a measurement apparatus.
Thus, based on \eqref{eq:Wigner-Deutsch-Q-all}, we see that
the measurement operators $\{M^{1,2}_{m_1,m_2}\}_{(m_1,m_2)\in\Omega_1\times\Omega_2}$
satisfying \eqref{eq:Ui-TSCP-1tok} with $n=k=2$ in the general framework are instantiated in the current setting
in the manner that
\begin{equation*}
  M^{1,2}_{m_1,m_2}=M^{\mathcal{W}+\mathcal{D}}_{m_1,m_2}
\end{equation*}
for every $m_1\in\Omega_1$ and $m_2\in\Omega_2$,
where the measurement operators $\{M^{\mathcal{W}+\mathcal{D}}_{l,m}\}_{(l,m)\in\Theta}$ on the right-hand side are defined by \eqref{eq:QWD-POT-WF-Mkl=dlkMFk}.
The state $\ket{\Psi^0_{\mathrm{init}}}$ of the system $\mathcal{S}$
immediately before the measurement $\mathcal{M}^1$ in the general framework
materializes as
the
state $U_{\mathcal{F}}(\ket{+}\otimes\ket{\Phi^{\mathcal{F}}_\mathrm{init}})$ of the composite system $\mathcal{S}+\mathcal{F}$ in the current setting.

First, we use Postulate~\ref{CF} to determine the state of the observer $\mathcal{W}$
immediately after the measurement $\mathcal{M}^\mathcal{W}$
in our world $\omega$.
According to Definition~\ref{def:unchanged-apparatus},
it is easy to check that 
the final states of the observer $\mathcal{W}$ are
unchanged after the measurement by
the observer $\mathcal{D}$
and therefore the final states of the observer $\mathcal{W}$ are confirmed
before the measurement by the observer $\mathcal{D}$,
where ``apparatus'' in Definition~\ref{def:unchanged-apparatus} should read ``observer''.
Using \eqref{eq:WDc-QF-resultant-state}, we see that, in our world $\omega$,
for every $n\in\N^+$ the state of the observer $\mathcal{W}$
immediately after the measurement $\mathcal{M}^\mathcal{D}$
is $\ket{\Phi^{\mathcal{W}}[\beta(n)]}$ in the $n$th composite system $\mathcal{S}+\mathcal{F}+\mathcal{W}+\mathcal{D}$.
Thus, it follows from Postulate~\ref{CF}~(ii) that, in our world $\omega$,
for every $n\in\N^+$ the state of the observer $\mathcal{W}$
immediately after the measurement $\mathcal{M}^\mathcal{W}$ is
\begin{equation*}%
  \ket{\Phi^{\mathcal{W}}[\beta(n)]}
\end{equation*}
in the $n$th composite system $\mathcal{S}+\mathcal{F}+\mathcal{W}+\mathcal{D}$.

We can use Postulate~\ref{Recursive Use} to determine the \emph{whole} state of
the composite system $\mathcal{S}+\mathcal{F}+\mathcal{W}$
immediately after the measurement $\mathcal{M}^\mathcal{W}$
in our world $\omega$,
instead of using Postulate~\ref{CF}.
First note that
each of the measurement operators
$\{I_\mathcal{S}\otimes M^{\mathcal{W}}_0,I_\mathcal{S}\otimes M^{\mathcal{W}}_1,I_\mathcal{S}\otimes M^{\mathcal{W}}_2\}$ and
$\{M^{\mathcal{D}}_+\otimes I_\mathcal{W},M^{\mathcal{D}}_-\otimes I_\mathcal{W}\}$ forms a PVM.
On the one hand,
applying $I_\mathcal{S}\otimes U_{\mathcal{W}}\otimes I_\mathcal{D}$
to the initial state
$(U_{\mathcal{F}}(\ket{+}\otimes\ket{\Phi^{\mathcal{F}}_\mathrm{init}}))
\otimes\ket{\Phi^{\mathcal{W}}_{\mathrm{init}}}\otimes\ket{\Phi^{\mathcal{D}}_{\mathrm{init}}}$
of the repeated once of the experiments
results in a state $\ket{\Psi^{\mathcal{W}(\mathcal{S}+\mathcal{F})}_{\mathrm{Total}}}$ given by
\begin{equation}\label{eq:WDc-QF-SPTS-immafterW}
\begin{split}
  \ket{\Psi^{\mathcal{W}(\mathcal{S}+\mathcal{F})}_{\mathrm{Total}}}
  &:=(I_\mathcal{S}\otimes U_{\mathcal{W}}\otimes I_\mathcal{D})((U_{\mathcal{F}}(\ket{+}\otimes\ket{\Phi^{\mathcal{F}}_\mathrm{init}}))\otimes\ket{\Phi^{\mathcal{W}}_{\mathrm{init}}}\otimes\ket{\Phi^{\mathcal{D}}_{\mathrm{init}}}) \\
  &=((I_\mathcal{S}\otimes U_{\mathcal{W}}\otimes I_\mathcal{D})\circ(U_{\mathcal{F}}\otimes I_\mathcal{W}\otimes I_\mathcal{D}))(\ket{+}\otimes\ket{\Phi^{\mathcal{F}}_{\mathrm{init}}}\otimes\ket{\Phi^{\mathcal{W}}_{\mathrm{init}}}\otimes\ket{\Phi^{\mathcal{D}}_{\mathrm{init}}}) \\
  &=\sum_{k=0,1}c_k\ket{k}\otimes\ket{\Phi^{\mathcal{F}}[k]}\otimes\ket{\Phi^{\mathcal{W}}[k]}\otimes\ket{\Phi^{\mathcal{D}}_{\mathrm{init}}},
\end{split}
\end{equation}
where the last equality follows from \eqref{eq:WDc-virtual-state-Total-W}.
On the other hand,
the final states of the observer $\mathcal{W}$ are unchanged after the measurement by
the observer $\mathcal{D}$, as we saw above.
Thus, it follows from Postulate~\ref{Recursive Use} and
\eqref{eq:WDc-QF-resultant-state} that, in our world $\omega$,
for each $n\in\N^+$ the state of
the $n$th composite system $\mathcal{S}+\mathcal{F}+\mathcal{W}+\mathcal{D}$
immediately after the measurement $\mathcal{M}^\mathcal{W}$ is given by
\begin{equation}\label{eq:WDc-QF-Proj-virtual-state-Total-W}
  P\left(\ket{\Phi^{\mathcal{W}}[\beta(n)]}\right)\ket{\Psi^{\mathcal{W}(\mathcal{S}+\mathcal{F})}_{\mathrm{Total}}},
\end{equation}
up to the normalization factor, where
\[
  P\left(\ket{\Phi^{\mathcal{W}}[\beta(n)]}\right)
  =I_{\mathcal{S}}\otimes I_{\mathcal{F}}\otimes
  \product{\Phi^{\mathcal{W}}[\beta(n)]}{\Phi^{\mathcal{W}}[\beta(n)]}\otimes I_{\mathcal{D}}.
\]
The vector \eqref{eq:WDc-QF-Proj-virtual-state-Total-W} equals
\[
  c_{\beta(n)}\ket{\beta(n)}\otimes\ket{\Phi^{\mathcal{F}}[\beta(n)]}\otimes\ket{\Phi^{\mathcal{W}}[\beta(n)]}\otimes\ket{\Phi^{\mathcal{D}}_{\mathrm{init}}}
\]
for every $n\in\N^+$,
due to \eqref{eq:WDc-QF-SPTS-immafterW}.
Hence, in our world $\omega$, for every $n\in\N^+$ 
the state of the composite system $\mathcal{S}+\mathcal{F}+\mathcal{W}$
immediately after the measurement $\mathcal{M}^\mathcal{W}$ is given by
\begin{equation}\label{eq:WDc-QF-SFW-state-immafter-MW}
  \ket{\beta(n)}\otimes\ket{\Phi^{\mathcal{F}}[\beta(n)]}\otimes\ket{\Phi^{\mathcal{W}}[\beta(n)]}
\end{equation}
in the $n$th composite system $\mathcal{S}+\mathcal{F}+\mathcal{W}+\mathcal{D}$.

In addition,
since $\omega$ is Martin-L\"of $R$-random and $\abs{c_0}^2+\abs{c_1}^2=1$,
it follows from Theorem~\ref{contraction2}, \eqref{eq:WDc-finitepsR=0}, \eqref{eq:WDc-QF-finitepsP}, and \eqref{eq:WDc-QF-bnneq2} that
$\beta$ is a Martin-L\"of $B$-random infinite binary sequence,
where $B$ is a finite probability space on $\{0,1\}$ such that
$B(0)=\abs{c_0}^2$ and $B(1)=\abs{c_1}^2$.
Therefore,
it follows from Theorem~\ref{FI} that for every $l\in\{0,1\}$ it holds that
\[
  \lim_{n\to\infty} \frac{N_l(\rest{\beta}{n})}{n}
  =\abs{c_l}^2,
\]
where $N_l(\rest{\beta}{n})$ denotes the number of the occurrences of $l$
in the prefix of $\beta$ of length $n$.
Thus, the form of the state \eqref{eq:WDc-QF-SFW-state-immafter-MW} implies that,
in our world $\omega$,
\emph{for each $l=0,1$
the following holds
in a proportion of $\abs{c_l}^2$ out of the infinite repetitions of the experiment:
The state of the system $\mathcal{S}$ is $\ket{l}$ and
both
the observers $\mathcal{F}$ and $\mathcal{W}$ record
this value $l$,
immediately after the measurement $\mathcal{M}^\mathcal{W}$}.

In this section, within our rigorous framework of quantum mechanics based on the principle of typicality
we have investigated
the Wigner-Deutsch collaboration
in the case where
the observer $\mathcal{F}$ is treated as a mere quantum system
and not as a measurement apparatus.
Based on the above analysis, we can see that
the results of the analysis of this case recover in a \emph{refined manner}
Case~2 in Section~\ref{sec:Wigner collaborates with Deutsch},
which is described and heuristically analyzed in terms of the conventional quantum mechanics in a vague manner in Section~\ref{sec:Wigner collaborates with Deutsch}.
Note that, as far as the measurements $\mathcal{M}^\mathcal{W}$ and $\mathcal{M}^\mathcal{D}$ are concerned,
the results of the analysis of this case also coincide with Case~1 in Section~\ref{sec:Wigner collaborates with Deutsch},
which is also described and heuristically analyzed in terms of the conventional quantum mechanics in a vague manner in Section~\ref{sec:Wigner collaborates with Deutsch}.

\section{Conclusion}
\label{sec-Concluding_remarks}

In this paper, we have extended
the rigorous framework of quantum mechanics based on the principle of typicality,
which was introduced by Tadaki~\cite{T16CCR,T16QIP,T16QIT35,T17SCIS,T18arXiv},
so that it can be applicable to
a situation where apparatuses perform measurements
over other apparatuses as well as over
a normal quantum system only being measured.
To be specific, we have introduced
the notion of the confirming point for the states of a system
and the notion of the confirming point for the final states of an apparatus.
Then, based on these notions,
we have introduced two postulates,
Postulate~\ref{CF} and Postulate~\ref{Recursive Use},
so that we are able to determine the states of the apparatuses
in the intermediate stages
among all apparatuses which constitute a whole measurement process.

Within this extended framework of quantum mechanics based on the principle of typicality,
we then have thoroughly performed analyses of the three intricate situations:
the Wigner's friend paradox,
Deutsch's thought experiment, and their combination,
named the Wigner-Deutsch collaboration.
Consequently, for each of the three situations,
we have been able to ``cut into'' the inside of the situation and have
clarified the states of the observers in the intermediate stages
in the chain of the measurements by multiple observers
where the state of the consciousness of the preceding observers are measured by the subsequent observers.

\section*{Acknowledgments}
\addcontentsline{toc}{section}{Acknowledgments}

This work was
supported by JSPS KAKENHI Grant Number~22K03409.

\end{document}